\documentclass[11pt]{article}
\usepackage[margin=1in]{geometry}
\usepackage[latin2]{inputenc}
\usepackage[english]{babel}
\usepackage{amsmath}
\usepackage{todonotes}
\usepackage{amsthm}
\usepackage{amssymb}
\usepackage{microtype}
\usepackage{xspace}
\usepackage{todonotes}
\usepackage{comment}
\usepackage{letltxmacro}
\usepackage{comment}
\usepackage{thmtools}
\usepackage{thm-restate}
\usepackage{mdframed}
\usepackage{enumitem}
\usepackage{color}
\usepackage{tikz}
\usepackage{algorithm2e}
\usetikzlibrary{external}
\usetikzlibrary{calc}
\usetikzlibrary{decorations.markings}
\usepackage{etoolbox}
\tikzexternaldisable
\tikzsetexternalprefix{fig-comp/}

\usepackage{url}
\usepackage{eufrak}
\setlist{noitemsep}

\usepackage{subfig}

\newif\ifabstract
\abstracttrue

\abstractfalse 

\newif\iffull
\ifabstract 
\usepackage[nomarkers,nolists]{endfloat}
\usepackage{mathptmx}
\DeclareMathAlphabet{\mathcal}{OMS}{cmsy}{m}{n}

\fullfalse \else \fulltrue  \fi 

\newtheorem{theorem}{Theorem}[section]
\newtheorem{lemma}[theorem]{Lemma}
\newtheorem{claim}[theorem]{Claim}
\newtheorem{corollary}[theorem]{Corollary}

\theoremstyle{definition}

\newcommand{\retheorem}{theorem}

\makeatletter
\def\th@plain{%
  \thm@notefont{}% same as heading font
  \itshape % body font
}
\def\th@definition{%
  \thm@notefont{}% same as heading font
  \normalfont % body font
}
\makeatother

\newcommand{\myparagraph}[1]{\textbf{#1}}

\newcommand{\dist}{\mathrm{dist}}
\newcommand{\Oh}{\ensuremath{\mathcal{O}}}

\newcommand{\R}{\ensuremath{\mathbb{R}}}
\newcommand{\N}{\ensuremath{\mathbb{N}}}

\def\cqedsymbol{\ifmmode$\lrcorner$\else{\unskip\nobreak\hfil
\penalty50\hskip1em\null\nobreak\hfil$\lrcorner$
\parfillskip=0pt\finalhyphendemerits=0\endgraf}\fi} 

\newcommand{\cqed}{\renewcommand{\qed}{\cqedsymbol}}

\graphicspath{{figures/}}

\newcommand{\executeiffilenewer}[3]{%
\ifnum\pdfstrcmp{\pdffilemoddate{#1}}%
{\pdffilemoddate{#2}}>0%
{\immediate\write18{#3}}\fi%
} 

\newcommand{\svg}[2]{\def\svgwidth{#1}%
 \executeiffilenewer{figures/#2.svg}{figures/#2.pdf}%
{inkscape -z -D --file=figures/#2.svg %
--export-pdf=figures/#2.pdf --export-latex}%
{\input{figures/#2.pdf_tex}}}

\newcommand{\rectfam}{\mathcal{R}}
\newcommand{\pointfam}{\mathcal{P}}

\newcommand{\probIOutB}{\textsc{$k$-Internal Out-Branching}\xspace}
\newcommand{\probLOutB}{\textsc{$k$-Leaf Out-Branching}\xspace}
\newcommand{\probkdCent}{\textsc{$(k,d)$-Center}\xspace}
\newcommand{\probPVC}{\textsc{Partial Vertex Cover}\xspace}

\newcommand{\cnfsat}{\textsc{CNF-SAT}\xspace}
\newcommand{\gridtiling}{\textsc{Grid Tiling}\xspace}
\newcommand{\probClique}{\textsc{Clique}\xspace}
\newcommand{\probDS}{\textsc{Dominating Set}\xspace}
\newcommand{\probScatter}{\textsc{$d$-Scattered Set}\xspace}
\newcommand{\probdDS}{\textsc{$d$-Dominating Set}\xspace}
\newcommand{\probIS}{\textsc{Independent Set}\xspace}
\newcommand{\probFVS}{\textsc{Feedback Vertex Set}\xspace}
\newcommand{\probKPath}{\textsc{Longest Path}\xspace}
\newcommand{\probMWay}{\textsc{Multiway Cut}\xspace}
\newcommand{\probSTSP}{\textsc{Subset TSP}\xspace}
\newcommand{\probTSP}{\textsc{TSP}\xspace}
\newcommand{\probSCSS}{\textsc{Strongly Connected Steiner Subgraph}\xspace}
\newcommand{\probSteiner}{\textsc{Steiner Tree}\xspace}

\newcommand{\covProb}{{\textsc{Disjoint Network Coverage}}\xspace}
\newcommand{\partbic}{{\textsc{Partitioned Biclique}}\xspace}
\newcommand{\partsub}{{\textsc{Partitioned Subgraph Isomorphism}}\xspace}
\newcommand{\wei}{{\bf w}}
\newcommand{\Fam}{{\mathcal{F}}}
\newcommand{\SolObj}{{\mathcal{Z}}}
\newcommand{\Obj}{{\mathcal{D}}}
\newcommand{\Cli}{{\mathcal{C}}}
\newcommand{\onum}{d}
\newcommand{\cnum}{c}
\newcommand{\pla}{\mathbf{pla}}
\newcommand{\pri}{\pi}
\newcommand{\Pri}{\Pi}
\newcommand{\cst}{\lambda}

\newcommand{\sen}{s}
\newcommand{\loc}{\mathbf{loc}}
\newcommand{\rad}{r}
\newcommand{\cen}{\mathbf{cen}}
\newcommand{\iso}{\eta}
\newcommand{\noose}{\vec{\delta}}
\newcommand{\md}{\text{mid}}
\newcommand{\cn}{\mathbf{cen}}
\newcommand{\bw}{\mathbf{bw}}

\newcommand{\tree}{T}
\newcommand{\sep}{S}
\newcommand{\walk}{W}
\newcommand{\Val}{\mathbf{Val}}

\newcommand{\Rad}{\mathbf{Rad}}

\newcommand{\Med}{\mathbf{Med}}
\newcommand{\extree}{\widehat{T}}
\newcommand{\Vorpar}{\mathbb{M}}
\newcommand{\dgm}{\widetilde{\mathcal{H}}}
\newcommand{\pdgm}{\mathcal{H}}
\newcommand{\Rp}{\mathbb{R}^{+}}
\newcommand{\minf}{-\infty}
\newcommand{\pinf}{+\infty}
\newcommand{\ass}{$(\clubsuit)$\xspace}
\newcommand{\crv}{\vec{\gamma}}
\newcommand{\enc}{\mathbf{enc}}
\newcommand{\exc}{\mathbf{exc}}
\newcommand{\bnd}{\partial}

\newcommand{\impFaces}{\mathcal{E}}

\newcommand{\Vorsepfam}{\mathcal{N}}

\newcommand{\Ii}{\mathcal{I}}

  \date{}

  \author{
  D\'{a}niel Marx\thanks{
	Institute for Computer Science and Control, Hungarian Academy of Sciences (MTA SZTAKI)
	\texttt{dmarx@cs.bme.hu}.
  }
  \and
  Micha\l{} Pilipczuk\thanks{
    Institute of Informatics, University of Warsaw, Poland, 
    \texttt{michal.pilipczuk@mimuw.edu.pl}
  }
  }

\title{Optimal parameterized algorithms for planar facility location problems using 
Voronoi diagrams\thanks{The research leading to these results has received funding from the European Research Council under the European Union's Seventh Framework Programme (FP/2007-2013) / ERC Grant Agreements n.~280152 (D.~Marx) and n.~267959 (M.~Pilipczuk, while he was affiliated with the University of Bergen, Norway). The further work of M.~Pilipczuk on this research at the University of Warsaw is supported by Polish National Science Centre grant DEC-2013/11/D/ST6/03073; at this time, M.~Pilipczuk held also a post-doc position at Warsaw Center of Mathematics and Computer Science. The work of D.~Marx is also supported by Hungarian Scientific Research Fund (OTKA) grant NK105645.}}

\begin{document}

%\thispagestyle{empty}
%\begin{titlepage}
%\def\thepage{}

\maketitle

\begin{abstract}
We study a general family of facility location problems defined on
planar graphs and on the 2-dimensional plane. In these problems, a
subset of $k$ objects has to be selected, satisfying certain packing
(disjointness) and covering constraints. Our main result is showing
that, for each of these problems, the $n^{\Oh(k)}$ time brute force
algorithm of selecting $k$ objects can be improved to
$n^{\Oh(\sqrt{k})}$ time.  The algorithm is based on an idea that was
introduced recently in the design of geometric QPTASs, but was not yet
used for exact algorithms and for planar graphs.  We focus on the
Voronoi diagram of a hypothetical solution of $k$ objects, guess a
balanced separator cycle of this Voronoi diagram to obtain a set that
separates the solution in a balanced way, and then recurse on the
resulting subproblems.

The following list is an exemplary selection of concrete consequences of our main result.
We can solve each of the following problems in time $n^{\Oh(\sqrt{k})}$, where $n$ is the total size of the input: 
\begin{itemize}
\item \probScatter: find $k$ vertices in an edge-weighted planar graph that pairwise are at distance at least $d$ from each other ($d$ is part of the input).
\item \probdDS (or \probkdCent): find $k$ vertices in an edge-weighted planar graph such that every vertex of the graph is at distance at most $d$ from at least one selected vertex ($d$ is part of the input).
\item Given a set $\Obj$ of connected vertex sets in a planar graph $G$, find $k$ disjoint vertex sets in $\Obj$.
\item Given a set $\Obj$ of disks in the plane (of possibly different radii), find $k$ disjoint disks in $\Obj$.
\item Given a set $\Obj$ of simple polygons in the plane, find $k$ disjoint polygons in $\Obj$.
\item Given a set $\Obj$ of disks in the plane (of possibly different radii) and a set $\pointfam$ of points, find $k$ disks in $\Obj$ that together cover the maximum number of points in $\pointfam$.
\item Given a set $\Obj$ of axis-parallel squares in the plane (of possibly different sizes) and a set $\pointfam$ of points, find $k$ squares in $\Obj$ that together cover the maximum number of points in $\pointfam$.
\end{itemize}
It is known from previous work that, assuming the Exponential Time
Hypothesis (ETH), there is no $f(k)n^{o(\sqrt{k})}$ time algorithm for
any computable function $f$ for any of these problems. Furthermore, we
give evidence that packing problems have $n^{\Oh(\sqrt{k})}$ time
algorithms for a much more general class of objects than covering
problems have. For example, we show that, assuming ETH, the problem
where a set $\Obj$ of axis-parallel rectangles and a set $\pointfam$
of points are given, and the task is to select $k$ rectangles that
together cover the entire point set, does not admit an $f(k)n^{o(k)}$
time algorithm for any computable function $f$.

%%% Local Variables: 
%%% mode: latex
%%% TeX-master: "voronoi"
%%% End: 

\end{abstract}

%\end{titlepage}

%\tableofcontents
\clearpage

\section{Introduction}

Parameterized problems often become easier when restricted to planar
graphs: usually significantly better running times can be achieved and
sometimes even problems that are W[1]-hard on general graphs can
become fixed-parameter tractable on planar graphs. In most cases, the
improved running time involves a square root dependence on the
parameter: for example, it is often of the form $2^{\Oh(\sqrt{k})}\cdot
n^{\Oh(1)}$ or $n^{\Oh(\sqrt{k})}$. The appearance of the square root
can be usually traced back to the fact that a planar graph with $n$
vertices has treewidth $\Oh(\sqrt{n})$. Indeed, the theory of
bidimensionality gives a quick explanation why problems such as
\probIS, \probKPath, \probFVS, \probDS, or even distance-$r$ versions
of \probIS and \probDS (for fixed $r$) have algorithms with running
time $2^{ \Oh(\sqrt{k})}\cdot n^{\Oh(1)}$
\cite{DBLP:journals/siamdm/DemaineFHT04,DBLP:journals/talg/DemaineFHT05,DBLP:journals/jacm/DemaineFHT05,DBLP:journals/cj/DemaineH08,DBLP:journals/combinatorica/DemaineH08,DBLP:conf/gd/DemaineH04,DBLP:journals/csr/DornFT08,DornPBF10,DBLP:journals/siamcomp/FominT06,DBLP:conf/esa/Thilikos11}. In
all these problems, there is a relation between the size of the
largest grid minor and the size of the optimum solution, which allows
us to bound the treewidth of the graph in terms of the parameter of
the problem. More recently, subexponential parameterized algorithms
have been explored also for problems where there is no such
straightforward parameter-treewidth bound: for example, for \probPVC
and \probDS \cite{DBLP:journals/ipl/FominLRS11}, \probIOutB and
\probLOutB \cite{DBLP:conf/stacs/DornFLRS10}, \probMWay
\cite{DBLP:conf/icalp/KleinM12}, \probSTSP
\cite{DBLP:conf/soda/KleinM14}, \probSCSS
\cite{DBLP:conf/soda/ChitnisHM14}, \probSteiner
\cite{DBLP:conf/stacs/PilipczukPSL13,DBLP:conf/focs/PilipczukPSL14}. For
some of these problems, it is easy to see that they are
fixed-parameter tractable on planar graphs, and the challenge is to
make the dependence on $k$ subexponential, e.g., to obtain
$2^{\Oh(\sqrt{k})}\cdot n^{\Oh(1)}$ time (or perhaps $2^{\Oh(\sqrt{k}\log
  k)}\cdot n^{\Oh(1)}$ time) algorithms. Others are W[1]-hard on planar
graphs, and then the challenge is to improve the known $n^{\Oh(k)}$ time
algorithm to $n^{\Oh(\sqrt{k})}$ time. For all these problems, there are
matching lower bounds showing that, assuming the Exponential Time
Hypothesis (ETH) of Impagliazzo, Paturi, and Zane~\cite{MR1894519,DBLP:journals/eatcs/LokshtanovMS11},
there are no $2^{o(\sqrt{k})}\cdot n^{\Oh(1)}$ time algorithms (for FPT
problems) or $n^{o(\sqrt{k})}$ time algorithms (for W[1]-hard
problems).

A similar ``square root phenomenon'' has been observed in the case of
geometric problems: it is usual to see a square root in the exponent
of the running time of algorithms for NP-hard problems defined in the
2-dimensional Euclidean plane. For example, \probTSP and \probSteiner
on $n$ points can be solved in time $2^{ \Oh(\sqrt{n}\log n)}$
\cite{DBLP:conf/focs/SmithW98}. Most relevant to our paper is the fact
that \probIS for unit disks (given a set of $n$ unit disks,
select $k$ of them that are pairwise disjoint) and the discrete
$k$-center problem (given a set of $n$ points and a set of
$n$ unit disks, select $k$ disks whose union covers every point) can
be solved in time $n^{\Oh(\sqrt{k})}$ by geometric separation theorems
and shifting arguments
\cite{DBLP:journals/siamcomp/AgarwalOS06,DBLP:journals/algorithmica/AgarwalP02,DBLP:journals/jal/AlberF04,DBLP:journals/algorithmica/HwangLC93,
  DBLP:conf/compgeom/MarxS14}, improving on the trivial $n^{\Oh(k)}$
time brute force algorithm. However, all of these algorithms are
crucially based on a notion of area and rely on the property that all
the disks have the same size (at least approximately). Therefore, it
seems unlikely that these techniques can be generalized to the case
when the disks can have very different radii or to planar-graph
versions of the problem, where the notion of area is
meaningless. Using similar techniques, one can obtain approximation
schemes for these and related geometric problems, again with the
limitation that the objects need to have (roughly) the same area.\iffull

\fi Very recently, a new and powerful technique emerged from a line of
quasi-polynomial time approximation schemes (QPTAS) for geometric
problems
\cite{DBLP:conf/focs/AdamaszekW13,DBLP:conf/soda/AdamaszekW14,
  DBLP:conf/compgeom/Har-Peled14, DBLP:journals/corr/MustafaRR14}. As
described explicitly by Har-Peled
\cite{DBLP:conf/compgeom/Har-Peled14}, the main idea is to reason
about the Voronoi diagram of the $k$ objects in the solution. In
particular, we are trying to guess a separator consisting of
$\Oh(\sqrt{k})$ segments that corresponds to a balanced separator of the
Voronoi diagram. In this paper, we show how this basic idea and its
extensions can be implemented to obtain $n^{\Oh(\sqrt{k})}$ time exact
algorithms for a wide family of geometric packing and covering
problems in a uniform way.  In fact, we show that the algorithms can be made to work in the much more
general context of planar graph problems.

\myparagraph{Algorithmic results.} We study a general family of
facility location problems for planar graphs, where a set of $k$
objects has to be selected, subject to certain independence and
covering constraints. Two archetypal problems from this family are (1)
selecting $k$ vertices of an edge-weighted planar graph that are at
distance at least $d$ from each other (\probScatter) and (2) selecting
$k$ vertices of an edge-weighted planar graph such that every vertex
of the graph is at distance at most $d$ from some selected vertex
(\probdDS); for both problems, $d$ is a real value being part of the
input. We show that, under very general conditions, the trivial
$n^{\Oh(k)}$ time brute force algorithm can be improved to
$n^{\Oh(\sqrt{k})}$ time for problems in this family.  Our result is
not just a simple consequence of bidimensionality and bounding the
treewidth of the input graph. Instead, we focus on the Voronoi diagram
of a hypothetical solution, which can be considered as a planar graph
with $\Oh(k)$ vertices. It is known that such a planar graph has a
balanced separator cycle of length $\Oh(\sqrt{k})$, which can be
translated into a separator that breaks the instance in way suitable
for using recursion on the resulting subproblems: each subproblem
contains at most a constant fraction of the $k$ objects of the
solution.  Of course, we do not know the Voronoi diagram of the
solution and its balanced separator cycle.  However, we argue that
only $n^{\Oh(\sqrt{k})}$ separator cycles can be potential
candidates. Thus, by guessing one of these cycles, we define and solve
$n^{\Oh(\sqrt{k})}$ subproblems. The reason why this scheme yields an $n^{\Oh(\sqrt{k})}$ time algorithm is the fact that recurrence relations of the form $f(k)=n^{\Oh(\sqrt{k})}f(k/2)$ resolves  to $f(k)=n^{\Oh(\sqrt{k})}$.

In Section~\ref{sec:problem-definition}, we define a general facility
location problem called \covProb, which contains numerous concrete
problems of interest as special cases. \iffull We defer to Section~\ref{sec:problem-definition} the formal definition
of the problem and the exact statement of the running time we can
achieve.\fi\ In this introduction, we discuss specific algorithmic results following from the general result.

Informally, the input of \covProb consists of an edge-weighted planar graph $G$, a
set $\Obj$ of objects (which are connected sets of vertices in $G$) and
a set $\Cli$ of clients (which are vertices of $G$). The task is to
select a set of exactly $k$ pairwise-disjoint\footnote{More preceisely, 
if the objects have different radii, then instead of requiring them to be pairwise disjoint, we require a technical condition called ``normality,'' which we define in Section~\ref{sec:problem-definition}.}
 objects that maximizes
the total number (or total prize) of the covered clients. We need to
define what we mean by saying that an object covers a client: the
input contains a radius for each object in $\Obj$ and a sensitivity
for each client in $\Cli$, and the client is considered to be covered
by an object if the sum of the radius and the sensitivity is at least
the distance between the object and the client. In the special case when
both the radius and the sensitivity are $0$, this is equivalent to
saying that the client is inside the object; when the radius is
$r$ and the sensitivity is $0$, then this is equivalent to saying that
the client is at distance at most $r$ from the object.  The
objects and the clients may be equipped with
weights and we may want to maximize/minimize the weight of the
selected objects or the weight of covered clients.

The first special case of the problem is when there are no clients at
all: then the task is to select $k$ objects that are pairwise
disjoint. Our algorithm solves this problem in complete generality:
the only condition is that each object is a connected vertex set (i.e. it induces a connected subgraph of $G$).
\begin{restatable}[packing connected sets]{\retheorem}{restateintroindependent}\label{thm:intro_independent}
  Let $G$ be a planar graph, $\Obj$ be a family of connected vertex sets of
  $G$, and $k$ be an integer. In time $|\Obj|^{\Oh(\sqrt{k})}\cdot
  n^{\Oh(1)}$, we can find a set of $k$ pairwise disjoint objects in
  $\Obj$, if such a set exists.
\end{restatable}
We can also solve the weighted version, where we want to select $k$
members of $\Obj$ maximizing the total weight. As a special case, if
we define the open ball $B_{v,d}$ to be the set of vertices at distance less $d$ from $v$ and let $\Obj=\{B_{v,d/2}\mid v\in V(G)\}$, then
Theorem~\ref{thm:intro_independent} gives us an $n^{\Oh(\sqrt{k})}$ time
algorithm for \probScatter, which asks for $k$ vertices that are at
distance at least $d$ from each other (with $d$ being part of the input). For unweighted graphs and fixed values of $d$, this
problem is actually fixed-parameter tractable and can be solved in
time $2^{\Oh(d\log d\cdot \sqrt{k})}\cdot n^{\Oh(1)}$ using a simple application of
bidimensionality \cite{DBLP:conf/esa/Thilikos11}.

If each object in $\Obj$ is a single vertex and $\rad(\cdot)$ assigns a
radius to each object (potentially different radii for different objects),
then we get a natural covering problem. Thus, the following theorem is also a corollary of our general result.

\begin{restatable}[covering vertices with centers of different radii]{\retheorem}{restateintrocovering}\label{thm:intro_covering}
  Let $G$ be a planar graph, let $D,C\subseteq V(G)$ be two subset of
  vertices, let $\rad\colon D\to \mathbb{Z}^+$ be a function, and $k$
  be an integer.  In time $|D|^{\Oh(\sqrt{k})}\cdot n^{\Oh(1)}$, we can
  find a set $S\subseteq D$ of $k$ vertices that maximizes the number
  of vertices covered in $C$, where a vertex $u\in C$ is covered by
  $v\in S$ if the distance of $u$ and $v$ is at most $\rad(v)$.
\end{restatable}

If $D=C=V(G)$, $\rad(v)=d$ for every $v\in V(G)$, and we are looking for a solution fully covering $C$, then we obtain as
a special case \probdDS (also called \probkdCent)\iffull,  that is, the problem of finding a set $S$ of
vertices such that every other vertex is at distance at most $d$ from
$S$\fi. Theorem~\ref{thm:intro_covering} gives an $n^{\Oh(\sqrt{k})}$ time
algorithm for this problem (with $d$ being part of the input). Again, for fixed $d$, bidimensionality
theory gives a $2^{\Oh(d\log d\cdot \sqrt{k})}\cdot n^{\Oh(1)}$ time algorithm \cite{DBLP:journals/talg/DemaineFHT05}.

Theorem~\ref{thm:intro_covering} can be interpreted as covering the
vertices in $C$ by very specific objects: we want to maximize the
number of vertices of $C$ in the union of $k$ objects, where each
object is a ball of radius $\rad(v)$ around a center $v$. If we require
that the selected objects of the solution are pairwise disjoint, then
we can generalize this problem to arbitrary objects.
\begin{restatable}[covering vertices with independent objects]{\retheorem}{restateintrocoveringx}\label{thm:intro_covering2}
  Let $G$ be a planar graph, let $\Obj$ be a set of connected vertex
  sets in $G$, let $C\subseteq V(G)$ be a set of vertices, and let $k$
  be an integer. In time $|\Obj|^{\Oh(\sqrt{k})}\cdot n^{\Oh(1)}$, we can
  find a set $S$ of at most $k$ pairwise disjoint objects in $\Obj$
  that maximizes the number of vertices of $C$ in the union of the
  vertex sets in $S$.
\end{restatable}

Our algorithmic results are also applicable to geometric problems.
By simple reductions\iffull\ (see Section~\ref{sec:appl})\fi, packing and covering geometric problems can be
reduced to problems on planar graphs \iffull (although one has to handle
rounding and precision issues carefully)\fi. In particular, given a set
of disks (of possibly different radii), the problem of selecting $k$
pairwise disjoint disks can be reduced to the problem of selecting
disjoint connected vertex sets in a planar graph, which can be solved
using Theorem~\ref{thm:intro_independent}. \iffull Alternatively, the
algorithmic techniques behind our main results
(Voronoi diagrams, balanced separators, recursion)
can be expressed directly in the geometric setting, yielding a
somewhat simpler algorithm that does not have to handle some of the
degeneracies appearing in the more general setting of planar graphs. In Section~\ref{sec:geom}, we discuss some of these geometric algorithms in a self-contained way, which may help the reader to understand the main ideas of the more technical planar graph algorithm.
\fi

\begin{restatable}[packing disks]{\retheorem}{restateintropackingdisks}\label{thm:intro_packingdisks}
  Given a set $\Obj$ of disks (of possibly different radii) in the
  plane, in time $|\Obj|^{\Oh(\sqrt{k})}$ we can find a set of $k$ pairwise
  disjoint disks, if such a set exists.
\end{restatable}
This is a strong generalization of the results of Alber and
Fiala~\cite{DBLP:journals/jal/AlberF04}, which gives an
$|\Obj|^{\Oh(\sqrt{k})}$ time algorithm only if the ratio of the radii of the
smallest and largest disks can bounded by a constant (in particular, if
all the disks are unit disks). \iffull

\fi As Theorem~\ref{thm:intro_independent} works for arbitrary connected
sets of vertices, we can prove the analog of
Theorem~\ref{thm:intro_packingdisks} for most reasonable sets of connected
geometric objects. We do not want dwell on exactly what kind of
geometric objects we can handle (e.g., whether the objects can have holes, how the boundaries are described etc.), hence we state the
result only for simple (that is, non-self-crossing) polygons.

\begin{restatable}[packing simple polygons]{\retheorem}{restateintropackingconvex}\label{thm:intro_packingconvex}
  Given a set $\Obj$ of simple polygons in the plane, in time
  $|\Obj|^{\Oh(\sqrt{k})}\cdot n^{\Oh(1)}$ we can find a set of $k$ polygons in
  $\Obj$ with pairwise disjoint closed interiors, if such a set exists. Here $n$ is the total number of vertices of the polygons in $\Obj$.
\end{restatable}

Geometric covering problems can be also reduced to planar problems.
The problem of covering the maximum number of points
by selecting $k$ disks from a given set $\Obj$ of disks can be reduced
to a problem on planar graphs and then
Theorem~\ref{thm:intro_covering} can be invoked. This reduction relies
on the fact that covering by a disk of radius $r$ can be expressed as
being at distance at most $r$ from the center of the disk, which is
precisely what Theorem~\ref{thm:intro_covering} is about.
\begin{restatable}[covering points with disks]{\retheorem}{restateintrocoveringdisks}\label{thm:intro_coveringdisks}
  Given a set $\Cli$ of points and a set $\Obj$ of disks (of possibly
  different radii) in the plane, in time $|\Obj|^{\Oh(\sqrt{k})}\cdot
  |\Cli|^{\Oh(1)}$ we can find a set of $k$ disks in $\Obj$ maximizing
  the total number of points they cover in $\Cli$.
\end{restatable}
The problem of covering points with axis-parallel squares (of
different sizes) can be handled similarly. Observe that
an axis-parallel square with side length $s$ covers a point $p$ if and only if $p$ is at
distance at most $s/2$ from the center of the square {\em in the
  $\ell_{\infty}$ metric.} This allows us to reduce the geometric problem
to the planar problem solved by Theorem~\ref{thm:intro_covering}.

\begin{restatable}[covering points with squares]{\retheorem}{restateintrocoveringsquares}\label{thm:intro_coveringsquares}
  Given a set $\Cli$ of points and a set $\Obj$ of axis-parallel
  squares (of possibly different size) in the plane, in time
  $|\Obj|^{\Oh(\sqrt{k})}\cdot |\Cli|^{\Oh(1)}$ we can find a set of $k$
  squares in $\Obj$ maximizing the total number of points they cover in
  $\Cli$.
\end{restatable}

\myparagraph{Hardness results.}\iffull\ There are several lower bounds
suggesting that our main algorithmic result is, in many aspects, best
possible---both in terms of the form of the running time and the
generality of the problem being solved. The \covProb problem we define
in Section~\ref{sec:problem-definition} gives a very general family of
problems, including many artificial problems. Therefore, it is not
very enlightening to show that the running time we obtain for \covProb
cannot be improved. What we really want to know is whether our
algorithm gives the best possible running time in concrete special
cases of interest, such as those in
Theorems~\ref{thm:intro_independent}--\ref{thm:intro_coveringsquares}.

\fi
There have been investigations of the parameterized complexity of
various geometric packing and covering problems, giving tight
ETH-based lower bounds in many cases
\cite{DBLP:conf/esa/Marx05,DBLP:conf/iwpec/Marx06,DBLP:conf/focs/Marx07a}. Many
of these reductions can be described conveniently using the
\gridtiling problem as the source of reductions.  Reductions using
\gridtiling involve a quadratic blow up in the parameter and therefore
they give lower bounds with a square root in the exponent.  For
example, assuming ETH, the problem of finding $k$ disjoint objects
from a set of unit disks, or a set of axis-parallel unit squares, or a
set of unit segments (of arbitrary directions) cannot be solved in
time $f(k)n^{o(\sqrt{k})}$ for any computable function $f$ \cite{DBLP:conf/iwpec/Marx06,DBLP:conf/focs/Marx07a}.  This
shows the optimality of Theorems~\ref{thm:intro_packingdisks} and
\ref{thm:intro_packingconvex}. In \probDS problem for unit disks, the
task is to select $k$ of the disks such that every disk is either
selected or intersected by a selected disk; assuming ETH, there is no
$f(k)n^{o(\sqrt{k})}$ time algorithm for this problem for any
computable function $f$ \cite{DBLP:conf/iwpec/Marx06}. Observe that if
$\Obj$ is a set of unit disks (that is, disks of radius 1) and $C$ set
of centers of these disks, then a subset of $S\subseteq \Obj$ is a
dominating set if and only if replacing every disk in $S$ with a disk
of radius 2 covers every point in $C$. Therefore, covering points with
disks (of the same radius) is more general than \probDS for unit
disks, hence the optimality of Theorem~\ref{thm:intro_coveringdisks}
follows from the lower bounds for \probDS for unit disks. In a similar
way, the optimality of Theorem~\ref{thm:intro_coveringsquares} follows
from the lower bounds on \probDS for axis-parallel unit squares
\cite{DBLP:conf/iwpec/Marx06}.

\iffull
\iffull
As we shall see in Section~\ref{sec:appl}, there are reductions from the geometric
problems to the planar problems.
\else
There are simple reductions from the geometric
problems to the corresponding planar problems.
\fi Besides the algorithmic consequences, these reductions allow us to transfer lower bounds for
geometric problems to the corresponding planar problems. In
particular, it follows that, assuming ETH, the running time in
Theorem~\ref{thm:intro_independent} cannot be improved to
$f(k)n^{o(\sqrt{k})}$ (even if $|\Obj|=n^{\Oh(1)}$). With a direct
reduction from \gridtiling, one can also show that there is no
$f(k)n^{o(\sqrt{k})}$ time algorithm for \probScatter and \probdDS on
planar graphs, with $d$ being part of the input (these reductions will
appear elsewhere).
\fi

Comparing packing results Theorems~\ref{thm:intro_independent} and
\ref{thm:intro_packingconvex} with covering results
Theorems~\ref{thm:intro_covering}, \ref{thm:intro_coveringdisks}, and
\ref{thm:intro_coveringsquares}, one can
observe that our algorithm solves packing problems in much
wider generality than covering problems. It seems that we can handle
arbitrary objects in packing problems, while it is essential for
covering problems that each object is a ``ball,'' that is, it is defined as
a set of points that are at most at a certain distance from a
center. (Theorem~\ref{thm:intro_covering2} seems to be an exception,
as it is a covering problem with arbitrary objects, but notice that we
require independent objects in the solution, so it is actually a
packing problem as well.)  We present a set of hardness results
suggesting that this apparent difference between packing and covering
problems is not a shortcoming of our algorithm, but it is inherent to the
problem: there are very natural geometric covering problems where the
square root phenomenon does not occur.

Our strongest lower bound is not based not on ETH, but on the variant
called Strong Exponential Time Hypothesis (SETH), which can be
informally stated as $n$-variable \cnfsat not having algorithms with
running time $(2-\epsilon)^n$ for any $\epsilon>0$ (cf.~\cite{DBLP:journals/eatcs/LokshtanovMS11}).
\iffull
 Using a result of
P\u{a}tra\c{s}cu and Williams \cite{DBLP:conf/soda/PatrascuW10} and a simple reduction from \probDS, we show
that if the task is to cover all the vertices of a planar graph $G$ by selecting $k$ sets from a collection $\Obj$ of connected vertex sets, then  is unlikely that one can do significantly better than trying
all $|\Obj|^k$ possible sets of objects. 
\begin{restatable}[covering vertices with connected sets, lower bound]{\retheorem}{restateintroplanarcover}\label{thm:intro_planarcover}
  Let $G$ be a planar graph and let $\Obj$ be a set of connected
  vertex sets of $G$.  Assuming SETH, there is no $f(k)\cdot
  (|\Obj|+|V(G)|)^{k-\epsilon}$ time algorithm for any computable
  function $f$ and any $\epsilon>0$ that decides if there are $k$ sets
  in $\Obj$ whose union covers $|V(G)|$.
\end{restatable}

A similar reduction gives a lower bound for covering points with convex polygons.
\else
Using a result of
P\u{a}tra\c{s}cu and Williams \cite{DBLP:conf/soda/PatrascuW10} and a simple reduction from \probDS, we show
that if the task is to cover points with convex polygons, then  is unlikely that one can do significantly better than trying
all $|\Obj|^k$ possible sets of polygons. 
\fi

\begin{restatable}[covering points with convex polygons, lower bound]{\retheorem}{restateintroconvexhard}\label{thm:intro_convexhard}
Let $\Obj$ be a set of convex polygons and let $\pointfam$ be a set of points in the plane.
Assuming SETH, there is no $f(k)\cdot (|\Obj|+|\pointfam|)^{k-\epsilon}$ time algorithm for any computable function $f$ and $\epsilon>0$ that decides if there are $k$ 
polygons in $\Obj$ that together cover $\pointfam$.
\end{restatable}

The convex polygons appearing in the hardness proof of
Theorem~\ref{thm:intro_convexhard} are relatively ``fat'' \iffull(i.e., the
area of each polygon is at most a constant factor smaller than the
smallest enclosing disk), \fi and they have an unbounded number of
vertices. Therefore, it may still be possible that the square root
phenomenon occurs for simpler polygons and we can have $n^{\Oh(\sqrt{k})}$ time algorithms. We show that this is not the case: we give two lower bounds for axis-parallel rectangles. The first bound is for ``thin'' rectangles (of only two types), while the second bound is for rectangles that are ``almost squares.''

\iffull
\begin{restatable}[covering points with thin rectangles, lower bound]{\retheorem}{restateintrorectharda}\label{thm:intro_recthard1}
  Consider the problem of covering a set $\pointfam$ of points by
  selecting $k$ axis-parallel rectangles from a set $\Obj$. Assuming
  ETH, there is no algorithm for this problem with running time $f(k)\cdot(|\pointfam|+|\Obj|)^{o(k)}$ for any
  computable function $f$, even if each rectangle in $\Obj$ is of size
  $1\times k$ or $k\times 1$.
\end{restatable}
\begin{restatable}[covering points with almost squares, lower bound]{\retheorem}{restateintrorecthardb}\label{thm:intro_recthard2}
  Consider the problem of covering a set $\pointfam$ of points by
  selecting $k$ axis-parallel rectangles from a set $\Obj$. Assuming
  ETH, for every $\epsilon_0>0$, there is no algorithm for this problem with running
  time $f(k)\cdot(|\pointfam|+|\Obj|)^{o(k/\log k)}$ for any computable function $f$, even if
  each rectangle in $\Obj$ has both width and height in the range
  $[1-\epsilon_0,1+\epsilon_0]$.
\end{restatable}
\else
\begin{theorem}[covering points with rectangles, lower bound]\label{thm:intro_recthard}
Consider the problem of covering a set $\pointfam$ of points by selecting $k$ axis-parallel rectangles from a set $\Obj$.
\begin{enumerate}
\item Assuming ETH, there is no algorithm for this problem with running time
  $f(k)\cdot(|\pointfam|+|\Obj|)^{o(k)}$ for any computable function $f$, even if each
  rectangle in $\Obj$ is of size $1\times k$ or $k\times 1$.
\item Assuming ETH, for every $\epsilon_0>0$, there is no algorithm for this problem
  with running time $f(k)\cdot(|\pointfam|+|\Obj|)^{o(k/\log k)}$ for any computable function
  $f$, even if each rectangle in $\Obj$ has both width and height in
  the range $[1-\epsilon_0,1+\epsilon_0]$.
\end{enumerate}
\end{theorem}
\fi

\iffull Theorem~\ref{thm:intro_recthard2} \else Theorem~\ref{thm:intro_recthard} \fi shows that even
a minor deviation from the setting of
Theorem~\ref{thm:intro_coveringsquares} makes it unlikely that
$n^{\Oh(\sqrt{k})}$ algorithms exist. Therefore, it seems that for
covering problems the existence of the square root phenomenon depends
not on the objects being simple, or fat, or almost the same size, but
really on the fact that the objects are defined as balls in a metric.

\myparagraph{Our techniques.} The standard technique of
bidimensionality does not seem to be applicable to our problems: it is
not clear for any of the problems how the existence of an
$\Oh(\sqrt{k})\times \Oh(\sqrt{k})$ grid minor helps in solving the
problem\iffull, hence we cannot assume that the input graph has
treewidth $\Oh(\sqrt{k})$.\else.\fi\ In more recent subexponential
parameterized algorithms, we can observe a different algorithmic
pattern: instead of trying to bound the treewidth of the {\em input
  graph,} we define a ``skeleton graph'' describing the structure of
the {\em solution} and use planarity of the skeleton to bound its
treewidth. Then the fact that this skeleton graph has treewidth
$\Oh(\sqrt{k})$ (and, in particular, has balanced separators of size
$\Oh(\sqrt{k})$) can be used to solve the problem. The right choice of
the skeleton graph can be highly nonobvious: for example, for
\probMWay \cite{DBLP:conf/icalp/KleinM12}, the skeleton graph is the
union of the dual of the solution with a minimum Steiner tree, while
for \probSTSP \cite{DBLP:conf/soda/KleinM14}, the skeleton is the
union of the solution with a locally optimal solution. In our problem,
we again define a skeleton graph based on a solution and exploit that
its treewidth is $\Oh(\sqrt{k})$. This time, the right choice for the
skeleton seems to be the Voronoi diagram of the $k$ objects forming
the solution. \iffull Then we exploit the fact that this Voronoi
diagram has separator cycles of length $\Oh(\sqrt{k})$ to find a
suitable way of separating the instance into subproblems.\fi

Given a set $\pointfam$ of points in the plane, the Voronoi region of
a point $p\in \pointfam$ consists of those points of the plane that
are closer to $p$ than any other member of $\pointfam$. The boundaries
of the Voronoi regions are segments that are equidistant to two points from $\pointfam$, forming a diagram that can be considered
a planar graph. We can define Voronoi regions of a graph and a set $\Obj$ of
disjoint connected vertex sets in a similar way, by classifying vertices according to
the closest object in $\Obj$. If the graph is planar, then we can use
the edges on the boundary of the regions (in the dual graph) to construct a planar graph
that is an analog of the Voronoi diagram. If
$|\Obj|=k$, then this diagram has $\Oh(k)$ vertices and\iffull, by a well-known
property of planar graphs,\fi\ has treewidth $\Oh(\sqrt{k})$.

The separator cycle that we need is actually a noose: a closed curve
that intersects the graph only in its vertices and visits each face at
most once. It can be deduced from known results on sphere cut
decompositions that a planar graph with $k$ vertices has a noose that
visits $\Oh(\sqrt{k})$ vertices and there are at most $\frac{2}{3}k$ faces strictly
inside/outside the noose. The noose can be described by a cyclic
sequence of $\Oh(\sqrt{k})$ vertices and faces. Such a sequence of
vertices and faces of the Voronoi diagram can be translated into a
sequence of shortest paths connecting points from $\Obj$ and the
branch vertices of the Voronoi diagram, forming a closed cycle in the
original graph, separating the inside and the outside. These separator
cycles are the most important conceptual objects for our
algorithm. The crucial observation is that, because of the properties
of the Voronoi diagram, objects of the optimum solution inside the
cycle cannot interact with outside world: they cannot intersect the
cycle and, in covering problems, if a point is sufficiently close to
an object inside the cycle, then it is sufficiently close also to an
object on the boundary.

Of course, we do not know the Voronoi diagram of the $k$ objects in
the solution. But since the separator cycles are defined by a
selection of $\Oh(\sqrt{k})$ objects/branch vertices, we can
enumerate $n^{\Oh(\sqrt{k})}$ candidates for them. We branch into $n^{\Oh(\sqrt{k})}$ subproblems indexed by these candidates, and in each subproblem we assume that the selected candidate corresponds to a balanced
noose in the Voronoi diagram of the solution. This assumption has
certain consequences and we modify the instance accordingly. For
example, we can deduce that certain members of $\Obj$ cannot be in the
solution, and hence they can be removed from $\Obj$. After these
modifications, we can observe that the problem falls apart into into
subproblems: this is because there cannot be any interaction between
the objects inside and outside the separator. Therefore, we can
recursively solve these subproblems, where the parameter value is at
most $\frac{2}{3}k$. Because of guessing the separator cycle, we
eventually solve $n^{\Oh(\sqrt{k})}$ subproblems, each with parameter
value at most $\frac{2}{3}k$, which results in the claimed
$n^{\Oh(\sqrt{k})}$ running time following from the recursive formula.

Lower bounds on how $k$ has to appear in the exponent, such as the
lower bounds in
\iffull
Theorems~\ref{thm:intro_convexhard}--\ref{thm:intro_recthard2},\else
Theorems~\ref{thm:intro_convexhard}--\ref{thm:intro_recthard},\fi
\ can be
obtained by parameterized reductions from a W[1]-hard problem. The
strength of the lower bound depends on how the parameter $k$ changes in
the reduction. The reductions based on \gridtiling involve a quadratic
blowup and hence are able to rule out only algorithms of the form
$f(k)n^{o(\sqrt{k})}$ (assuming ETH). The lower bounds in
\iffull
Theorems~\ref{thm:intro_convexhard}--\ref{thm:intro_recthard2}\else
Theorems~\ref{thm:intro_convexhard}--\ref{thm:intro_recthard}\fi
\ are
stronger: they show that even $n^{\Oh(\sqrt{k})}$ time algorithms are
far from being possible. Therefore, they are very different from
typical hardness proofs for planar and geometric problems based on
\gridtiling. Of particular interest is 
\iffull\ Theorem~\ref{thm:intro_recthard2}\else\ 
the second statement of Theorem~\ref{thm:intro_recthard}\fi, where a tight reduction from
\probClique seems problematic, as it seems difficult to
implement the $\Oh(k^2)$ pairwise interaction of the \probClique problem
with squares of almost the same size. Instead, we rely on a nontrivial
hardness result for \partsub \cite{marx-toc-treewidth}, which gives a
strong lower bound even when the graph $H$ to be found is sparse.

%%% Local Variables: 
%%% mode: latex
%%% TeX-master: "voronoi"
%%% End: 

\renewcommand{\retheorem}{rtheorem}

\section{Geometric problems}
\label{sec:geom}

\def\bndmax{8}
\newcommand{\voronoiscale}{1}
\newcommand{\hyp}[6]{
  \pgfmathsetmacro{\voronoiscalex}{1}
  \pgfmathsetmacro{\xz}{#3*\voronoiscalex}
    \pgfmathsetmacro{\yz}{#4*\voronoiscalex}
    \pgfmathsetmacro{\xq}{#1*\voronoiscalex}
    \pgfmathsetmacro{\yq}{#2*\voronoiscalex}
     \pgfmathsetmacro{\xd}{(\xq-\xz)}
     \pgfmathsetmacro{\yd}{(\yq-\yz)}
     \pgfmathsetmacro{\dist}{sqrt((\xd)^2+(\yd)^2)}
  \pgfmathsetmacro{\c}{1}
  \pgfmathsetmacro{\a}{(#5)/(\dist)}
    \pgfmathsetmacro{\b}{(sqrt((\c)^2-(\a)^2))} 
    \pgfmathsetmacro{\xb}{\xd/(2*\c)}
    \pgfmathsetmacro{\yb}{\yd/(2*\c)}
%\pgfmathsetmacro{\myrange}{\arccos(0.03*max(\dist,20))}
    \pgfmathsetmacro{\myrange}{atan(40/min(\dist,20))}
%\pgfmathsetmacro{\myrange}{85}%\arccos(0.01*\dist)}
\path[#6] 
   plot[samples=25,domain=-\myrange:\myrange,variable=\t] 
    ({
    (\xz+\xq)/2 +
 (\xb*\a/cos(\t)) - (\yb*\b*tan(\t))
    },{
      (\yz+\yq)/2+
 (\yb*\a/cos(\t)) + (\xb*\b*tan(\t))
    }) 
%-- ({\xq+60(\xd-\yd)+62*\xd/\voronoiscalex},{\yq+60*(\yd+\xd)+62*\yd/\voronoiscalex})
%-- ({\xq+60(\xd+\yd)+62*\xd/\voronoiscalex},{\yq+60*(\yd-\xd)+62*\yd/\voronoiscalex})
-- ({\xq+122*\xd/(\dist*\voronoiscalex)},{\yq+122*\yd/(\dist*\voronoiscalex)})
 -- cycle;

\fill[red] 
    ({
    (\xz+\xq)/2 +
 (\xb*\a/cos(\myrange)) - (\yb*\b*tan(\myrange))
    },{
      (\yz+\yq)/2+
 (\yb*\a/cos(\myrange)) + (\xb*\b*tan(\myrange))
    }) 
 circle (0.1);
\fill[red] 
    ({
    (\xz+\xq)/2 +
 (\xb*\a/cos(-\myrange)) - (\yb*\b*tan(-\myrange))
    },{
      (\yz+\yq)/2+
 (\yb*\a/cos(-\myrange)) + (\xb*\b*tan(-\myrange))
    }) 
 circle (0.1);
}

\newcommand{\drawvoronoi}{
\tikzexternalenable
\begin{tikzpicture}[scale=0.6]
\path[clip]
  (-\bndmax,-\bndmax) rectangle (\bndmax,\bndmax);
\foreach \x/\y/\r/\cc/\dcc in \disks
{
 \begin{scope}
 \foreach \xx/\yyb/\rr/\ccc/\dccc in \disks
  {
 \ifdefstrequal{\x/\y}{\xx/\yyb}{}{\hyp{\x*\voronoiscale}{\y*\voronoiscale}{\xx*\voronoiscale}{\yyb*\voronoiscale}{\rr-\r}{clip}}
 }
  \fill[\cc]  (-\bndmax,-\bndmax) rectangle (\bndmax,\bndmax);
 \end{scope}
}
\foreach \x/\y/\r/\cc/\diskcolor in \disks
{
\fill[\diskcolor] (\x*\voronoiscale,\y*\voronoiscale) circle (\r);
\draw[black] (\x*\voronoiscale,\y*\voronoiscale) circle (\r);
\begin{scope}
 \foreach \xx/\yyh/\rr/\ccc/\diskcolor in \disks
  {
 \ifdefstrequal{\x/\y}{\xx/\yyh}{}{\hyp{\x*\voronoiscale}{\y*\voronoiscale}{\xx*\voronoiscale}{\yyh*\voronoiscale}{\rr-\r}{clip}}
 }
 \foreach \xx/\yyq/\rr/\ccc/\diskcolor in \disks
  {
 \ifdefstrequal{\x/\y}{\xx/\yyq}{}{\hyp{\x*\voronoiscale}{\y*\voronoiscale}{\xx*\voronoiscale}{\yyq*\voronoiscale}{\rr-\r}{draw}}
 }
\end{scope}
}
\draw (-\bndmax,-\bndmax) rectangle (\bndmax,\bndmax);
\end{tikzpicture}
\tikzexternaldisable
}

\newcommand{\disks}{}
%%% Local Variables: 
%%% mode: latex
%%% TeX-master: "voronoi"
%%% End: 

\definecolor{colsua}{HTML}{FFFFCE}
\definecolor{colsub}{HTML}{FFDDDD}
\definecolor{colsuc}{HTML}{E5FFE5}
\definecolor{colsud}{HTML}{FFDAFF}
\definecolor{colsue}{HTML}{CECEFF}
\definecolor{colsuf}{HTML}{CEFFFF}

\definecolor{colaa}{HTML}{FFFFCE}
\definecolor{colab}{HTML}{FFDDDD}
\definecolor{colac}{HTML}{E5FFE5}
\definecolor{colad}{HTML}{FFDAFF}
\definecolor{colae}{HTML}{CECEFF}
\definecolor{colaf}{HTML}{CEFFFF}

\definecolor{colda}{HTML}{FDFF37}
\definecolor{coldb}{HTML}{FF7A73}
\definecolor{coldc}{HTML}{72FF73}
\definecolor{coldd}{HTML}{FFA0FF}
\definecolor{colde}{HTML}{3737FF}
\definecolor{coldf}{HTML}{37FFFF}

Our main algorithmic result is a technique for solving a
general facility location problem on planar graphs in time
$n^{\Oh(\sqrt{k})}$. With simple reductions, we can use these algorithms
to solve 2-dimensional geometric problems\iffull\ (see Section~\ref{sec:appl})\fi. However, our main
algorithmic ideas can be implemented also directly in the geometric
setting, giving self-contained algorithms for geometric
problems. These geometric algorithms avoid some of the technical
complications that arise in the planar graph counterparts, such as the
Voronoi diagram having bridges or shortest paths sharing
subpaths.
\iffull Unfortunately, a large part of the paper is devoted to the
formal handling of these issues. Therefore, it could be instructive for
the reader to see first a self-contained presentations of some of the
geometric results.
\fi

\textbf{Packing unit disks.}  We start with the \probIS problem for
unit disks: given a set $\Obj$ of closed disks of unit radius in the plane,
the task is to select a set of $k$ pairwise disjoint disks.
% We will
% consider the equivalent problem where, given a set of points (the
% centers of disks), the task is to select $k$ points that are at
% distance more than $2r$ from each other for some fixed distance
% $r$ (from now on, we call such a set of points {\em independent}). 
This problem is known to be solvable in time $n^{\Oh(\sqrt{k})}$
\cite{DBLP:journals/jal/AlberF04,DBLP:conf/compgeom/MarxS14}. We
present another $n^{\Oh(\sqrt{k})}$ algorithm for the problem,
demonstrating how we can solve the problem recursively by focusing on
the Voronoi diagram of a hypothetical solution.  This idea appeared
recently in the context of constructing quasi-polynomial time
approximation schemes (QPTAS) for geometric problems (see, e.g.,
\cite{DBLP:conf/compgeom/Har-Peled14}), but it has not been used
explicitely for exact algorithms.

While the algorithm we present in this section on its own does not
deliver any new result yet, it is significantly different from the
previous algorithms, which crucially use the notion of ``area,'' in
particular, by using the fact that a region of area $\Oh(1)$ can contain
only $\Oh(1)$ independent unit disks. Our algorithm uses only the notion
of distance, making it possible to translate it to the language of
planar graphs, where the notion of area does not make
sense. Furthermore, as we shall see, generalizations to disks of
different radii and to covering problems are relatively easy for our
algorithm. In what follows, it is oftem more convenient to think in terms of an equivalent formulation of the problem where instead of packing disks from $\Obj$ we are packing their centers subject to a constraint that every pair of packed centers has to be at distance more than $2$ from each other. We will switch between these two formulations implicitly.

Let $\pointfam$ be a set of points in the plane. The {\em Voronoi region} of
$p\in \pointfam$ is the set of those points $x$ in the plane that are
``closest'' to $p$ in the sense that the distance of $x$ and $\pointfam$ is
exactly the distance of $x$ and $p$ (see Figure~\ref{fig:vor}(a)). The Voronoi region of $p$ can be
obtained as the intersection of half-planes: for every $p'\in \pointfam$
different from $p$, the Voronoi region is contained in the half-plane
of points whose distance from $p$ is not greater than the distance
from $p'$. This implies that every Voronoi region is convex.

Even though defining Voronoi diagrams and working with them is much
simpler in the plane than for their analogs in planar graphs\iffull\ (see
Section~\ref{sec:partitions})\fi, there is a technical difficulty
specific to the plane. The issue is that the Voronoi region of a point
$p\in \pointfam$ can be infinite and consequently the Voronoi diagram
consists of finite segments and infinite rays. Therefore, it is not
clear what we mean by the graph of the Voronoi diagram. While this
complication does not give any conceptual difficulty in the algorithm,
we need to address it formally.

First, we introduce three new ``guard'' points into $\pointfam$, at
distance more than $r$ from the other points and from each other.  We
introduce these three guards in such a way that every original point
is inside the triangle formed by them (see
Figure~\ref{fig:vor}(b)). It is easy to see now that the Voronoi
region of every original point is finite. Therefore, the only infinite
regions are the regions of the three guards and it follows that finite
segments of the Voronoi diagram form a 2-connected planar graph and
there are three infinite rays in the infinite face. In the case of the
packing problem, if we introduce three new disks corresponding to the
three guards, then these three disks can be always selected into
every solution. Thus instead of trying to find $k$ independent disks
in the original set, we can equivalently try to find $k+3$ disks in
the new set. As increasing $k$ by 3 does not change the asymptotic
running time $n^{\Oh(\sqrt{k})}$ we are aiming for, in the following
we assume that the set $\pointfam$ of center points has this form,
that is, contains three guard points.

\begin{figure}[t]
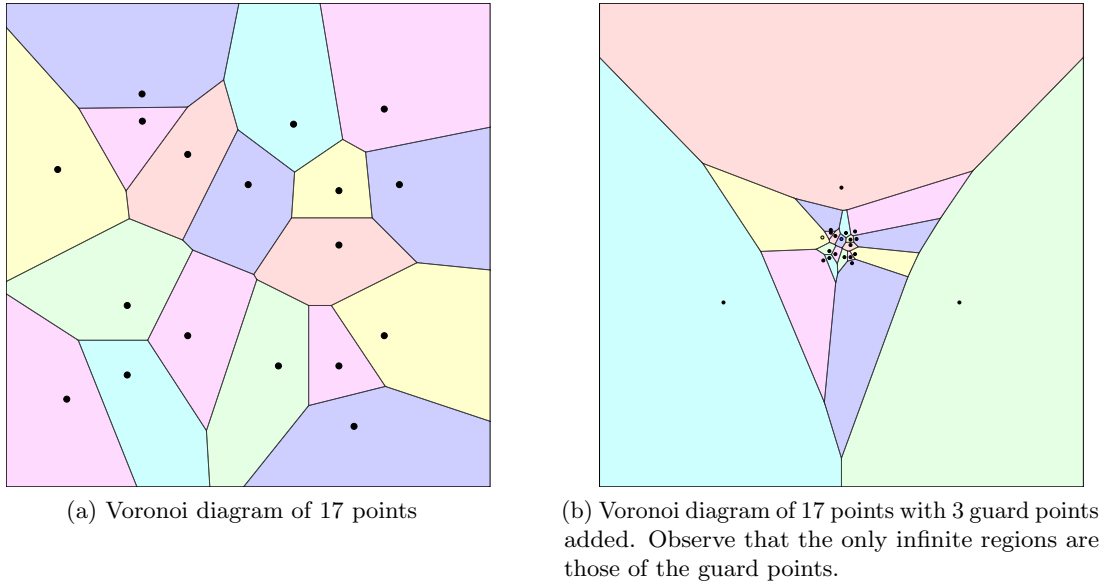

        \centering
        \subfloat[Voronoi diagram of 17 points]{\input{voronoi-fig1}}
\qquad
\subfloat[Voronoi diagram of 17 points with 3 guard points added. Observe that the only infinite regions are those of the guard points.]{\input{voronoi-fig2}}

\caption{Voronoi diagrams}\label{fig:vor}
\end{figure}

If a vertex of the Voronoi diagram has degree more than 3, then this
means that there are four points appearing on a common circle. In
order to simplify the presentation, we may introduce small
perturbation to the coordinates to ensure that this does not happen for any four
points. Moreover, we may identify the ``endpoints'' of the three
infinite rays into a new point at infinity. Therefore, in the
following we assume that the Voronoi diagram is actually a 2-connected
3-regular planar graph. Let us emphasize again that these technicalities
appear only in the geometric setting and will not present a problem
when we are proving the main result for planar graphs.

% If we take different subsets $\pointfam'$ of $k$ points from
% $\pointfam$, then they have different Voronoi diagrams, with vertices
% at different locations. However, we can efficiently find all the
% points of the plane that can possibly be vertices of such Voronoi
% diagrams, as every such point has to be equidistant from at least
% three points of $\pointfam$. Therefore, we construct a set $\impVert$ of $\Oh(|\pointfam|^3)$ points by considering every triple
% $(p_1,p_2,p_3)$ of vertices from $\pointfam$ and finding the point
% equidistant from $p_1$, $p_2$, $p_3$ (there is a unique such point on the plane,
% unless the three vertices a co-linear, in which case there is no such
% point). This set $\impVert$ has the property that no matter which subset $\pointfam'$
% of $\pointfam$ we take, every vertex of the Voronoi diagram is from
% $\impVert$. We will call $\impVert$ the set of {\em important points.}

We are now ready to explain the main combinatorial idea behind the
algorithm for finding $k$ independent unit disks. Consider now a hypothetical solution
consisting of $k$ independent disks and let us consider the Voronoi
diagram of the centers of these $k$ points (see Figure~\ref{fig:vordisk}(a)). To emphasize that we consider the Voronoi diagram of the centers of the $k$ disks in the solution and {\em not} the centers of the $n$ disks in the input, we call this diagram the {\em solution Voronoi diagram.}  As we discussed above, we
may assume that the solution Voronoi diagram is a 2-connected 3-regular planar
graph. There are various separator theorems in the literature showing
that a $k$-vertex planar graph has balanced separators of size
$\Oh(\sqrt{k})$. Certain technicalities appear in these
theorems: for example, they may require the graph to be triangular or
2-connected, the separator may be a cycle in the primal graph or in
the dual graph, etc. Therefore, in Section~\ref{sec:balanc-noos-plane} we give a short proof
showing that the known results on sphere cut decompositions imply a
separator theorem of the form we need. A {\em{noose}} of a
plane graph $G$ is a closed curve $\delta$ on the sphere such that
$\delta$ alternately travels through faces of $G$ and vertices of $G$
and every vertex and face of $G$ is visited at most once. We show that
every 3-regular planar graph $G$ with $k$ faces has a noose $\delta$
of length $\Oh(\sqrt{k})$ (that is, going through $\Oh(\sqrt{k})$ faces
and vertices) that is {\em face balanced} in the sense that there are
at most $\frac{2}{3}k$ faces of $G$ strictly inside $\delta$ and at
most $\frac{2}{3}k$ faces of $G$ strictly outside $\delta$.

\begin{figure}[t]
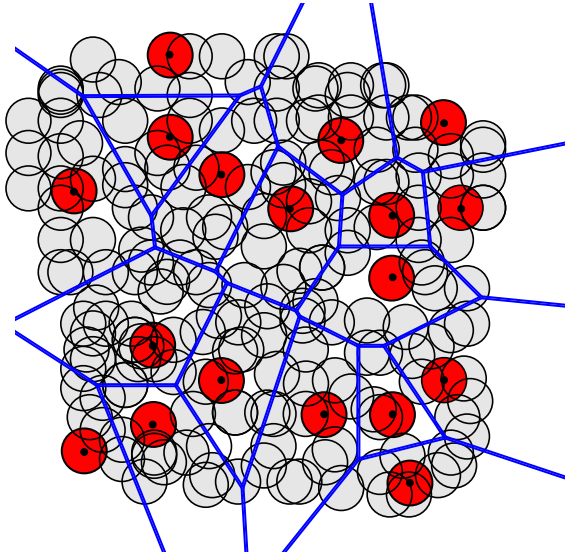
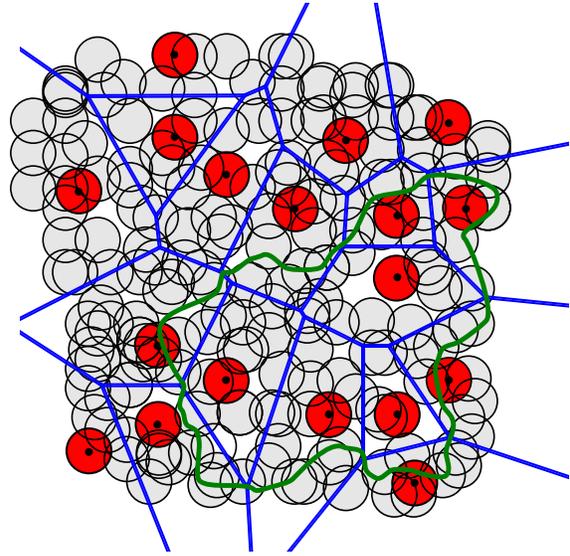
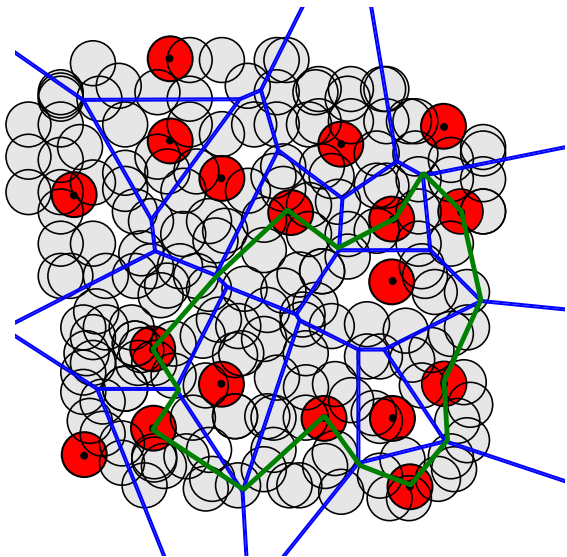
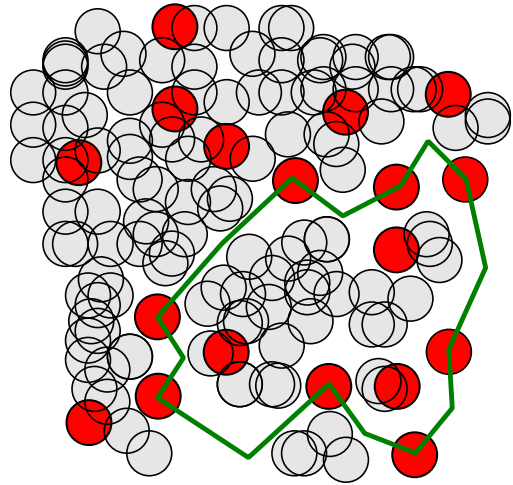

        \centering
        \subfloat[The Voronoi diagram of the centers of the disks in solution.]{\svg{0.45\columnwidth}{vor3}}
        \qquad
        \subfloat[A noose in the Voronoi diagram]{\svg{0.45\columnwidth}{vor4}}\\

        \subfloat[The polygon corresponding to the noose.]{\svg{0.45\columnwidth}{vor5}}
\qquad
        \subfloat[Removing every disk intersecting the polygon breaks the problem into two independent parts.]{\svg{0.45\columnwidth}{vor7}}

\caption{Using a noose in the Voronoi diagram for divide and conquer.}\label{fig:vordisk}
      \end{figure}

      Consider a face-balanced noose $\delta$ of length $\Oh(\sqrt{k})$
      as above (see the green curve in
      Figure~\ref{fig:vordisk}(b)). Noose $\delta$ goes through
      $\Oh(\sqrt{k})$ faces of the solution Voronoi diagram, which
      corresponds to a set $Q$ of $\Oh(\sqrt{k})$ disks of the
      solution. The noose can be turned into a polygon $\Gamma$ with
      $\Oh(\sqrt{k})$ vertices the following way (see the green polygon
      in Figure~\ref{fig:vordisk}(c)). Consider a subcurve of $\delta$ that is contained in the face corresponding to disk $p\in Q$ and its endpoints 
      are vertices $x$ and $y$ of the solution Voronoi
      diagram. Then we can ``straighten'' this subcurve by replacing it with a
      straight line segment connecting $x$ and the center of $p$, and a
      straight line segment connecting the center of $p$ and
      $y$. Therefore, the vertices of the polygon $\Gamma$ are center
      points of disks in $Q$ and vertices of the solution Voronoi
      diagram. Observe that $\Gamma$ intersects the Voronoi regions of
      the points in $Q$ only. This follows from the convexity of the
      Voronoi regions: the segment between the center of $p\in Q$ and
      a point on the boundary of the region of $p$ is fully contained
      in the region of $p$. In particular, this means that $\Gamma$
      does not intersect any disk other than those in $Q$.

The main idea is to use this polygon $\Gamma$ to separate the problem
into two subproblems. Of course, we do not know the solution Voronoi diagram and hence we have no way of computing from it the balanced noose
$\delta$ and the polygon $\Gamma$. However, we can efficiently list
$n^{\Oh(\sqrt{k})}$ candidate polygons. By definition, every vertex of
  the polygon $\Gamma$ is either the center of a disk in $\Obj$ or a vertex of the
solution  Voronoi diagram. Every vertex of the solution Voronoi diagram is equidistant
  from the the centers of three disks in $\Obj$ and for any three such center points (in general
  position) there is a unique point in the plane equidistant from
  them. Thus every vertex of the polygon $\Gamma$ is either a center of a disk in $\Obj$ or can
  be described by a triple of disks in $\Obj$. This means that $\Gamma$
  can be described by an $\Oh(\sqrt{k})$-tuple of disks from $\Obj$. That
  is, by branching into $n^{\Oh(\sqrt{k})}$ directions, we may assume
  that we have correctly guessed the subset $Q$ of the solution and the polygon $\Gamma$.

What do we do with the set $Q$ and the polygon $\Gamma$? First, if $Q$ is indeed part of the solution, then we may remove these disks from $\Obj$ and decrease the target number of disks to be found by $|Q|$. Second, we perform the following cleaning steps:
\begin{enumerate}
\item[(1)] Remove any disk that intersects a disk in $Q$.
\item[(2)] Remove any disk that intersects  $\Gamma$.
\end{enumerate}
Indeed, if $Q$ is part of the solution, then no other disk
intersecting $Q$ can be part of the solution. Moreover, we have
observed above that in the solution the polygon $\Gamma$ is contained in the Voronoi regions of
the points in $Q$ and hence no disk other than the disks in $Q$
intersects $\Gamma$, justifying the removal of such disks. We say that
these removed disks are {\em banned} by $(Q,\Gamma)$.

After these cleaning steps, the instance falls apart into two
independent parts: each remaining disk is either strictly inside $\Gamma$ or
strictly outside $\Gamma$ (see Figure~\ref{fig:vordisk}(d)). Moreover, recall that the noose $\delta$ was
face balanced and hence there are at most $\frac{2}{3}k$ faces of the
solution Voronoi diagram inside/outside $\delta$. This implies that
the solution contains at most $\frac{2}{3}k$ center points
inside/outside $\Gamma$. Therefore, in the two recursive
calls, we need to look for at most that many independent
disks. For $k':=1,\dots,\lfloor\frac{2}{3}k\rfloor$, we recursively try
to find exactly $k'$ independent disks from the input restricted to
the inside/outside $\Gamma$, resulting in $2\cdot \frac{2}{3}k$ recursive
calls.\footnote{Doing a recursive call for each $k'$ may seem unnecessarily complicated at this point:
what we really need is a single recursive call returning the maximum number of independent disks, or $\frac{2}{3}k$ independent disks, whichever is smaller.
However, we prefer to present the algorithm in a way similar to how the
more general problems will be solved later on.} Taking into account
the $n^{\Oh(\sqrt{k})}$
 guesses for $Q$ and $\Gamma$, the number of
subproblems we need to solve is $2\cdot \frac{2}{3}k\cdot
n^{\Oh(\sqrt{k})}=n^{\Oh(\sqrt{k})}$ (as $k\le n$, otherwise there is no
solution) and the parameter value is at most $\frac{2}{3}k$ in each
subproblem. Therefore, if we denote by $T(n,k)$ the time needed to
solve the problem with at most $n$ points and parameter value at most
$k$, we arrive to the recursion
\[
T(n,k)=n^{\Oh(\sqrt{k})} \cdot T(n,(2/3)k).
\]
Solving the recursion gives
\begin{align*}
T(n,k)& {=n^{\Oh(\sqrt{k})} \cdot n^{\Oh(\sqrt{\frac{2}{3}k})} \cdot
 n^{\Oh(\sqrt{(\frac{2}{3})^2 k})} \cdot n^{\Oh(\sqrt{(\frac{2}{3})^3 k})} \cdots}\\ &{=
 n^{\Oh((1+(\frac{2}{3})^{\frac{1}{2}}+(\frac{2}{3})^{\frac{2}{2}}+(\frac{2}{3})^{\frac{3}{2}}+\dots)\sqrt{k})}=n^{\Oh(\sqrt{k})},} 
\end{align*}
as the coefficient of $\sqrt{k}$ in the exponent is a constant (being
the sum of a geometric series with ratio $\sqrt{2/3}$).  Therefore,
the total running time for finding $k$ independent disks is
$n^{\Oh(\sqrt{k})}$.  This proves the first result: packing unit disks
in the plane in time $n^{\Oh(\sqrt{k})}$.  Let us repeat that this
result was known before
\cite{DBLP:journals/jal/AlberF04,DBLP:conf/compgeom/MarxS14}, but as
we shall see, our algorithm based on Voronoi diagrams can be
generalized to objects of different size, planar graphs, and covering
problems.

\textbf{Covering points by unit disks.}  Let us now consider the
following problem: given a set $\Obj$ of unit disks and a set $\Cli$
of client points, we need to select $k$ disks from $\Obj$ that
together cover every point in $\Cli$.  We show that this problem
can be solved in time $n^{\Oh(\sqrt{k})}$ using an approach based on
finding separators in the Voronoi diagram.

Similarly to the way we handled the packing of unit disks, we can
consider the Voronoi diagram of the center points in the
solution. Note, however, that this time the disks in the solution are
not necessarily disjoint, but this does not change the fact that their
center points (which can be assumed to be distinct) define a Voronoi
diagram. Therefore, it will be convenient to switch to an equivalent
formulation of the problem described in terms of the centers of the
disks: $\Obj$ is a set of points and we say that a selected point in
$\Obj$ covers a point in $\Cli$ if their distance is at most $1$.

Similarly to the case of packing, we can try $n^{\Oh(\sqrt{k})}$
possibilities to guess a set $Q\subseteq \Obj$ of center points and a
polygon $\Gamma$ that corresponds to a face-balanced noose. The
question is how to use $\Gamma$ to split the problem into two
independent subproblems.  The cleaning steps (1) and (2) above for the
packing problem are no longer applicable: the solution may contain
disks intersecting the disks with centers in $Q$ and the solution may
contain further disks intersecting the polygon $\Gamma$. What we do
instead is the following. First, if we assume that $Q$ is part of the
solution, then any point in $\Cli$ covered by some point in $Q$ can be
removed, as it is already covered. Second, we know that in the
solution Voronoi diagram every point of $\Gamma$ belongs to the
Voronoi region of some point in $Q$, hence we can remove any point
from $\Obj$ that is inconsistent with this assumption. That is, if
there is a $p\in\Obj$ and $v\in \Gamma$ such that $p$ is closer to $v$
than to every point in $Q$, then $p$ can be safely removed from $\Obj$;
we say in this case that $(Q,\Gamma)$ {\em bans $p$}. For a
$p\in \Obj$ and for each segment of $\Gamma$, it is not difficult to
check if the segment contains such a point $v$ (we omit the
details). Thus we have now the following two cleaning steps:
\begin{enumerate}
\item[(1)] Remove every point from $\Cli$ that is covered by $Q$.
\item[(2)] Remove every point from $\Obj$ that is closer to a point of $\Gamma$ than every point in $Q$.
\end{enumerate}
Let $\Obj_\textup{in}$ and $\Obj_\textup{out}$ be the remaining points
in $\Obj$ strictly inside/outside $\Gamma$ and let $\Cli_\textup{in}$
and $\Cli_\textup{out}$ be the remaining points in $\Cli$ strictly
inside/outside $\Gamma$.  We know that the solution contains at most
$\frac{2}{3}k$ center points inside/outside $\Gamma$. Therefore, for
$k'=1,\dots,\lfloor\frac{2}{3}{k}\rfloor$, we solve two subproblems, with point sets
$(\Obj_\textup{in}, \Cli_\textup{in})$ and
$(\Obj_\textup{out},\Cli_\textup{out})$.

If there is a set of $k_\textup{in}$ of center points in $\Obj_\textup{in}$
covering $\Cli_\textup{in}$ and there is a set of $k_\textup{out}$
center points in $\Obj_\textup{out}$ covering $\Cli_\textup{out}$, then,
together with $Q$, they form a solution of
$|Q|+k_\textup{in}+k_\textup{out}$ center points. By solving the defined
subproblems optimally, we know the minimum value of $k_\textup{in}$
and $k_\textup{out}$ required to cover $\Cli_\textup{in}$ and
$\Cli_\textup{out}$, and hence we can determine the smallest solution
that can be put together this way.  But is it true that we can always
put together an optimum solution this way? The problem is that, in
principle, the solution may contain a center point $p\in \Obj_\textup{out}$
that covers some point $q\in \Cli_\textup{in}$ that is not covered by
any center point in $\Obj_\textup{in}$. In this case, in the optimum solution
the number of center points selected from $\Obj_\textup{in}$ can be strictly
less than what is needed to cover $\Cli_\textup{in}$ and hence the way
we are putting together a solution cannot result in an optimum
solution.
%  If the $k_Q$ disks in $Q$
% cover $m_Q$ points of $\Cli$, there is a set of $k_\textup{in}$ points
% in $\Obj_\textup{in}$ covering $m_\textup{in}$ points of
% $\Cli_\textup{in}$, and there is a set of $k_\textup{out}$ points in
% $\Obj_\textup{out}$ covering $m_\textup{out}$ points of
% $\Cli_\textup{out}$, then they can be put together to form a solution
% of $k_Q+k_\textup{in}+k_\textup{out}$ points of $\Obj$ that cover
% $m_Q+m_\textup{in}+m_\textup{out}$. Note that no double counting
% occurs: the set of points covered by $Q$, the set $\Cli_\textup{in}$,
% and the set in $\Cli_\textup{out}$ are disjoint.
% By solving the
% defined subproblems optimally, we know the maximum value of
% $m_\textup{in}$ and $m_\textup{out}$ for every $k_\textup{in}$ and
% $k_\textup{out}$, and can compute the maximum of
% $m_Q+m_\textup{in}+m_\textup{out}$ that can be obtained. But is it true that we can always
% put together an optimum solution this way? The problem is that, in
% principle, the solution may contain center point from
% $x\in \Obj_\textup{out}$ that covers point $c\in \Cli_\textup{in}$ that is
% not covered by any other disk. In this case, the number of points
% covered can be strictly more than $m_Q+m_\textup{in}+m_\textup{out}$
% and hence trying to maximize this sum does not necessarily gives us an
% optimum solution.

Fortunately, we can show that this problem never arises, for the
following reason. Suppose that there is such a $p\in
\Obj_\textup{out}$ and $q\in \Cli_\textup{in}$. As $p$ is outside
$\Gamma$ and $q$ is inside $\Gamma$, the segment connecting $p$ and
$q$ has to intersect $\Gamma$ at some point $v\in \Gamma$, which means
$\dist(p,q)=\dist(p,v)+\dist(v,q)$. By cleaning step (2), there has to be a $p'\in
Q$ such that $\dist(p',v)\le \dist(p,v)$, otherwise $p$ would be banned and we
would have removed it from $\Obj$.  This means that
$\dist(p,q)=\dist(p,v)+\dist(v,q)\ge \dist(p',v)+\dist(v,q)\ge \dist(p',q)$. Therefore, if $p$
covers $q$, then so does $p'\in Q$. But in this case we would have
removed $q$ from $\Cli$ in the first cleaning step. Thus we can indeed
obtain an optimum solution the way we proposed, by solving optimally
the defined subproblems.

As in the case of packing, we have $2\cdot\frac{2}{3}k\cdot
n^{\Oh(\sqrt{k})}=n^{\Oh(\sqrt{k})}$ subproblems, with parameter value at
most $\frac{2}{3}k$. Therefore, the same recursion applies to the
running time, resulting in an $n^{\Oh(\sqrt{k})}$ time algorithm.

\textbf{Packing in planar graphs.} How can we translate the geometric ideas
explained above to the context of planar graphs?  Let $G$ be an
edge-weighted planar graph and let $\Fam$ be a set of disjoint 
``objects,'' where each object is a connected set of vertices in
$G$. Then we can define the analog of the Voronoi regions in a
straightforward way: for every $p\in \Fam$, let $M_p$ contain every
vertex $v$ to which $p$ is the closest object in $\Fam$, that is,
$\dist(v,\Fam)=\dist(v,p)$. By a perturbation of the edge weights, we may
assume that there are no ties: vertex $v$ cannot be at exactly the
same distance from two objects $p_1,p_2\in \Fam$. It follows that the
sets $M_p$ form a partition of $V(G)$ (here we use that the objects in
$\Fam$ are disjoint). It is easy to verify that region $M_p$ has the
following convexity property: if $v\in M_p$ and $P$ is a shortest path
between $v$ and $p$, then every vertex of $P$ is in $M_p$. 

While Voronoi regions are easy to define in graphs, the proper definition of
Voronoi diagrams is far from obvious and it is also nontrivial how a
noose $\delta$ in the Voronoi diagram defines the analog of the
polygon $\Gamma$. We leave the discussion of these issues to
Section~\ref{sec:algo}, here we only define in an abstract way what
our goal is and state in Lemma~\ref{lem:guardedenum0} below (a
simplified version of) the main technical tool that is at the core of
the algorithm.  Note that the statement of
Lemma~\ref{lem:guardedenum0} involves only the notion of Voronoi
regions, hence there are no technical issues in interpreting and using
it. However, in the proof we have to define the analog of the Voronoi
diagram for planar graphs and address issues such that this diagram is
not 2-connected etc. We defer the required technical definitions
to Section~\ref{sec:algo}.

Let us consider first the packing problem: given an edge-weighted
graph $G$, a set $\Obj$ of $d$ objects (connected subsets of vertices),
and an integer $k$, the task is to find a subset $\Fam\subseteq \Obj$ of
$k$ pairwise disjoint objects.  Looking at the algorithm for packing
unit disks described above, what we need is a suitable {\em guarded
  separator}, which is a pair $(Q,\Gamma)$ consisting of a set
$Q\subseteq \Obj$ of $\Oh(\sqrt{k})$ objects and a subset
$\Gamma\subseteq V(G)$ of vertices. If there is a hypothetical
solution $\Fam\subseteq \Obj$ consisting of $k$ disjoint objects, then
we would like to have a guarded separator $(Q,\Gamma)$ satisfying the
following three properties: (1) $Q$ is subset of the solution, (2) $\Gamma$ is fully contained in the
Voronoi regions of the objects in $Q$, and and (3) $\Gamma$  separates the
objects in $\Fam$ in a balanced way. Our main technical result is that
it is possible to enumerate a set of $d^{\Oh(\sqrt{k})}$
guarded separators such that for every solution $\Fam$, one of the
enumerated guarded separators satisfies these three properties. We state here a simplified version that is suitable for packing problems.

\begin{lemma}\label{lem:guardedenum0}
  Let $G$ be an $n$-vertex edge-weighted planar graph, $\Obj$ a set of
  $d$ connected subsets of $V(G)$, and $k$ an integer. We can
  enumerate (in time polynomial in the size of the output) a set
  $\Vorsepfam$ of $d^{\Oh(\sqrt{k})}$ pairs $(Q,\Gamma)$ with
  $Q\subseteq \Obj$, $|Q|=\Oh(\sqrt{k})$, $\Gamma\subseteq V(G)$ such
  that the following holds. If $\Fam\subseteq \Obj$ is a set of $k$
  pairwise disjoint objects, then there is a pair $(Q,\Gamma)\in
  \Vorsepfam$ such that
\begin{enumerate}
\item $Q\subseteq \Fam$,
\item if $(M_p)_{p\in\Fam}$ are the Voronoi regions of $\Fam$, then $\Gamma\subseteq \bigcup_{p\in Q}M_p$,
\item for every connected component $C$ of $G-\Gamma$, there are at most
  $\frac{2}{3}k$ objects of $\Fam$ that are fully contained in $C$.
\end{enumerate}
\end{lemma}
The proof goes along the same lines as the argument for the geometric
setting. After carefully defining the analog of the Voronoi diagram,
we can use the planar separator result to obtain a noose $\delta$. The
same way as this noose was turned into a polygon in the geometric
algorithm, we ``straighten'' the noose into a closed walk in the graph
connecting using shortest paths $\Oh(\sqrt{k})$ objects and
$\Oh(\sqrt{k})$ vertices of the Voronoi diagram. The vertices of this
walk separate the objects that are inside/outside the noose, hence it has
the required properties. Thus by trying all sets of $\Oh(\sqrt{k})$
objects and $\Oh(\sqrt{k})$ vertices of the Voronoi diagram, we can
enumerate a suitable set $\Vorsepfam$. A technical difficulty in the
proof is that the definition of the vertices of the Voronoi diagram is
nontrivial. Moreover, to achieve the bound $d^{\Oh(\sqrt{k})}$ instead
of $n^{\Oh(\sqrt{k})}$, we need a nontrivial way of finding a set of
$d^{\Oh(1)}$ candidate vertices; unlike in the geometric setting,
enumerating vertices equidistant from three objects is not sufficient.

Armed with the set $\Vorsepfam$ given by Lemma~\ref{lem:guardedenum0}, the packing problem can be solved in
a way analogous to how we handled unit disks. We guess a pair
$(Q,\Gamma)\in\Vorsepfam$ that satisfies the three properties of
Lemma~\ref{lem:guardedenum0}. From the first property, we know that
$Q$ is part of the solution, hence $Q$ can be removed from $\Obj$ and
the target number $k$ of objects to be found can be decreased by
$|Q|$. From the second property, we know that in the solution $\Fam$
the set $\Gamma$ is fully contained in the Voronoi regions of the
objects in $Q$. Thus we can remove any object from $\Obj$ intersecting
$\Gamma$. That is, we have to solve the problem on the graph
$G-\Gamma$. We can solve the problem independently on the components
of $G-\Gamma$ and the third property of Lemma~\ref{lem:guardedenum0}
implies that each component contains at most $\frac{2}{3}k$ objects of
the solution. Thus for each component $C$ of $G-\Gamma$ containing at
least one object of $\Obj$ and for $k'=1,\dots,\lfloor\frac{2}{3}k\rfloor$, we
recursively solve the problem on $C$ with parameter $k'$, that is, we
try to find $k'$ disjoint objects in $C$. Assuming that $(Q,\Gamma)$
indeed satisfied the properties of Lemma~\ref{lem:guardedenum0} for
the solution $\Fam$, the solutions of these subproblems allow us to
put together a solution for the original problem. As at most $d$ components of $G-\Gamma$ can
contain objects from $\Obj$, we recursively solve at most $d\cdot
\frac{2}{3}k$ subproblems for a given $(Q,\Gamma)$. Therefore, the
total number of subproblems we need to solve is at most $d\cdot
\frac{2}{3}k \cdot |\Vorsepfam|=d\cdot \frac{2}{3}k \cdot
d^{\Oh(\sqrt{k})}=d^{\Oh(\sqrt{k})}$ (assuming $k\le d$). If we denote by
$T(n,d,k)$ the running time for $|V(G)|\le n$ and $|\Obj|\le d$, then we
arrive to the recursion $T(n,d,k)=d^{\Oh(\sqrt{k})}\cdot
T(n,d,(2/3)k)+n^{\Oh(1)}$, which gives $T(n,k)=d^{\Oh(\sqrt{k})}\cdot
n^{\Oh(1)}$.

\textbf{Covering in planar graphs.}  Let us consider now the following
graph-theoretic analog of covering points by unit disks: given an
edge-weighted planar graph $G$, two sets of vertices $\Obj$ and $\Cli$,
and integers $k$ and $r$, the task is to find a set $\Fam\subseteq
\Obj$ of $k$ vertices that covers every vertex in
$\Cli$. Here we say that $p\in \Obj$ covers $q\in \Cli$ if $\dist(p,q)\le
r$, that is, we can imagine $p$ to represent a ball of radius $r$ in the graph with
center at $p$. Note that, unlike in the case of packing, $\Obj$ is a set of vertices, not a set of connected sets.

Let $\Fam$ be a hypothetical solution. We can construct the set
$\Vorsepfam$ given by Lemma~\ref{lem:guardedenum0} and guess a guarded
separator $(Q,\Gamma)$ satisfying the three properties. As we assume
that $Q$ is part of the solution (second property of
Lemma~\ref{lem:guardedenum0}), we decrease the target number $k$ of
vertices to select by $|Q|$ and we can remove from $\Cli$ every vertex
that is covered by some vertex in $Q$; let $\Cli'$ be the remaining set
of vertices. Moreover, by the third property of
Lemma~\ref{lem:guardedenum0}, we can assume that in the solution $\Fam$,
the set $\Gamma$ is fully contained in the Voronoi regions of the
vertices in $Q$. This means that if there is a $p\in \Obj\setminus Q$ and
$v\in \Gamma$ such that $\dist(p,v)<\dist(p,Q)$, then $p$ can be
removed from $\Obj$, since it is surely not part of the solution. As in the case of covering with disks, we say that $(Q,\Gamma)$ {\em bans} $p$. Let
$\Obj'$ be the remaining set of vertices. For every component $C$ of
$G-S$ and $k'=1,\dots,\lfloor\frac{2}{3}k\rfloor$, we recursively solve the problem
restricted to $C$, that is, with the restrictions $\Obj'[C]$ and
$\Cli'[C]$ of the vertex sets. It is very important to point out that
now (unlike how we presented the packing problem above) we {\em do not}
change the graph $G$ in each call: we use the same graph $G$ with the
restricted sets $\Obj'[C]$ and $\Cli'[C]$. The reason is that
restricting to the graph $G[C]$ could potentially change the distances
between two vertices $x,y\in C$, as it is possible that the shortest
$x-y$ path leaves $C$. 

 If $k_C$ is the minimum number of vertices
in $\Obj'[C]$ that can cover $\Cli'[C]$, then we know that there are
$|Q|+ \sum k_C$ vertices in $\Obj$ that cover every vertex in $\Cli$. We
argue that if there is a solution, we can obtain a solution
this way. Analogously to the case of covering vertices with disks, we have
to show that in the solution $\Fam$, every vertex in $\Cli'[C]$ is
covered by some vertex in $\Obj'[C]$: it is not possible that there are
two distinct components $C_1,C_2$ of $G-\Gamma$ and some $q\in \Cli'[C_1]$ is
covered only by $p\in \Obj'[C_2]$. Suppose that this is the case, and
let $P$ be a shortest $p-q$ path (which has length at most $r$). As $p$
and $q$ are in two different components of $G-\Gamma$, the path $P$
has to intersect $\Gamma$ at some vertex $v\in \Gamma$ and we have
$\dist(p,q)=\dist(p,v)+\dist(v,q)$. We know that there has to be a
$p'\in Q$ such that $\dist(p',v)\le \dist(p,v)$, otherwise $p$ would be banned and it would not be in $\Obj'$.
Then the same calculation as in the geometric  case shows that $p'\in Q$ also covers $q$, which means that 
we removed $q$ and it cannot be in $\Cli'$. Thus we can
indeed obtain a solution by solving the subinstances restricted to the
components of $G-\Gamma$. As in the case of packing, we have at most
$d\cdot \frac{2}{3}k\cdot d^{\Oh(\sqrt{k})}$ subproblems and the running time
$d^{\Oh(\sqrt{k})}\cdot n^{\Oh(1)}$ follows the same way.
%  This means that
% $\dist(p,q)=\dist(p,v)+\dist(v,q)\ge \dist(p',v)+\dist(v,q)\ge
% \dist(p',q)$. That is, if $p$ covers $q$, then so does $p'$ and
% hence we removed $q$ and it cannot be in $\Cli'$. Thus we can
% indeed obtain a solution by solving the subinstances restricted to the
% components of $G-\Gamma$. As in the case of packing, we have at most
% $(dk)\cdot d^{\Oh(\sqrt{k})}$ subproblems and the running time
% $d^{\Oh(\sqrt{k})}\cdot n^{\Oh(1)}$ follows the same way.
 
\textbf{Covering in planar graphs (maximization version).}  Let us
consider now the variant of the previous problem where we want to
select $k$ vertices from $\Obj$ that cover the maximum number of vertices
in $\Cli$. We proceed the same way as before: we guess a guarded
separator $(Q,\Gamma)$, remove from $\Cli$ the set of vertices covered
by $Q$ (let $m$ be the number of these vertices), remove from $\Obj$ the
set of vertices banned by $(Q,\Gamma)$, and recursively solve the problem for
every component $C$ of $G-\Gamma$ and every $k'=1,\dots,
\lfloor\frac{2}{3}k\rfloor$. Having solved these subproblems, we have at our hand
the values $m(C,k')$, the maximum number of vertices in $\Cli'[C]$ that
can be covered by $k'$ vertices from $\Obj'[C]$. How can we compute from
these values the maximum number of vertices in $\Cli$ that can be
covered by $k$ vertices from $\Obj$? We need to solve a knapsack-type
problem: for each component $C$, we have to select a solution
containing a certain number $0\le k_C \le \frac{2}{3}k$ vertices of
$\Obj$ such that the sum of the $k_C$'s is at most $k-|Q|$ and the sum
of the values $m(C,k_C)$ is maximum possible. Given the values
$m(C,k')$, this maximum can be computed by a standard polynomial-time
dynamic programming algorithm. Thus instead of just deciding if $k$
vertices from $\Obj$ are sufficient to cover every vertex in $\Cli$, we
can also find a set of $k$ vertices covering the maximum number of
vertices from $\Cli$. Let us point out that in this problem it is really
essential that we solve the subinstances for each value of $k'$
instead of just finding a minimum/maximum cardinality solution: as the
optimum solution may not cover all vertices in a component $C$, computing the
minimum number of vertices required to cover all vertices in $C$ is clearly not sufficient.

\textbf{Covering in planar graphs (nonuniform radius).}  A natural
generalization of the covering problem is when every vertex $p\in \Obj$
is given a radius $r(p)\ge 0$ and now we say that $p$ covers a vertex
$q\in \Cli$ if $\dist(p,q)\le r(p)$. That is, now the vertices in $\Obj$
represent balls with possibly different radii. 

There are two ways in which we can handle this more general
problem. The first is a simple graph-theoretic trick. Let $R$ be the
maximum of all the $r(p)$'s. For every $p\in \Obj$, let us attach a
path of length $R-r(p)$ to $p$, let us replace $p$ in $\Obj$ with the
other end $p'$ of this path. Observe that now a vertex $q\in \Cli$ is
at distance at most $r(p)$ from $p$ if and only if it is at distance
at most $R$ from $p'$. Therefore, by solving the resulting covering
problem with uniform radius $R$, we can solve the original problem as
well. Note that, interestingly, this trick uses the flexibility of the
problem being stated in terms of graphs and there is no analogous
geometric trick to solve the problem of covering points with disks of
nonuniform radius in the plane: by attaching the paths, we are
distorting the distance metric in a specific way, whose geometric
meaning is difficult to interpret in the 2-dimensional plane.

The second way is somewhat more complicated, but it seems to be the
robust mathematical solution of the issue and it has also a geometric
interpretation.  The issue of nonuniform radius can be handled by
working with the {\em additively weighted} version of the Voronoi
diagram, that is, instead of defining the Voronoi regions of
$\pointfam$ by comparing the distances $\dist(p,v)$ for $p\in
\pointfam$, we compare the weighted distances $\dist(p,v)-r(p)$. It
can be verified that the main arguments of the algorithm described
above go through. For example, it remains true that Voronoi regions
have the convexity property that if $v$ is in the region of $p$, then
the shortest path from $v$ to $p$ is fully contained in the Voronoi region of
$p$. Given a guarded separator $(Q,\Gamma)$, the crucial property
allowing us to separate the problem to the components of $G-\Gamma$
was that if $p\in\Obj$ and $q\in\Cli$ are in two different components
and $p$ covers $q$, then some $p'\in Q$ also covers $q$. This property
also remains true: we can redo the same calculation to compare
$\dist(p,q)-r(p)$ and $\dist(p',q)-r(p')$.

\textbf{Combining covering and packing.} Comparing the algorithms for
packing objects in planar graphs and for covering vertices by balls of
radius $r$ in planar graphs, we can observe that the set $\Obj$ played
very different roles in the two algorithms: in the packing algorithm
$\Obj$ contained the actual objects, while in the covering algorithm
$\Obj$ contained only the centers of the balls. The reason for this is
fundamental: in order to define the Voronoi regions of the
solution $\Fam\subseteq \Obj$, we need $\Fam$ to be a set of disjoint
objects (which is not true for the balls in the solution of the
covering problem, but of course true for the centers of these
balls). This means that we cannot generalize the covering algorithm to
objects different from metric balls simply by putting 
objects into $\Obj$ more general than single vertices. However, we have no trouble
obtaining a common generalization if we require that the objects
selected from $\Obj$ are disjoint. That is, in the {\em independent
  covering} problem, we are given a set $\Obj$ of connected objects in
a planar graph $G$, a set $\Cli$ of vertices, and integers $k$ and $r$,
the task is to select a {\em pairwise disjoint} subset $\Fam\subseteq \Obj$
of size exactly $k$ that covers the maximum number of vertices in $\Cli$
(where, as usual, an object $p\in \Obj$ covers a vertex $q\in \Cli$ if
$\dist(p,q)\le r$). Exactly the same algorithm as above goes through: as
the solution $\Fam$ is disjoint, it is possible to define the Voronoi
regions.

In the next section, we define the \covProb problem that
generalizes all our applications. The main algorithmic result is
expressed as an algorithm for this problem
(Theorem~\ref{thm:main1}). To handle covering with nonuniform radius,
the problem formulation allows different radius $r(p)$ for each $p\in
\Obj$. However, if we allow nonuniform radius, then the requirement
stating that we have to select disjoint objects should be replaced by
a technical condition that we have to select a {\em normal} family
of objects (see Section~\ref{sec:problem-definition}). Essentially,
this condidition states that in the solution $\Fam$, every vertex of every object
$p\in \Fam$ should be in the weighted Voronoi region of $p$. 
In other words, if we allow arbitrary radii, it could in principle happen that one object $p$ has so large radius compared to some other object $p'$ that the weighted Voronoi region of $p$ actually contains some vertices of $p'$, thus ``eating'' the object $p'$ itself and making it not contained in its own Voronoi region. The normality condition states that, in addition to being disjoint, the solution $\Fam$ does not contain such pathological situations. Observe that when the objects are single vertices only, then this problem is not very dangerous: if $p$ ``dominates'' $p'$ in the way described above, then every client covered by $p'$ is also covered by $p$, and we can discard $p'$ from $\Obj$ because taking $p$ instead is always more profitable. Unfortunately, if the objects are not just single vertices, or they are equipped with nonuniform costs, then this
technicality makes the arising independent covering problems with
nonuniform radii somewhat unnatural.

%%% Local Variables: 
%%% mode: latex
%%% TeX-master: "voronoi"
%%% End: 

\section{The general problem}
\label{sec:prelims}
\iffull

In this section we introduce the generic problem \covProb. The main result of this paper is an algorithm for this problem, expressed in Theorem~\ref{thm:main1}. Before we give the algorithm for \covProb, in Section~\ref{sec:appl} we shall see how the concrete results mentioned in the introduction (i.e., Theorem~\ref{thm:intro_independent}-\ref{thm:intro_coveringsquares}) follow from Theorem~\ref{thm:main1} by simple reductions.

\subsection{Problem definition}
\label{sec:problem-definition}
\else
\label{sec:problem-definition}
\fi

Suppose we are given an undirected graph $G$ embedded on a sphere $\Sigma$, together with a positive edge weight function $\wei\colon E(G)\to \Rp$. We are given a family of {\em{objects}} $\Obj$ and a family of {\em{clients}} $\Cli$. 

Every object $p\in \Obj$ has three attributes. It has its {\em{location}} $\loc(p)$, which is a nonempty subset of vertices of $G$ such that $G[\loc(p)]$ is a connected graph. Moreover, it has its {\em{cost}} $\cst(p)$, which is a real number. Note that costs may be negative. Finally, it has its {\em{radius}} $\rad(p)$, which is just a nonnegative real value denoting the strength of domination imposed by $p$. \iffull Note that locations of objects can intersect, and even there can be multiple objects with exactly the same location.
\fi

Every client $q\in \Cli$ has three attributes. It has its {\em{placement}} $\pla(q)$, which is just a vertex of $G$ where the client resides. It has also its {\em{sensitivity}} $\sen(q)$, which is a real value that denotes how sensitive the client is to domination from objects. Finally, it has also {\em{prize}} $\pri(q)$, which is a real value denoting the prize gained by dominating the client. Note that there can be multiple clients placed in the same vertex, and the prizes may be negative.

We say that a subfamily $\Fam\subseteq \Obj$ is {\em{normal}} if locations of objects from $\Fam$ are disjoint, and moreover $\dist(\loc(p_1),\loc(p_2))>\rad(p_1)-\rad(p_2)$ for all pairs $(p_1,p_2)$ of different objects in $\Fam$; here $\dist(X,Y)$ denotes the distance between two vertex sets in $G$ w.r.t. edge weights $\wei$. As we require the same inequality for the pair $(p_2,p_1)$ as well, it follows that actually $\dist(\loc(p_1),\loc(p_2))>|\rad(p_1)-\rad(p_2)|$ is true. In particular, this implies disjointness of locations of objects from $\Fam$, but if all radii are equal, then normality boils down to just disjointness of locations. Note also that a subfamily of a normal family is also normal.

We say that a client $q$ is {\em{covered}} by an object $p$ if the following holds: $\dist(\pla(q),\loc(p))\leq \sen(q)+\rad(p)$. \iffull In other words, a client gets covered by an object if the distance from the object does not exceed the sum of its sensitivity and the object's domination radius.
\fi

We are finally ready to define \covProb. As input we get an edge-weighted graph $G$ embedded on a sphere, families of objects $\Obj$ and clients $\Cli$ (described using locations, costs, radii, placements, sensitivities, and prizes), and a nonnegative integer $k$. For a subfamily $\SolObj\subseteq \Obj$, we define its {\em{revenue}}, denoted $\Pri(\SolObj)$, as the total sum of prizes of clients covered by at least one object from $\SolObj$ minus the total sum of costs of objects from $\SolObj$. In the \covProb problem, the task is to find a subfamily $\SolObj\subseteq \Obj$ such that the following holds:
\begin{enumerate}[label=(\roman*)]
\item\label{pr-ex} Family $\SolObj$ is normal and has cardinality {\em{exactly}} $k$.
\item\label{pr-mx} Subject to the previous constraint, family $\SolObj$ maximizes the revenue $\Pri(\SolObj)$.
\end{enumerate}
It can happen that there is no subfamily $\SolObj$ satisfying property~\ref{pr-ex}. In this case, value $\minf$ should be reported by the algorithm. For an instance $(G,\Obj,\Cli,k)$ of \covProb, we denote $\onum=|\Obj|$, $\cnum=|\Cli|$, and $n=|V(G)|$. Observe that we can assume that $k\leq \onum$ and that the input graph $G$ is simple: loops can be safely removed, and it is safe to only keep the edge with the smallest weight from any pack of parallel edges.

\iffull Note that if we do not provide any clients in the problem and we set all the radii to be equal to $0$, then we arrive exactly at the problem of packing disjoint objects in the graph. However, by introducing also clients we can ask for (partial) domination-type constraints in the graph, as well as define prize-collecting objectives. In this manner, \covProb generalizes both packing problems as well as covering problems. The caveat is, however, that we need to require that the objects that cover clients are pairwise disjoint, and moreover they have to form a normal family. 
\fi

The main result of this paper is the following theorem.

\begin{theorem}[Main result]\label{thm:main1}
\covProb can be solved in time $\onum^{\Oh(\sqrt{k})}\cdot (\cnum n)^{\Oh(1)}$.
\end{theorem}

\iffull
\subsection{Applications of Theorem~\ref{thm:main1}}\label{sec:appl}

In this section we provide formal argumentation of how the concrete results mentioned in the introduction follow from Theorem~\ref{thm:main1}. While for problems on planar graphs the reductions are straightforward, for geometric problems some non-trivial technicalities arise. We remark that we are proving the exact statements given in the introduction, but the generality of the \covProb problem would equally easily allow us to solve more general variants with different costs, prizes, sensitivities of clients, etc.

%In particular, we need to be careful about how floating-point precision issues are handled in the reductions.

\paragraph*{Applications to planar graphs: Theorems~\ref{thm:intro_independent}-\ref{thm:intro_covering2}.}

\restateintroindependent*
Theorem~\ref{thm:intro_independent} follows by taking $\Obj$ to be the set of objects (identified with their locations), assigning each of them radius and cost equal to $0$, and considering no clients (putting $\Cli=\emptyset$). As distances play no role in the problem, we can simply put unit weights.

\restateintrocovering*

For Theorem~\ref{thm:intro_covering}, for every $v\in D$ we construct an object $p_v$ with $\loc(p_v)=\{v\}$, $\cst(v)=0$ and $\rad(p_v)=\rad(v)$. Then, for every $u\in C$ we construct a client $q_u$ with $\pla(q_u)=u$, $\sen(q_u)=0$ and $\pri(q_u)=1$. Let $\Obj,\Cli$ be the sets of constructed objects and clients, respectively. We claim that the optimum value for the input instance is equal to the maximum among optimum revenues for instances $(G,\Obj,\Cli,k')$ of \covProb, for $0\leq k'\leq k$. On one hand, any solution to any such instance $(G,\Obj,\Cli,k')$ trivially yields a solution to the input instance with the same revenue. On the other hand, if $S\subseteq D$ is a solution to the input instance, then observe that without loss of generality we may assume that $S$ does not contain any two distinct vertices $v_1,v_2$ such that $\dist(v_1,v_2)\leq \rad(v_1)-\rad(v_2)$. This is because when $\dist(v_1,v_2)\leq \rad(v_1)-\rad(v_2)$, then any client covered by $v_1$ is also covered by $v_2$, and $v_1$ can be safely removed from $S$. Then a solution $S$ with this property corresponds to a subfamily $\SolObj\subseteq \Obj$ that is normal and has size $|S|\leq k$, i.e., it is a solution for $(G,\Obj,\Cli,|S|)$ with revenue equal to the number of vertices of $C$ covered by $S$. We conclude that to solve the input instance it suffices to solve all the instances $(G,\Obj,\Cli,k')$ of \covProb for $0\leq k'\leq k$, using the algorithm of Theorem~\ref{thm:main1}.

\restateintrocoveringx*

Theorem~\ref{thm:intro_covering2} follows by taking $\Obj$ to be the set of objects (identified with their locations), assigning each of them radius and cost equal to $0$, and for every $u\in C$ introducing a client $q_u$ with $\pla(q_u)=u$, $\sen(q_u)=0$ and $\pri(q_u)=1$. As distances play no role in the problem, we can simply put unit weights.

\paragraph*{Applications to geometric problems: Theorems~\ref{thm:intro_packingdisks}-\ref{thm:intro_coveringsquares}.}

In all the geometric reductions that follow, we perform arithmetic operations only of the following types:
\begin{itemize}
\item Computing new points on the plane with coordinates given as constant-size rational expressions (i.e., yielded by the four basic arithmetic operations) over the input coordinates and radii.
\item Computing Euclidean distances between the obtained points.
\end{itemize}
Note that finding intersection of two segments between pairs of points given on the input gives a point that is compliant with this requirement.

Thus, if the input coordinates and radii are $s$-bit integers, then the coordinates of the points obtained in the reductions can be stored as rationals represented using $\text{poly}(s)$ bits. Then the distances can be safely stored using $\text{poly}(s,n)$-bit precision in order to avoid rounding errors when adding at most $n$ of them. Similarly, if the input coordinates and radii are given as floating point numbers with $s$ bits before the point and $s$ bits after the point, then we can scale them up to get $2s$-bit integers, and perform the same analysis.

\restateintropackingdisks*

\begin{proof}% [Proof of Theorem~\ref{thm:intro_packingdisks}]
We provide a reduction to the problem considered in Theorem~\ref{thm:intro_independent}. Let $I=(\Obj,k)$ be the input instance. Based on the disk set $\Obj$, we define a planar graph $G$ as follows. First, construct a set of points $X$ (see Figure~\ref{fig:packingdisks}) by taking (a) all the centers of disks from $\Obj$, and (b) for any two disks $A_1,A_2$ that intersect, an arbitrarily chosen point in their intersection (for example, a point on the segment between the centers for which distances to the centers are in the same ratio as radii of the disks). Start the construction of $G$ by putting $V(G)=X$. Then, for each pair of points $x_1,x_2\in X$, draw a segment between $x_1$ and $x_2$ on the plane, provided that no other vertex of $X$ lies on this segment. These segments form so far the edge set of $G$, but they may cross. To make $G$ planar, for every crossing of two segments introduce a new vertex at the intersection point, and connect all the old and new points according to the segments. That is, every former edge between two points $x_1,x_2\in X$ becomes a path from $x_1$ to $x_2$ traversing consecutive crossing points, and the embedding of this path is exactly the segment between $x_1$ and $x_2$. To make $G$ edge-weighted, to every edge $uv\in E(G)$ assign weight equal to the Euclidean length of the segment between $u$ and $v$. Observe that $G$ is connected and $|V(G)|\leq \Oh(|\Obj|^4)$.
\begin{figure}
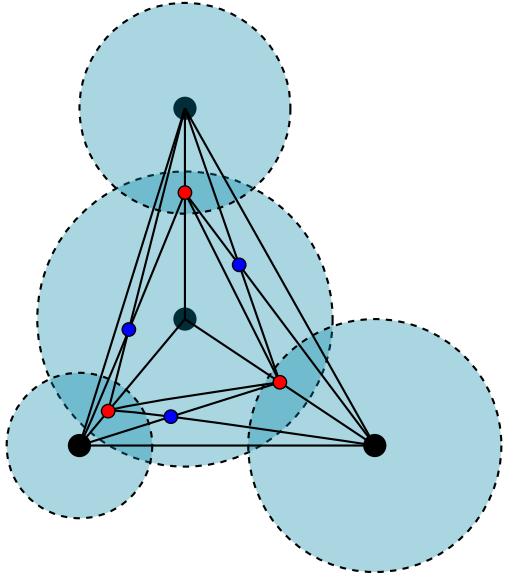

\begin{center}
\svg{0.4\linewidth}{packingdisks}
\end{center}
\caption{Proof of Theorem~\ref{thm:intro_packingdisks}. The black dots are the centers of the disks, the red dots are the points chosen in the intersections of pairs of disks, and the blue dots were introduced at the intersections of the segments to make the graph planar.}\label{fig:packingdisks}
\end{figure}

Now, for every disk $A\in \Obj$ with radius $r_A$ construct a vertex set $X(A)$ as follows: $X(A)$ comprises all the vertices of $G$ that are at distance at most $r_A$ from the center of $A$ (which belongs to $V(G)$) in the graph $G$. Obviously $G[X(A)]$ is connected, because it is defined as a ball in the graph $G$. Moreover, by the triangle inequality it follows that every vertex of $X(A)$ is actually embedded into the disk $A$. Note, however, that $X(A)$ may not contain some intersection points that are actually embedded into the disk $A$, but in $G$ they are at distance larger than $r_A$ from the center of $A$. Nonetheless, $X(A)$ contains all the elements of $X$ that are embedded into the disk $A$. This is because by the construction of $G$, for any $x\in X\cap A$ we have that $x$ is connected to the center of $A$ via a path consisting of parallel segments, which in particular has length equal to the Euclidean distance from $x$ to this center.

Let $\Obj_0=\{X(A)\ |\ A\in \Obj\}$ be the family of constructed vertex subsets, and let $I'=(G,\Obj_0,k)$ be the constructed instance of the problem considered in Theorem~\ref{thm:intro_independent}. We now claim that the input instance $I$ has a solution if and only if the constructed instance $I'$ has a solution.

On one hand, if $\mathcal{S}\subseteq \Obj$ is a subfamily of $k$ disjoint disks, then $\mathcal{S}_0=\{X(A)\ |\ A\in \mathcal{S}\}$ is a subfamily of $k$ vertex sets that are pairwise disjoint, and hence $I'$ has some solution. On the other hand, suppose there exists some solution $\mathcal{S}_0$ to $I'$, and let $\mathcal{S}\subseteq \Obj$ be the corresponding set of $k$ disks. We claim that these disks are pairwise disjoint. Suppose that, on the contrary, there are some distinct disks $A_1,A_2\in \mathcal{S}$ that intersect. Recall that for these disks we have added some point $x\in A_1\cap A_2$ to $X$. As we have argued, by the construction of $G$ this means that $x$ belongs both to $X(A_1)$ and to $X(A_2)$, which is a contradiction with the fact the vertex sets in $\mathcal{S}_0$ are pairwise disjoint. 
\end{proof}

\restateintropackingconvex*
\begin{proof}% [Proof of Theorem~\ref{thm:intro_packingconvex}]
We give a reduction to the problem considered in Theorem~\ref{thm:intro_independent}. Let $I=(\Obj,k)$ be the input instance. Based on the set of polygons $\Obj$, define a planar graph $G$ as follows. Start with taking $V(G)$ to be the set of all the vertices of all the polygons of $\Obj$, and then for every polygon $A\in \Obj$ draw all the sides of $A$ as edges of $G$, embedded as respective segments (see Figure~\ref{fig:packingconvex}). Then perform a similar construction as in the proof of Theorem~\ref{thm:intro_packingdisks}: in order to make $G$ planar, for every crossing of two segments introduce a vertex on their intersection. The edge set of $G$ is defined as the set of all the segments $xy$ between two vertices $x,y\in V(G)$, being either the original vertices of the polygons or the introduced crossing points, such that $xy$ is contained in some side of an original polygon and the interior of $xy$ does not contain any other vertex of $V(G)$. Thus, this definition yields a planar embedding of $G$. Note that the definition is also correct when we have two sides $xy$ and $x'y'$ of two different polygons that share a common subinterval (not being a single point); then vertices $x,y$ count as crossing points subdividing $x'y'$ and vice versa.
\begin{figure}
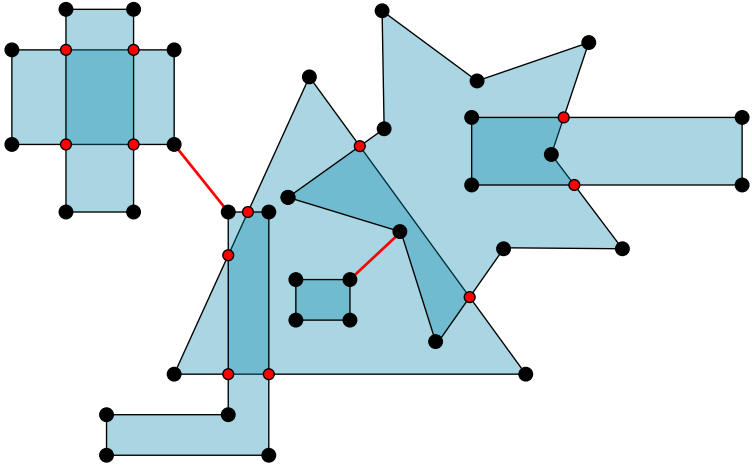

\begin{center}
\svg{0.6\linewidth}{packingconvex}
\end{center}
\caption{Proof of Theorem~\ref{thm:intro_packingconvex}. The black dots are the vertices of the polygons and the red dots are the points introduced at the intersection of segments. The red segments were introduced to make the graph connected.}\label{fig:packingconvex}
\end{figure}

Observe that $G$ defined in this way is planar, but may not be connected. Hence, we make $G$ connected by iteratively taking two connected components $C_1$, $C_2$ of $G$ that are incident to a common face, and adding an edge within this face between two arbitrarily chosen vertices from $C_1$ and $C_2$. Observe that $|V(G)|=\Oh(n^2)$.

Now, for every polygon $A\in \Obj$ construct a vertex set $X(A)$ as follows: $X(A)$ comprises all the vertices of $G$ that are embedded into the polygon $A$. We regard polygons as closed, so $X(A)$ in particular contains all the vertices of $A$. Note that the perimeter of $A$ naturally induces a simple cycle in $G$, and $X(A)$ is exactly the set of points enclosed or on this cycle. Since $G$ is connected, it also follows that $G[X(A)]$ is connected. Let $\Obj_0=\{X(A)\ |\ A\in \Obj\}$ be the set of constructed vertex sets, and let $I'=(G,\Obj_0,k)$ be the constructed instance of the problem considered in Theorem~\ref{thm:intro_independent}. We now claim that the input instance $I$ has a solution if and only if the constructed instance $I'$ has a solution.

On one hand, if $\mathcal{S}\subseteq \Obj$ is a subfamily of $k$ disjoint polygons, then $\mathcal{S}_0=\{X(A)\ |\ A\in \mathcal{S}\}$ is a subfamily of $k$ vertex sets that are pairwise disjoint, and hence $I'$ has some solution. On the other hand, suppose there exists some solution $\mathcal{S}_0$ to $I'$, and let $\mathcal{S}\subseteq \Obj$ be the corresponding set of polygons. We claim that these polygons are pairwise disjoint. Suppose that, on the contrary, there are some distinct polygons $A_1,A_2\in \mathcal{S}$ that intersect. If the perimeters of $A_1$ and $A_2$ intersect, then each their intersection point $x$ belongs to $V(G)$ and is contained in $X(A_1)\cap X(A_2)$, contradicting the fact that $X(A_1)$ and $X(A_2)$ are disjoint. Otherwise, either $A_1$ is entirely contained in $A_2$, or $A_2$ is entirely contained in $A_1$. In the former case every vertex of $A_1$ belongs to $V(G)$ and is contained in $X(A_1)\cap X(A_2)$, and in the latter case every vertex of $A_2$ has this property. In both cases we obtain a contradiction with the fact that $X(A_1)$ and $X(A_2)$ are disjoint.
\end{proof}

\restateintrocoveringdisks*
\restateintrocoveringsquares*

\begin{proof}% [Proof of Theorems~\ref{thm:intro_coveringdisks} and~\ref{thm:intro_coveringsquares}]
We prove both theorems at the same time by showing that the algorithm can be constructed whenever sets from $\Obj$ are defined as balls in some norm $||\cdot||$ on $\R^2$. More precisely, there exists a norm $||\cdot||$ on $\R^2$, such that every $A\in \Obj$ is defined as $A=\{x\in \R^2\ |\ ||x-c(A)||\leq r(A)\}$, where $c(A)$ is the center of ball $A$ and $r(A)$ is its radius. Then Theorem~\ref{thm:intro_coveringdisks} follows by taking $||\cdot||$ to be the $\ell_2$-norm and Theorem~\ref{thm:intro_coveringsquares} follows by taking $||\cdot||$ to be the $\ell_\infty$-norm.

We give a reduction to the problem considered in Theorem~\ref{thm:intro_covering}, and to this end we perform a similar construction as in the proof of Theorem~\ref{thm:intro_packingdisks}. Let $I=(\Obj,\Cli,k)$ be the input instance. Let $X$ be the set of all the centers of balls from $\Obj$ and all the points from $\Cli$. Start the construction of $G$ by putting $V(G)=X$. Then, for each pair of points $x_1,x_2\in X$ such that $x_1$ is a center of a ball from $\Obj$ and $x_2$ is a point from $\Cli$, draw a segment between $x_1$ and $x_2$ on the plane, provided that no other vertex of $X$ lies on this segment. These segments form so far the edge set of $G$, but they may cross. To make $G$ planar, for every crossing of two segments introduce a new vertex at the intersection point, and connect all the old and new points naturally according to the segments. To make $G$ edge-weighted, observe that every edge $e$ of $G$ is embedded as a segment on the plane, so put $\wei(e)$ to be the $||\cdot||$-length of this segment. Observe that $G$ is connected and $|V(G)|=\Oh(|\Obj|^2|\Cli|^2)$.
\begin{figure}
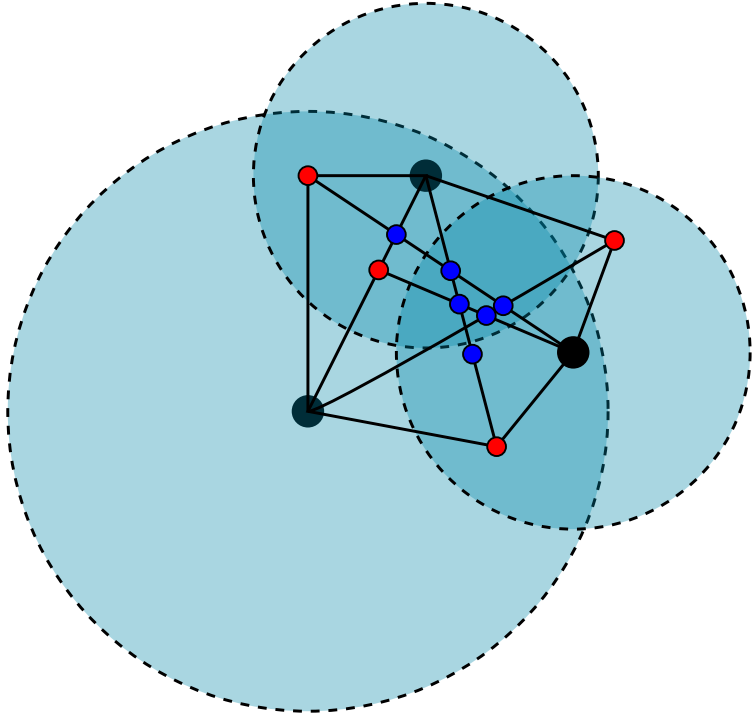

\begin{center}
\svg{0.6\linewidth}{coveringdisks}
\end{center}
\caption{Proof of Theorem~\ref{thm:intro_coveringdisks}. The black dots are the centers of the disks in $\Obj$, the red dots are in the points in $\Cli$, and the blue dots were introduced at the intersections of the segments to make the graph planar.}\label{fig:coveringdisks}
\end{figure}

Now, let us set $D=\{c(A)\ |\ A\in \Obj\}$ and $C=\Cli$. Construct an instance $I'$ of the problem considered in Theorem~\ref{thm:intro_covering} by taking graph $G$ with subsets of vertices $D$ and $C$, assigning $\rad(c(A))=\rad(A)$ for each $c(A)\in D$, and putting budget $k$. We now claim that a vertex $x\in C$ is covered by some $c(A)\in D$ in the instance $I'$ if and only if point $x\in \Cli$ belongs to $A\in \Obj$ in the input instance $I$. To prove this it suffices to show that $\dist_G(x,c(A))=||x-c(A)||$, since $x$ is covered by $c(A)$ in $I'$ if and only if $\dist_G(x,c(A))\leq \rad(A)$, whereas $x$ belongs to $A$ in $I$ if and only if $||x-c(A)||\leq \rad(A)$. First, by the triangle inequality it follows that distances in $G$ are lower bounded by distances w.r.t. norm $||\cdot||$ on the plane, so $\dist_G(x,c(A))\geq ||x-c(A)||$. On the other hand, since $x,c(A)\in X$, then by the construction of $G$ there exists a path in $G$ that connects $x$ and $c(A)$ and consists of parallel segments. In particular, this path has total length $||x-c(A)||$, which means that $\dist_G(x,c(A))\leq ||x-c(A)||$. We conclude that indeed $\dist_G(x,c(A))=||x-c(A)||$.

Solutions to the input instance $I$ correspond one-to-one to solutions of the constructed instance $I'$, and the argumentation of the previous paragraph shows that this correspondence preserves the number of points covered by the solution. Hence, it suffices to run the algorithm of Theorem~\ref{thm:intro_covering} on the instance $I'$.
\end{proof}

%%% Local Variables: 
%%% mode: latex
%%% TeX-master: "voronoi"
%%% End: 

\fi

%%% Local Variables: 
%%% mode: latex
%%% TeX-master: "voronoi"
%%% End: 

\ifabstract
\clearpage
\bibliographystyle{abbrv}
\bibliography{voronoi}
\appendix
\fi

\iffull

\section{The main algorithm: proof of Theorem~\ref{thm:main1}}\label{sec:algo}

\subsection{Notation and general definitions}

We first establish notation and recall some general definitions that will be used throughout the proof.

For a graph $G$, by $V(G)$ we denote its vertex set and by $E(G)$ we denote its edge set. Graph $G$ is {\em{simple}} if every edge connects two different vertices, and every pair of vertices is connected by at most one edge. On the other hand, $G$ is a {\em{multigraph}} if we allow (a) {\em{parallel edges}}, that is, multiple edges connecting the same pair of vertices, and (b) {\em{loops}}, that is, edges that connect some vertex $u$ with itself. A graph {\em{induced}} by a vertex subset $X\subseteq V(G)$, denoted $G[X]$, has $X$ as the vertex set, and its edge set comprises all the edges of $G$ whose both endpoints belong to $X$. A graph {\em{spanned}} by an edge subset $F\subseteq E(G)$, denoted $G[F]$, has $F$ an the edge set, and its vertex set comprises all the vertices of $G$ incident with at least one edge of $F$.

In this work we will be mostly working with edge-weighted graphs. That is, we assume that the input graph $G$ is given together with a positive edge weight function $\wei\colon E(G)\to \Rp$. The {\em{distance}} between two vertices $u,v\in V(G)$, denoted $\dist(u,v)$, is defined as the minimum total weight of a path connecting $u$ and $v$. We extend this notation to subsets in the obvious manner: $\dist(u,Y)$ is the minimum weight of a path from $u$ to any vertex of $Y$, and $\dist(X,Y)$ is the minimum weight of a path from any vertex of $X$ to any vertex of $Y$. Note that the $\dist(\cdot,\cdot)$ function satisfies the triangle inequality: $\dist(u,v)+\dist(v,w)\geq \dist(u,w)$, and similar inequalities hold for sets as well.

Let $\Sigma$ be the sphere $S^2=\{(x,y,z)\in \R^3\ |\ x^2+y^2+z^2=1\}$. A {\em{curve}} on $\Sigma$ is a homeomorphic image of the interval $[0,1]$, and a {\em{closed curve}} on $\Sigma$ is a homeomorphic image of the circle $S^1=\{(x,y)\in \R^2\ |\ x^2+y^2=1\}$. Note that, in particular, a curve and a closed curve do not have self-crossings, i.e., no point of $\Sigma$ is visited twice. The Jordan curve theorem states that for any closed curve $\gamma$ on $\Sigma$, $\Sigma\setminus \gamma$ consists of two connected sets homeomorphic to open disks. A {\em{sphere embedding}} of a graph $G$ is a mapping that maps vertices of $G$ to distinct points of $\Sigma$, and edges of $G$ to curves connecting respective endpoints that do not intersect apart from their endpoints. In case $e\in E(G)$ is a loop, the corresponding curve connects a point with itself, i.e., it is a closed curve. Note that this definition is also valid for multigraphs. It is well known that a graph is planar, i.e., it can be embedded on the $2$-dimensional plane in the same manner, if and only if it can be embedded on a sphere.

If $G$ is a sphere-embedded (multi)graph, then a {\em{face}} of $G$ is an inclusion-wise maximal connected set that does not contain any point of the embedding of $G$. Note that in particular every face is open. It is well known that if $G$ is connected, then every face is homeomorphic to an open disk. Note, however, that if $G$ has bridges or cutvertices, then there may be faces whose closures are not homeomorphic to closed disks. If the closure of each face is homeomorphic to a closed disk, then we say that the embedding is a {\em{$2$-cell embedding}}. The set of faces of a (sphere-embedded) graph $G$ is denoted by $F(G)$. The well-known Euler formula states that for every connected multigraph $G$ embedded on a sphere, it holds that $|V(G)|-|E(G)|+|F(G)|=2$. 

Let $f\in F(G)$ be a face of a connected, sphere-embedded multigraph $G$. We say that a vertex $v$ or an edge $e$ is {\em{incident}} with $f$ if it is contained in the closure of $f$. By the {\em{boundary of $f$}}, denoted $\bnd f$, we mean a walk in $G$ that visits consecutive vertices and edges incident with $f$ in the order of their appearance around $f$. Note that if the closure of $f$ is homeomorphic to a closed disk, then $\bnd f$ is a simple cycle in $G$. However, otherwise a vertex can be visited more than once on $\bnd f$ (this can happen if it is a cutvertex in $G$), and an edge can be traversed more than once on $\bnd f$ (this can happen if it is a bridge in $G$). In case $v$ is a cutvertex of $G$, then there is a face $f$ of $G$ that appears multiple times around $v$ in the planar embedding. Whenever some curve is leaving $v$ to $f$ (or entering $v$ from $f$), we will distinguish these different appearances and call then {\em{directions}}. Note that every such direction corresponds to an appearance of $v$ on $\bnd f$.

A sphere-embedded graph is {\em{triangulated}}, if for every face $f$, walk $\bnd f$ is a triangle in $G$; in particular, the graph is simple and the embedding is a $2$-cell embedding. It is well-known that for any simple graph $G$ embedded on a sphere, one can add edges to the embedding so that the new supergraph is triangulated.

\subsection{Simplifying assumptions}

Throughout the proof we assume that $(G,\Obj,\Cli,k)$ is the input instance of \covProb. Before we proceed to the proof, let us make a few simplifying assumptions about the input instance $(G,\Obj,\Cli,k)$. These assumptions can be made without loss of generality, and they will streamline further reasonings.

Firstly, we will assume that the graph $G$ is connected. Indeed, for disconnected graph we can apply the algorithm for every connected component of $G$ separately, for all the parameters $i=0,1,2,\ldots,k$, and then merge the results using a simple dynamic programming algorithm.

Secondly, we will assume that all the values $\dist(u,v)$ and $\dist(u,v)-\rad(p)$ for $u,v\in V(G)$ and $p\in \Obj$ are pairwise different, and moreover that the shortest paths between pairs of vertices in $G$ are unique. This can be easily obtained using the standard technique of breaking ties by adding very small, distinct values to the weights of edges as well as to the radii of objects. To do this, we need to use $(\onum+n)^{\Oh(1)}$ more bits of precision in the representation of floating point numbers. Since we never estimate precisely the exact running time of polynomial-time subroutines run on $G$, this does not influence the claimed asymptotic running time of the algorithm. For brevity, we will denote this assumption by \ass. 

Thirdly, we fix some sphere embedding of $G$, and without loss of generality we assume that $G$ is triangulated. This can be achieved by triangulating every face of $G$ using edges of weight $\pinf$ (equivalently, a value much larger than all the radii and sensitivities). After this operation, every face of $G$ is a triangle whose closure is homeomorphic to a closed disc (i.e., it is a $2$-cell embedding). Note that thus $G$ is $2$-connected, so in particular it does not have any cut-vertices, bridges, or vertices of degree $1$. We also assume that $G$ has more than $3$ vertices, since we can always add a new vertex within a triangular face connected with its vertices via edges of weight $\pinf$.

Finally, for every object $p\in \Obj$ we arbitrarily distinguish one vertex $\cen(p)\in \loc(p)$, called the {\em{center}} of object $p$. We also pick an arbitrary spanning tree $\tree(p)$ of $G[\loc(p)]$, rooted at $\cen(p)$. For concreteness, the reader may think of $\tree(p)$ as the tree of shortest paths from $\cen(p)$ in $G[\loc(p)]$, but we will never use this property.

\subsection{Voronoi partitions}\label{sec:partitions}

We want to define analogs of Voronoi regions and related notions in the context of planar graphs: given a set $\Fam$ of objects, we would like to partition the vertices of the graph according to which object in $\Fam$ is closest to a vertex. For this to make sense, we clearly need the assumption that the objects in $\Fam$ are pairwise disjoint, otherwise it is not clear how to classify vertices in the intersection of the objects. Moreover, to express the fact that different objects have different power of domination ($\rad(p)$ can be different for different objects), we need additively weighted Voronoi regions. That is, instead of comparing the distances to locations of objects from $\Fam$, we compare the values of distances to $\loc(p)$ minus $\rad(p)$, for $p\in \Fam$. Here we need the technical assumption that the family $\Fam$ is normal, which rules out degenerate situations by ensuring that every point of the object is in the Voronoi region of the object itself, and hence the Voronoi regions are nonempty and connected.

Formally, suppose $\Fam\subseteq \Obj$ is a normal subfamily of objects. Then we define the {\em{Voronoi partition}} of $G$ with respect to $\Fam$, denoted $\Vorpar_{\Fam}$, as a partition of $V(G)$ into parts $\{M_p\}_{p\in \Fam}$ as follows: a vertex $v\in V(G)$ belongs to part $M_{p_0}$ if and only if $\dist(v,\loc(p_0))-\rad(p_0)$ is the smallest value among $\{\dist(v,\loc(p))-\rad(p)\}_{p\in \Fam}$. Note that assumption \ass implies there can be no ties between different objects. Set $M_p$ shall be called the {\em{Voronoi region}} of $p$. For some Voronoi partition $\Vorpar_\Fam$ and $p\in \Fam$, the Voronoi region of $p$ will be also denoted by $\Vorpar_\Fam(p)$. Whenever the family we are working with will be clear from the context, we shall drop the index $\Fam$. This remark holds for all the objects defined in the sequel that depend on the currently considered family $\Fam$.

The following simple lemma shows that the normality condition is a sanity check that the Voronoi regions are well-defined.

\begin{lemma}\label{lem:sanity}
Suppose $\Fam\subseteq \Obj$ is a normal subfamily of objects and let $\Vorpar=\Vorpar_\Fam$ be the Voronoi partition w.r.t. $\Fam$. Then for every $p\in \Fam$, it holds that $\loc(p)\subseteq \Vorpar(p)$.
\end{lemma}
\begin{proof}
For the sake of contradiction, suppose there are two distinct objects $p_1$ and $p_2$ such that some vertex $v\in \loc(p_1)$ belongs to $\Vorpar(p_2)$. Then in particular we have that
$$\dist(v,\loc(p_2))-\rad(p_2)<\dist(v,\loc(p_1))-\rad(p_1)=-\rad(p_1),$$
and hence 
$$\dist(\loc(p_1),\loc(p_2))\leq \dist(v,\loc(p_2)) < \rad(p_2)-\rad(p_1).$$
This is a contradiction with the definition of normality.
\end{proof}

The following simple claim shows that Voronoi regions behave as expected.

\begin{lemma}\label{lem:connectivity}
Suppose $\Fam\subseteq \Obj$ is a normal subfamily of objects and let $\Vorpar=\Vorpar_\Fam$ be the Voronoi partition w.r.t. $\Fam$. Then for every $p\in \Fam$ and every $v\in \Vorpar(p)$, the unique shortest path from $v$ to $\loc(p)$ is entirely contained in $\Vorpar(p)$. In particular, $G[\Vorpar(p)]$ is connected.
\end{lemma}
\begin{proof}
Let $P$ be the shortest path connecting $v$ with some $w\in \loc(p)$. Note that the uniqueness of $P$ follows from \ass. For the sake of contradiction, suppose there exists $v'\in V(P)$ such that $\dist(v',\loc(p'))-\rad(p')<\dist(v',\loc(p))-\rad(p)$, for some $p'\in \Fam$, $p'\neq p$. Since $P$ is the shortest path, we have:
\begin{eqnarray*}
\dist(v,\loc(p))-\rad(p) & = & \dist(v,w)-\rad(p) \\
& = & \dist(v,v')+\dist(v',w)-\rad(p) \\
& = & \dist(v,v')+\dist(v',\loc(p))-\rad(p) \\
& > & \dist(v,v')+\dist(v',\loc(p'))-\rad(p') \geq \dist(v,\loc(p'))-\rad(p').
\end{eqnarray*}
This is a contradiction with $v\in \Vorpar(p)$. The second part of the claim follows from the first part, Lemma~\ref{lem:sanity}, and the fact that $G[\loc(p)]$ is connected.
\end{proof}

Lemma~\ref{lem:connectivity} gives rise to the following definition of a spanning tree of a Voronoi region. For a normal subfamily of objects $\Fam\subseteq \Obj$ and some $p\in \Fam$, we shall distinguish a spanning tree $\extree_\Fam(p)$ of $G[\Vorpar(p)]$ constructed as follows:
\begin{itemize}
\item start with $\extree_\Fam(p)=\tree(p)$;
\item for every $v\in \Vorpar(p)$, add to $\extree_\Fam(p)$ the unique shortest path from $v$ to $\loc(p)$.
\end{itemize}
By Lemma~\ref{lem:connectivity} and properties of shortest paths, $\extree_\Fam(p)$ defined in this manner is a spanning tree of $G[\Vorpar(p)]$. Note that we may view $\extree_\Fam(p)$ as rooted in $\cen(p)$, the center of $p$. However, $\extree_\Fam(p)$ is {\em{not}} necessarily the tree of shortest paths from $\cen(p)$ in $G[\Vorpar(p)]$; for example, it is possible that the shortest path between two vertices of $\loc(p)$ is not contained fully in $\loc(p)$.

The following statement is straightforward, but it is instructive to state it explicitly.

\begin{lemma}\label{lem:inclusion}
Let $\Fam'\subseteq \Fam\subseteq \Obj$ be two normal families of objects. Then for every $p\in \Fam'$ it holds that $\Vorpar_{\Fam'}(p)\supseteq \Vorpar_\Fam(p)$ and $\extree_{\Fam'}(p)$ is a supergraph of $\extree_\Fam(p)$.
\end{lemma}

In other words, when we restrict the family of objects, then the Voronoi regions and their spanning trees can only grow.

\subsection{Voronoi (pre-)diagrams}\label{sec:diagrams}

Let us fix some normal subfamily of objects $\Fam$, let $\ell=|\Fam|$, and let $\Vorpar=\Vorpar_\Fam$ be the corresponding Voronoi partition. Intuitively, $\Vorpar$ partitions $G$ into regions, such that if we put an edge between adjacent regions, then the resulting graph is planar. We now formalize this concept by introducing {\em{Voronoi prediagrams}} and {\em{Voronoi diagrams}}. More precisely, Voronoi (pre-)diagrams correspond to the dual of the adjacency graph between the regions.

A {\em{Voronoi prediagram}} for $\Fam$ is a subgraph $\pdgm_\Fam$ of $G^*$, the dual of $G$, defined as follows:
\begin{enumerate}
\item Start with the graph $G^*$.
\item For every $p\in \Fam$, remove all the dual edges corresponding to edges of $\extree(p)$.
\item Iteratively remove vertices of degree $1$ in the obtained graph, up to the point when the minimum degree of the graph is $2$.
\end{enumerate}
In this definition we drop the weights of edges, i.e., $\pdgm$ is an unweighted graph. See Figure~\ref{fig:prdg-const} for an example. The following lemma formalizes the properties of the Voronoi prediagram.

\begin{figure}[t]
                \centering
                \def\svgwidth{0.7\columnwidth}
                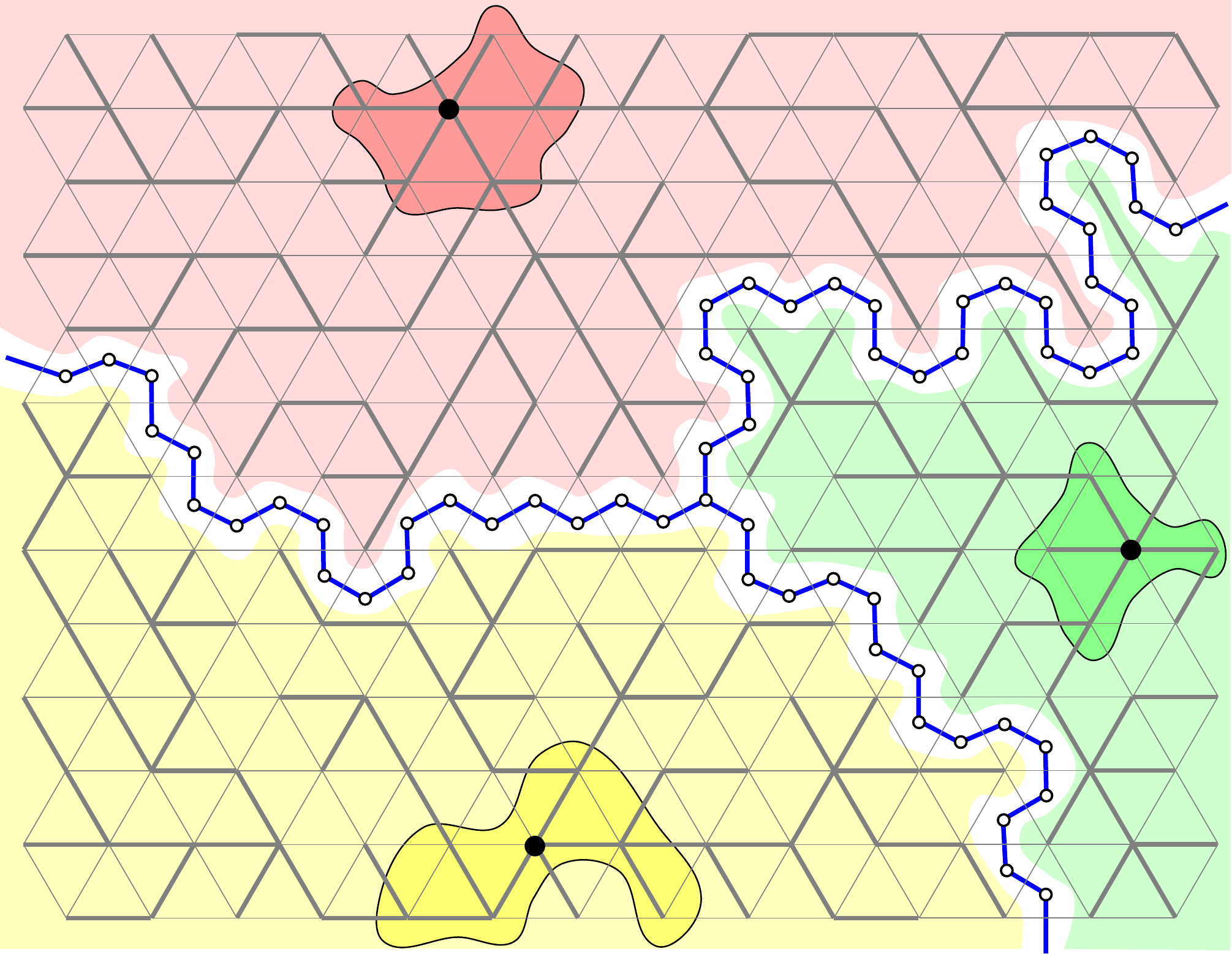
\caption{Construction of the prediagram in an area between three objects: red, yellow, and green. The Voronoi regions are depicted in light red, light yellow, and light green. Gray edges belong to the primal graph $G$, and thick gray edges belong to spanning trees $\extree_{\Fam}(p)$. The white vertices and thick blue edges belong to the prediagram $\pdgm_{\Fam}$.}\label{fig:prdg-const}
\end{figure}

\begin{lemma}\label{lem:pdiagram}
The prediagram $\pdgm$ is a simple, connected graph embedded on the sphere, where every vertex is of degree $2$ or $3$. Moreover, $\pdgm$ has exactly $\ell$ faces, and faces of $\pdgm$ correspond to objects of $\Fam$ in the following sense: object $p$ is associated with a face $f^*_p\in F(\pdgm)$ such that $f^*_p$ contains all the vertices of $\Vorpar(p)$, and no other vertex.
\end{lemma}
\begin{proof}
Observe that $\pdgm$ is a subgraph of $G^*$, so it inherits a sphere embedding from $G^*$. The fact that $\pdgm$ is simple follows from the fact that $G^*$ is simple, since $G$ is triangulated, simple, and has more than $3$ vertices. Also, the fact that $G$ is triangulated implies that $G^*$ is a $3$-regular graph. Hence $\pdgm$, as a subgraph of $G^*$, has maximum degree at most $3$. The fact that $G^*$ has minimum degree at least $2$ follows from the construction.

We now show that $\pdgm$ is connected. Obviously $G^*$ is connected, as a dual of a sphere-embedded graph. We now argue that during the step when the dual edges to the edges of some $\extree(p)$ are removed, we also cannot spoil the connectivity of $\pdgm$. Let $W$ be a walk in $G^*$ that traverses the set of faces incident to vertices and edges of $\extree(p)$, but does not use any dual edge to an edge of $\extree(p)$; $W$ is constructed by traversing $\extree(p)$ around, with respect to its planar embedding of $G$. All the edges used $W$ are not removed when $\extree(p)$ is removed, and could not have been removed when some previous $\extree(p')$ was removed. Hence, walk $W$ remains in the graph after $\extree(p)$ is removed. Since the removed dual edges of $\extree(p)$ always connect two faces visited by $W$, we infer that all the vertices of $G^*$ that could have been possibly disconnected by removal of $\extree(p)$, are actually still connected by $W$. Therefore, after the second step of the construction of $\pdgm$ the graph is still connected. Since removal of a degree-$1$ vertex does not spoil connectivity, we infer that at the end $\pdgm$ is indeed connected.

We are left with proving the correspondence between objects of $\Fam$ and faces of $\pdgm$. Since $\pdgm$ is a subgraph of $G^*$, every face of $\pdgm$ consists of a union of a nonempty set of original faces of $G^*$, which moreover were connected in $G^*$. Every face of $G^*$ corresponds to vertex of $G$, which means that every face $f^*$ of $\pdgm$ contains a nonempty subset of vertices $X(f^*)$, which moreover is connected in $G$. Note now that for every $p\in \Fam$, removal of the dual edges of $\extree(p)$ merges all the faces of $G^*$ corresponding to vertices of $\Vorpar(p)$, which means that vertices of every Voronoi region $\Vorpar(p)$ are contained in one face $f^*_p\in F(\pdgm)$. It remains to show that no face $f^*\in F(\pdgm)$ can accommodate more than one Voronoi region. If this was the case, then, since $X(f^*)$ is connected in $G$, there would be an edge $uv\in E(G)$, $u,v\in X(f^*)$, that would connect two vertices from different Voronoi regions, and whose dual was removed in the construction of $\pdgm$. However, in the construction of $\pdgm$, in the first step we removed only duals to the edges that connect two vertices of the same Voronoi region, and in the second step we removed only vertices of degree $1$, which cannot merge two faces of the graph.
\end{proof}

Lemma~\ref{lem:pdiagram} shows that the prediagram $\pdgm$ is a suitable skeleton of the final diagram. More precisely, $\pdgm$ consists of vertices of degree $3$ that are connected using long paths consisting of vertices of degree $2$. From now on we assume that $\ell>2$. If $\ell=1$ then $\pdgm$ is empty, and if $\ell=2$ then $\pdgm$ is a cycle. Both these cases are degenerate in the argumentation to follow.

Given the prediagram $\pdgm=\pdgm_\Fam$, we define the {\em{Voronoi diagram}} $\dgm=\dgm_\Fam$ as follows. To obtain $\dgm$ from $\pdgm$, we iteratively take a vertex of degree $2$ in the graph, remove it, and add a new edge between its neighbors. Thus, all the paths consisting of vertices of degree $2$ in $\pdgm$ are shortened to single edges. Note, that in this manner $\dgm$ may become a multigraph: 
\begin{itemize}
\item There may appear multiple edges between a pair of vertices $u$ and $v$, in case there was more than one path of vertices of degree $2$ connecting $u$ and $v$ in $\pdgm$.
\item There may appear a loop at some vertex $u$, in case there was a path of vertices of degree $2$ traveling from $u$ back to itself.
\end{itemize}
In the following, we assume that a loop at vertex $u$ contributes with $2$ to the degree of $u$. The following lemma formalizes the properties of the Voronoi diagram $\dgm$.

\begin{lemma}\label{lem:diagram}
If $\ell>2$, then the diagram $\dgm$ is a connected $3$-regular multigraph embedded on a sphere. Moreover, $\dgm$ has exactly $2\ell-4$ vertices, $3\ell-6$ edges, and $\ell$ faces, and the set of faces of $\dgm$ is exactly the same as the set of faces of $\pdgm$.
\end{lemma}
\begin{proof}
The diagram $\dgm$ is constructed from $\pdgm$ by taking every path $P$ in $\pdgm$ that begins and ends in vertices $u,v$ of degree $3$ but travels only through vertices of degree $2$, and replacing it with a single edge $uv$ (possibly $u=v$). The embedding of $\dgm$ into the sphere is defined as follows: The vertices of $V(\dgm)$ are embedded exactly in the same manner as in the embedding of $\pdgm$. Whenever a path $P$ is replaced with an edge $uv$ between the endpoints of $P$, then the embedding of $uv$ is defined as the union of the embeddings of the edges of $P$. Thus, the faces of $\dgm$ are exactly the same as faces of $\pdgm$.

When constructing $\dgm$, we removed all the vertices of degree $2$, while all the vertices of degree $3$ keep their degrees. Hence $\dgm$ is $3$-regular. Moreover, since $\pdgm$ was connected (Lemma~\ref{lem:pdiagram}), then by the construction so is $\dgm$. 

The fact that $|F(\dgm)|=\ell$ follows from Lemma~\ref{lem:pdiagram} and the fact that the faces of $\dgm$ are exactly the same as faces of $\pdgm$. The conclusion that $|V(\dgm)|=2\ell-4$ and $|E(\dgm)|=3\ell-6$ follows from the fact that $\dgm$ is connected and $3$-regular, and a simple application of the Euler's formula, which is valid also for connected multigraphs.\end{proof}

In the following, the vertices of $V(\dgm)$ will be also call the {\em{branching points}} of diagram $\dgm$. The following auxiliary lemma shows that loops of the diagram $\dgm$ have very simple structure.

\begin{figure}[t]
        \centering
        \subfloat[Voronoi partition $\Vorpar$ and the prediagram $\pdgm$]{
                \centering
                \def\svgwidth{0.45\columnwidth}
                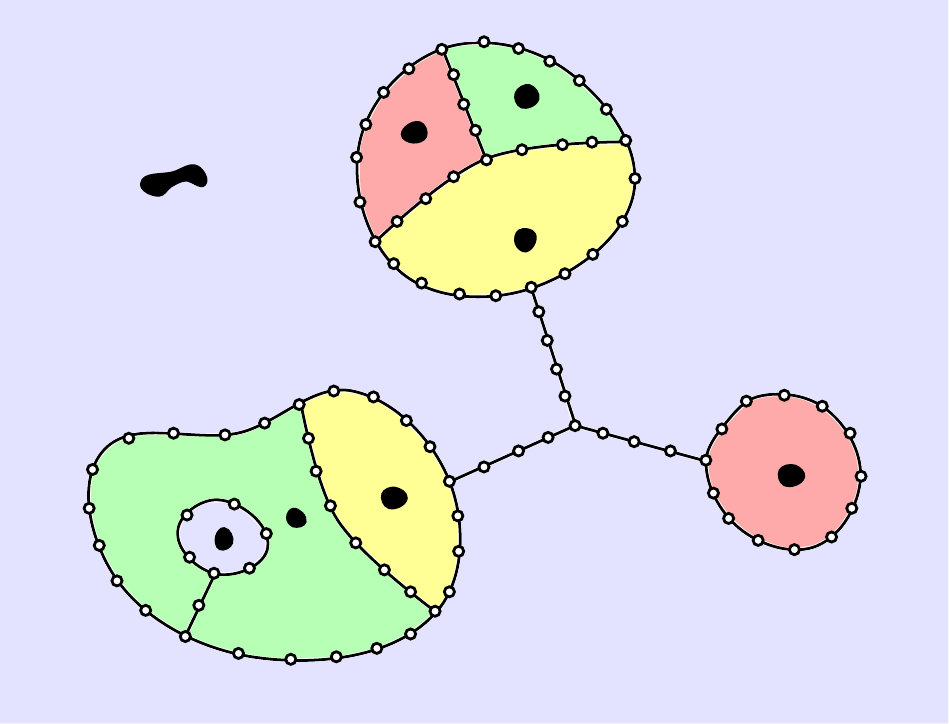
        }
        \qquad
        \subfloat[Diagram $\dgm$]{
                \centering
                \def\svgwidth{0.45\columnwidth}
                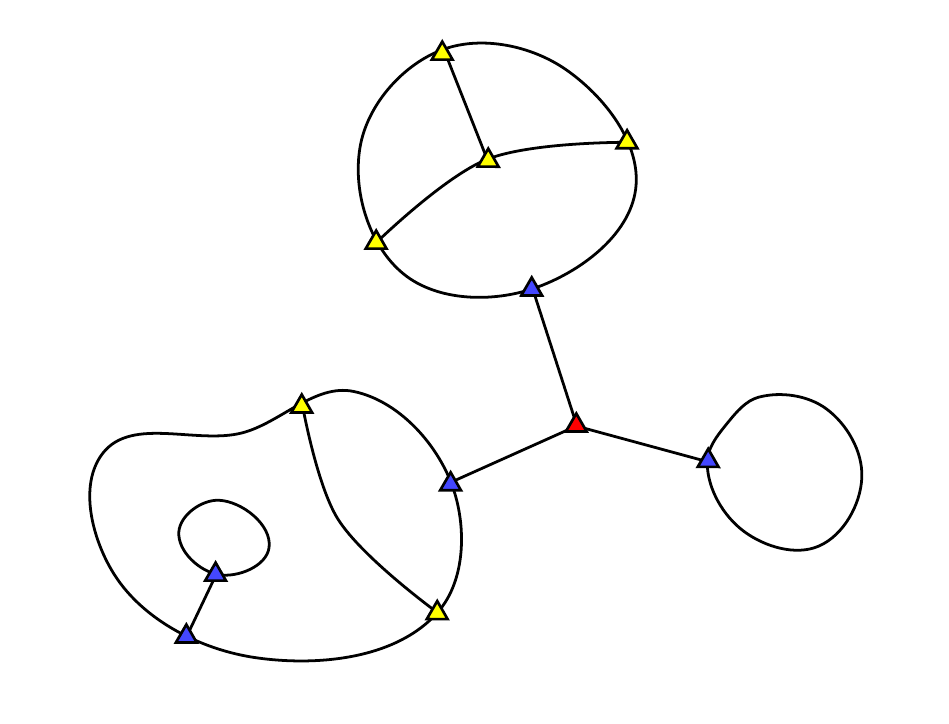
        }
\caption{Example of a Voronoi partition $\Vorpar_{\Fam}$ and the corresponding prediagram and diagram. Objects of $\Fam$ are denoted in black. Note that vertices in the prediagram $\pdgm_{\Fam}$ are faces of the original graph $G$. Branching points of the diagram $\dgm_{\Fam}$ are depicted in yellow, blue, and red, depending whether they correspond to type-$1$, -$2$, or -$3$ singular face according to definitions in Section~\ref{sec:impfaces}.}\label{fig:diagram}
\end{figure}

\begin{lemma}\label{lem:3-reg-loops}
Suppose $H$ is a connected $3$-regular multigraph embedded on a sphere. Then for every loop in $H$, one of the disks in which the loop separates the sphere is a face of $H$.
\end{lemma}
\begin{proof}
Consider a loop $e$ at vertex $u$ in $H$. Since $H$ is $3$-regular and loop $e$ contributes with $2$ to the degree of $u$, then $u$ is incident to one more edge $e'$. Then $e'$ is embedded into one of the disks into which $e$ partitions the sphere, and by the connectivity of $H$ we infer that the whole rest of $H$ is embedded into the same disk. The second disk contains therefore no point of the embedding of $H$, and hence is a face of $H$.
\end{proof}

\subsection{Important faces}\label{sec:impfaces}

In the algorithm we consider a large family of long sequences of faces of $G$, in hope of capturing a sequence that forms a short balanced separator of the Voronoi diagram imposed by the solution, on which the recursion step can be employed. Since balanced separators of this diagram can have width as large as $\Oh(\sqrt{k})$, we need to take into consideration all the sequences of at most this length. Of course, we could consider every face of $G$ as a possible candidate to be used in a separator, but then we would have $n^{\Oh(\sqrt{k})}$ candidate separators, which is too much for the running time promised in Theorem~\ref{thm:main1}.

In this section we prove the following: one can enumerate in polynomial time a family of at most $\onum^{4}$ faces of $G$, called further {\em{important faces}}, such that the candidates for faces used in separators can be safely picked from this family. Thus, we arrive at $\onum^{\Oh(\sqrt{k})}$ candidates for separators, instead of the original $n^{\Oh(\sqrt{k})}$. Formally, this section is devoted to the proof of the following result.

\begin{theorem}\label{thm:impFaces}
There exists a family $\impFaces\subseteq F(G)$ such that
\begin{itemize}
\item $|\impFaces|\leq \onum^4$;
\item for every normal subfamily of objects $\Fam\subseteq \Obj$, it holds that $V(\dgm_{\Fam})\subseteq \impFaces$.
\end{itemize}
Moreover, family $\impFaces$ can be constructed in time $\Oh(\onum^4\cdot n^{\Oh(1)})$.
\end{theorem}

Let us remark that Theorem~\ref{thm:impFaces} is needed {\em{only}} for the purpose of reducing the running time --- if we just allowed all the faces of $G$ to be used to build separators, then we would obtain an algorithm solving \covProb in time $\onum^{\Oh(\sqrt{k})}\cdot \cnum^{\Oh(1)}\cdot n^{\Oh(\sqrt{k})}$. In many of our applications in Section~\ref{sec:appl}, it actually holds that $n=\onum^{\Oh(1)}$, and hence we would obtain the same asymptotic running time without the speed-up due to the results of this section. Therefore, the result on enumerating a small family of important faces can be considered as optional, and the reader interested only in a $\onum^{\Oh(\sqrt{k})}\cdot \cnum^{\Oh(1)}\cdot n^{\Oh(\sqrt{k})}$ time algorithm is advised to skip this section.

We now move on to the proof of Theorem~\ref{thm:impFaces}, which spans the whole this section. However, before we proceed, we need some definitions and auxiliary lemmas. Let us fix some normal family of objects $\Fam\subseteq \Obj$, and let $\Vorpar=\Vorpar_\Fam$ be the corresponding Voronoi partition. For two distinct objects $p_1,p_2\in \Fam$, we say that Voronoi regions $\Vorpar(p_1)$ and $\Vorpar(p_2)$ {\em{meet}} at face $f$ if $\bnd f$ contains both a vertex of $\Vorpar(p_1)$ and a vertex of $\Vorpar(p_2)$. 

We say that face $f$ is a {\em{type-$1$ singular face}} for $\Fam$, if for a triple of distinct objects $p_1,p_2,p_3\in \Fam$, the Voronoi regions $\Vorpar(p_1),\Vorpar(p_2),\Vorpar(p_3)$ pairwise meet at $f$. Then we say that triple $(p_1,p_2,p_3)$ {\em{certifies}} that $f$ is a type-$1$ singular face. The following lemma shows that there are generally not so many type-$1$ singular faces.

\begin{lemma}\label{lem:type-1}
Suppose $\{p_1,p_2,p_3\}=\Fam\subseteq \Obj$ is a normal family of three objects. Then there are at most $2$ faces of $G$ that are type-$1$ singular faces for $\Fam$.
\end{lemma}
\begin{proof}
Let $M_i=\Vorpar_\Fam(p_i)$ for $i=1,2,3$. For the sake of contradiction, suppose that there are $3$ different faces $f_1,f_2,f_3$, such that $M_1,M_2,M_3$ all pairwise meet at each of $f_1,f_2,f_3$. Construct a planar graph $H$ as follows:
\begin{itemize}
\item For each $i=1,2,3$, introduce a vertex $w_i$ inside face $f_i$, and connect $w_i$ to all the vertices of $\bnd f_i$.
\item For each $j=1,2,3$, contract the whole subgraph $G[M_j]$ to a single vertex $m_j$. Note that this is possible since $G[M_j]$ is connected by Lemma~\ref{lem:connectivity}.
\item Remove all the vertices and edges of the graph apart from $\{w_i\colon i=1,2,3\}$, $\{m_j\colon j=1,2,3\}$, and the edges between them.
\end{itemize}
Observe that the obtained graph $H$ is isomorphic to $K_{3,3}$. After performing the first step of the construction we still had a graph with a sphere embedding, since the new vertices and edges could be drawn inside faces $f_i$. Then, we applied only edge contractions, vertex deletions and edge deletions. Hence, we have found a $K_{3,3}$ as a minor of a planar graph, a contradiction.
\end{proof}

Now we define the second type of singular faces. We say that face $f$ is a {\em{type-$2$ singular face}} for $\Fam$, if there exists a triple of distinct objects $p_1,p_2,p_3\in \Fam$ satisfying the following. One vertex $v_1$ of $\bnd f$ belongs to $\Vorpar(p_1)$ and two vertices $v_2,v_3$ of $\bnd f$ belong to $\Vorpar(p_2)$. Moreover, if we take a closed walk $W$ in $G$ constructed by concatenating paths in $\extree(p_2)$ from $\cen(p_2)$ to $v_2$ and $v_3$, and the edge $v_2v_3$, then this closed walk separates $\loc(p_1)$ from $\loc(p_3)$ on the sphere. Observe that this walk is not necessarily a path, but, as it is formed by two paths in the tree $\extree(p_2)$ closed by an edge, its removal partitions the sphere into two disks. The requirement is exactly that one of these disks contains $\loc(p_1)$ and the second contains $\loc(p_3)$. 

Similarly as before, we will say that triple $(p_1,p_2,p_3)$ {\em{certifies}} that $f$ is a type-$2$ singular face. The following lemma shows that, again, there are not so many type-$2$ singular faces.

\begin{figure}[t]
        \centering
        \subfloat[Type $1$]{
                \centering
                \def\svgwidth{0.28\columnwidth}
                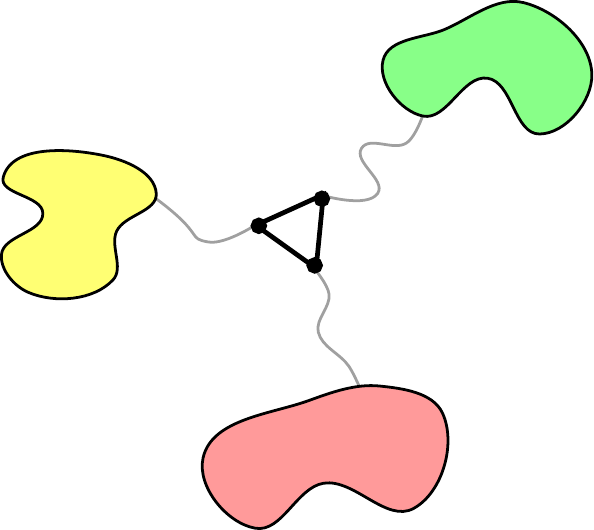
        }
        \quad
        \subfloat[Type $2$]{
                \centering
                \def\svgwidth{0.28\columnwidth}
                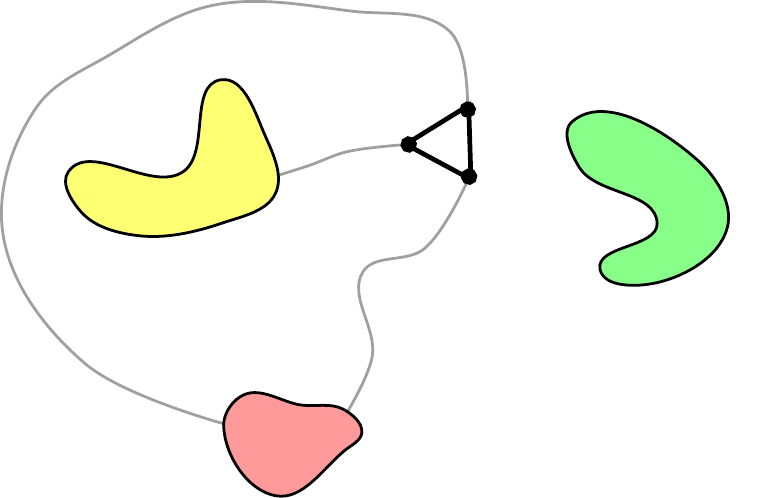
        }
        \quad
        \subfloat[Type $3$]{
                \centering
                \def\svgwidth{0.28\columnwidth}
                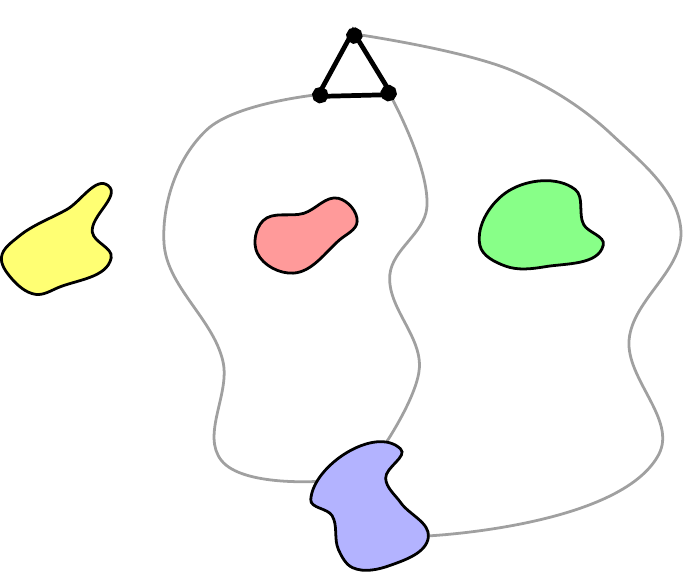
        }
\caption{Singular faces of types $1$, $2$, and $3$. Paths from a vertex to the nearest object have been depicted in gray.}\label{fig:types}
\end{figure}

\begin{lemma}\label{lem:type-2}
Suppose $\{p_1,p_2,p_3\}=\Fam\subseteq \Obj$ is a normal family of three objects. Then there is at most $1$ face of $G$ that is a type-$2$ singular face for $\Fam$ and is certified by the triple $(p_1,p_2,p_3)$.
\end{lemma}
\begin{proof}
Suppose there are two such faces $f$ and $f'$. We adopt the notation from the definition of a type-$2$ singular face for $f$: $\{v_1\}=V(\bnd f)\cap \Vorpar(p_1)$, $\{v_2,v_3\}=V(\bnd f)\cap \Vorpar(p_2)$, $P_t$ is the $\cen(p_2)$-$v_t$ path inside $\extree(p_2)$, for $t=2,3$, and $W$ is the closed walk formed by concatenation of $P_2$, $P_3$, and edge $v_2v_3$. Note that edge $v_2v_3$ does not belong to $\extree(p_2)$, since then $P_2$ would be a subpath of $P_3$ or vice versa, and $W$ would degenerate to a path that would not separate any two nonempty disks. Let $D_1$ be the open disk of $\Sigma\setminus W$ that contains $p_1$, and $D_3$ be the one that contains $p_3$. Since $v_1\in \Vorpar(p_1)$, $\loc(p_1)$ is contained in $D_1$, and $G[\Vorpar(p_1)]$ is connected, we infer that $v_1\in D_1$. Moreover, if $P_1$ is the $\cen(p_1)$-$v_1$ path inside $\extree(p_1)$, then $P_1$ is entirely contained in $D_1$.

Let $v_t'$ for $t=1,2,3$, $P_t'$ for $t=2,3$ and $W'$ be the analogical objects defined for $f'$. Since $f'$ contains a vertex of $\Vorpar(p_1)$, then we have that $f'$ is contained in $D_1$. Since $P_2'$ and $P_3'$ are paths inside the tree $\extree(p_2)$, and $W$ is defined also by paths inside the tree $\extree(p_2)$, we infer that both $P_2'$ and $P_3'$ are entirely contained in $D_1\cup W$. Moreover, they do not contain edge $v_2v_3$ (since this edge does not belong to $\extree(p_2)$). As $f'\neq f$ and the other face incident to $v_2v_3$ lies outside $D_1$, we have also that $v_2v_3\neq v_2'v_3'$. Consequently, we infer that the closed walk $W'$ does not use the edge $v_2v_3$.

To obtain a contradiction, we now exhibit a curve $\gamma$ on the sphere that connects $\cen(p_1)$ with $\cen(p_3)$ and does not touch $W'$; this will contradict the fact that $W'$ separates $\loc(p_1)$ from $\loc(p_3)$. Construct $\gamma$ by:
\begin{enumerate}[label=(\roman*)]
\item\label{p1} starting by traveling from $\cen(p_1)$ to $v_1$ using path $P_1$ inside $\extree(p_1)$;
\item\label{p2} traveling inside the face $f$ from $v_1$ to the middle of edge $v_2v_3$; and
\item\label{p3} finishing by traveling inside the disk $D_3$ from the middle of edge $v_2v_3$ to $\cen(p_3)$.
\end{enumerate}
In part \ref{p1} of $\gamma$ we do not touch $W'$, because we only travel through vertices of $\Vorpar(p_1)$, whereas all the vertices on $W'$ belong to $\Vorpar(p_2)$. Then, in part \ref{p2} we do not touch $W'$, since $W'$ does not use the edge $v_2v_3$. Finally, in part \ref{p3} we do not touch $W'$ since $W'$ is disjoint with $D_3$. Hence, $\gamma$ is a curve from $\cen(p_1)$ to $\cen(p_3)$ that avoids $W'$, a contradiction.
\end{proof}

Finally, we define the third type of singular faces. We say that face $f$ is a {\em{type-$3$ singular face}} for $\Fam$, if there exists a quadruple of distinct objects $p_0,p_1,p_2,p_3\in \Fam$ satisfying the following. Let $\{v_1,v_2,v_3\}=V(\bnd f)$. The first requirement is that $v_1,v_2,v_3\in \Vorpar(p_0)$, but no edge of $\bnd f$ is contained in $\extree(p_0)$. Let $P_t$ be the path inside $\extree(p_0)$ from $\cen(p_0)$ to $v_t$, for $t=1,2,3$. Let $W_t$ be the closed walk formed by concatenation of paths $P_{t+1}$, $P_{t+2}$ and the edge $v_{t+1}v_{t+2}$, where the indices behave cyclically. Observe that removal of walks $W_1,W_2,W_3$ partitions sphere $\Sigma$ into four open disks: disks $D_t$ for $t=1,2,3$ such that $W_t$ contains the boundary of $D_t$, and the last disk being simply the face $f$. Then the second requirement is that $\loc(p_t)$ is entirely contained in $D_t$, for $t=1,2,3$.

Similarly as before, we will say that a quadruple $(p_0,p_1,p_2,p_3)$ {\em{certifies}} that $f$ is a type-$3$ singular face. The following lemma shows that, again, there are not so many type-$3$ singular faces.

\begin{lemma}\label{lem:type-3}
Suppose $\{p_0,p_1,p_2,p_3\}=\Fam\subseteq \Obj$ is a normal family of four objects. Then there is at most $1$ face of $G$ that is a type-$3$ singular face for $\Fam$ and is certified by the quadruple $(p_0,p_1,p_2,p_3)$.
\end{lemma}
\begin{proof}
Suppose there are two such faces $f$ and $f'$. We adopt the notation from the definition of a type-$3$ singular face for $f$, and we will follow the same notation but with primes for the face $f'$. By assumption, no edge of $f$ or $f'$ belongs to $\extree(p_0)$.

Recall that removal of walks $W_1,W_2,W_3$ partitions $\Sigma$ into four open disks: $D_1,D_2,D_3$ and $f$. Since $f'\neq f$, $f'$ is contained in one of the disks $D_1,D_2,D_3$. Without loss of generality suppose $f'$ is contained in $D_1$. Removal of $W_1$ from $\Sigma$ partitions $\Sigma$ into two disks: one of them is simply $D_1$, and the second is a disk $\widetilde{D_1}$ that contains $D_2$, $D_3$ and $f$. Observe that since $W_1$ is formed by edge $v_2v_3$ and two paths $P_2,P_3$, which are contained in the tree $\extree(p_0)$, then each of the paths $P_1',P_2',P_3'$ has to be entirely contained either in $D_1\cup W_1$ or in $\widetilde{D_1}\cup W_1$ --- this is because paths $P_1',P_2',P_3'$ are also contained in the tree $\extree(p_0)$. Since $\bnd f'\subseteq D_1\cup W_1$, we infer that all these paths are contained in $D_1\cup W_1$. Since $f'$ is also contained in $D_1$, we conclude that all the closed walks $W_1',W_2',W_3'$ are contained in $D_1\cup W_1$.

Now examine disk $\widetilde{D_1}$. Notice that this disk is disjoint with all the walks $W_1',W_2',W_3'$ and does not contain $f'$, which means that it is entirely contained in one of the disk $D_1'$, $D_2'$, or $D_3'$. However, $\widetilde{D_1}$ contains both $\loc(p_2)$ and $\loc(p_3)$, while every disk $D_1'$, $D_2'$, $D_3'$ contains only one of them. This is a contradiction.
\end{proof}

In Lemmas~\ref{lem:type-1},~\ref{lem:type-2}, and~\ref{lem:type-3} we focused only on families $\Fam$ of size $3$, $3$, and $4$, respectively. The following lemma, which follows immediately from Lemma~\ref{lem:inclusion} and the definition of singular faces of types $1$, $2$, and $3$, justifies why we can do it.

\begin{lemma}\label{lem:small-subfamily}
Suppose $\Fam\subseteq \Obj$ is a normal family of objects. Then for a face $f$ of $G$, the following conditions hold:
\begin{itemize}
\item If $f$ is a type-$1$ singular face for $\Fam$, certified by a triple $(p_1,p_2,p_3)$, then it is a type-$1$ singular face for $\{p_1,p_2,p_3\}$, certified by the same triple $(p_1,p_2,p_3)$.
\item If $f$ is a type-$2$ singular face for $\Fam$, certified by a triple $(p_1,p_2,p_3)$, then it is a type-$2$ singular face for $\{p_1,p_2,p_3\}$, certified by the same triple $(p_1,p_2,p_3)$.
\item If $f$ is a type-$3$ singular face for $\Fam$, certified by a quadruple $(p_0,p_1,p_2,p_3)$, then it is a type-$3$ singular face for $\{p_0,p_1,p_2,p_3\}$, certified by the same quadruple $(p_0,p_1,p_2,p_3)$.
\end{itemize}
\end{lemma}

Now we show that for a normal family $\Fam$, every branching point of $\dgm_\Fam$ is always a singular face, and moreover its type depends on the number of bridges it is incident to in $\dgm_\Fam$. Note that since $\dgm_\Fam$ is $3$-regular, every its vertex is incident to $0$, $1$, or $3$ bridges of $\dgm_\Fam$.

\begin{lemma}\label{lem:types-bridges}
Suppose $\Fam\subseteq \Obj$ is a normal subfamily of objects, and let $\dgm=\dgm_\Fam$ be the corresponding Voronoi diagram. Then for every branching point $f\in V(\dgm)$, the following holds:
\begin{enumerate}[label=(\arabic*)]
\item\label{i1} If $f$ is not incident to any bridge in $\dgm$, then $f$ is a type-$1$ singular face for $\Fam$.
\item\label{i2} If $f$ is incident to exactly one bridge in $\dgm$, then $f$ is a type-$2$ singular face for $\Fam$.
\item\label{i3} If $f$ is incident to exactly three bridges in $\dgm$, then $f$ is a type-$3$ singular face for $\Fam$.
\end{enumerate}
In particular, every branching point of $\dgm$ is a singular face for $\Fam$ of one of the three types.
\end{lemma}
\begin{proof}
Let $\tilde{e}_1^*,\tilde{e}_2^*,\tilde{e}_3^*$ be the three edges incident to $f$ in $\dgm$ (possibly two of these edges are equal, in case $f$ is incident to a loop in $V(\dgm)$). Then, in the prediagram $\pdgm$, $f$ was incident to three edges $e_1^*,e_2^*,e_3^*$, which are the first edges of the paths contracted to $\tilde{e}_1^*,\tilde{e}_2^*,\tilde{e}_3^*$, respectively (these three edges are pairwise different, since $\pdgm$ is simple). Edges $e_1^*,e_2^*,e_3^*$ are dual to edges $e_1,e_2,e_3$, respectively, which form the boundary of the face $f$.

We perform a case study depending on how many of edges $\tilde{e}_1^*,\tilde{e}_2^*,\tilde{e}_3^*$ are bridges in $\dgm$ (note that a loop is not considered a bridge). Observe that, for $t=1,2,3$, edge $\tilde{e}_t^*$ is a bridge in $\dgm$ if and only if $e_t^*$ is a bridge in $\pdgm$.

Suppose first that none of edges $\tilde{e}_1^*,\tilde{e}_2^*,\tilde{e}_3^*$ are bridges. Then $f$ is incident to three different faces of $\dgm$. Since faces of $\dgm$ correspond one-to-one to objects of $\Fam$ (see Lemmas~\ref{lem:pdiagram} and Lemmas~\ref{lem:diagram}), there are three distinct objects $p_1,p_2,p_3\in \Fam$ such that $f$ is incident to corresponding faces $f^*_{p_1},f^*_{p_2},f^*_{p_3}$ of $\dgm$. This means that one of the vertices of $f$ belongs to $\Vorpar(p_1)$, second to $\Vorpar(p_2)$, and third to $\Vorpar(p_3)$, and the triple $(p_1,p_2,p_3)$ certifies that $f$ is a type-$1$ singular face for $\Fam$. This proves \ref{i1}.

Suppose second that $\tilde{e}_1^*$ is a bridge in $\dgm$, while edges $\tilde{e}_2^*,\tilde{e}_3^*$ are not bridges. This means that there exists two distinct objects $p_1,p_2\in \Fam$, with corresponding faces of $\dgm$ being $f^*_{p_1}$ and $f^*_{p_2}$, such that the face incident to $f$ between $\tilde{e}_2^*,\tilde{e}_3^*$ is $f^*_{p_1}$, while the face $f^*_{p_2}$ resides on both sides of the edge $\tilde{e}_1^*$. Let $H_1$ be the subgraph of $\dgm$ induced by all the branching points that become disconnected from $f$ by removing the bridge $\tilde{e}_1^*$. Since $\dgm$ is $3$-regular, $H_1$ is not a tree, so it has at least two faces. Only one face of $H_1$ contains face $f^*_{p_2}$, and all the other faces of $H_1$ are actually faces of $\dgm$. Let then $p_3$ be any object of $\Fam$ whose corresponding face $f^*_{p_3}$ of $\dgm$ is a face of $H_1$. Similarly as in the previous case, one vertex $v_1\in V(\bnd f)$ belongs to $\Vorpar(p_1)$ and two vertices of $v_2,v_3\in V(\bnd f)$ belong to $\Vorpar(p_2)$. Moreover these two vertices $v_2,v_3$ are connected by the primal edge $e_1$. Consider now a closed walk $W$ obtained by concatenating: walk $P_2$ inside $\extree(p_2)$ from $\cen(p_2)$ to $v_2$, walk $P_3$ inside $\extree(p_2)$ from $\cen(p_2)$ to $v_3$, and edge $e_1=v_2v_3$. Removal of walk $W$ from the sphere separates the sphere into two open disks, one of which contains all the branching points of $V(H_1)$ --- and in particular also $\loc(p_3)$ --- and the second contains all other the branching points of $\dgm$ --- and in particular also $\loc(p_1)$. Hence, $W$ separates $\loc(p_1)$ from $\loc(p_3)$. This means that triple $(p_1,p_2,p_3)$ certifies that $f$ is a type-$2$ singular face for $\Fam$. This proves \ref{i2}.

Finally, suppose that all the edges $\tilde{e}_1^*,\tilde{e}_2^*,\tilde{e}_3^*$ are bridges in $\dgm$. This means that there is one object $p_0\in \Fam$ and corresponding face $f^*_{p_0}$ of $\dgm$, such that $f^*_{p_0}$ is on both sides of each of the edges $\tilde{e}_1^*,\tilde{e}_2^*,\tilde{e}_3^*$. Let $\{v_1,v_2,v_3\}=V(\bnd f)$, where $e_t=v_{t+1}v_{t+2}$ for $t=1,2,3$. Hence, in particular we have that $v_1,v_2,v_3\in \Vorpar(p_0)$. Similarly as in the previous case, let $H_t$ be the subgraph of $\dgm$ induced by all the branching points that become disconnected from $f$ by removing the bridge $\tilde{e}_t^*$, for $t=1,2,3$. Each $H_t$ is not a forest, since $\dgm$ is $3$-regular. Moreover, each $H_t$ has one face containing the original face $f^*_{p_0}$, whereas all the other faces of $\dgm$ are partitioned into face sets $F(H_1)$, $F(H_2)$, $F(H_3)$. Hence, there exist three distinct objects $p_1,p_2,p_3\in \Fam$, distinct from $p_0$, such that the corresponding faces $f^*_{p_1}$, $f^*_{p_2}$, $f^*_{p_3}$ of $\dgm$ are actually faces of $H_1$, $H_2$, and $H_3$, respectively. We now define $P_t$ to be the path inside $\extree(p_0)$ from $\cen(p_0)$ to $v_t$, for $t=1,2,3$, and $W_t$ to be the closed walk formed by concatenation of paths $P_{t+1}$, $P_{t+2}$ and the edge $v_{t+1}v_{t+2}$. Similarly as in the previous point, we see that each walk $W_t$ separates $\loc(p_t)$ from $\loc(p_{t+1})$ and $\loc(p_{t+2})$. This means that quadruple $(p_0,p_1,p_2,p_3)$ certifies that $f$ is a type-$3$ singular face for $\Fam$. This proves \ref{i3}, and concludes the proof.
\end{proof}

We are finally ready to prove Theorem~\ref{thm:impFaces}, that is, present an enumeration algorithm for important faces. Family $\impFaces$ will consist of a union of three families $\impFaces_1,\impFaces_2,\impFaces_3$. Families $\impFaces_1$ and $\impFaces_2$ are constructed by considering every triple objects $p_1,p_2,p_3\in \Fam$ s.t. $\{p_1,p_2,p_3\}$ is normal, and including all the faces that are type-$1$, resp. type-$2$, singular faces for $\{p_1,p_2,p_3\}$, certified by $(p_1,p_2,p_3)$. Lemmas~\ref{lem:type-1} and~\ref{lem:type-2} guarantees that for every triple we have at most two such faces for type $1$ and at most one such face for type $2$, so in total we have that $|\impFaces_1|\leq 2\onum(\onum-1)(\onum-2)$ and $|\impFaces_2|\leq \onum(\onum-1)(\onum-2)$. Similarly, family $\impFaces_3$ is constructed by considering every quadruple of objects $p_0,p_1,p_2,p_3$ s.t. $\{p_0,p_1,p_2,p_3\}$ is normal, and including all the faces that are type-$3$ singular faces for $\{p_0,p_1,p_2,p_3\}$, certified by $(p_0,p_1,p_2,p_3)$. Lemma~\ref{lem:type-3} guarantees that for every quadruple we have at most one such face, and thus $|\impFaces_3|\leq \onum(\onum-1)(\onum-2)(\onum-3)$. Note that, for a given triple or quadruple of objects, it can be verified in polynomial time using the definition whether a given face is a type-$1$/type-$2$/type-$3$ singular face for this triple/quadruple. Hence, families $\impFaces_1,\impFaces_2,\impFaces_3$ can be constructed in time $\Oh(\onum^4\cdot n^{\Oh(1)})$.

We now set $\impFaces=\impFaces_1\cup \impFaces_2\cup \impFaces_3$, so in particular we have the following:
$$|\impFaces|\leq 3\onum(\onum-1)(\onum-2)+\onum(\onum-1)(\onum-2)(\onum-3)=\onum^2(\onum-1)(\onum-2)<\onum^4.$$
Let us consider any normal subfamily of objects $\Fam\subseteq \Obj$. By Lemma~\ref{lem:small-subfamily}, we have that if $f\in F(G)$ is a type-$t$ singular face for $\Fam$, then $f$ has been included in respective family $\impFaces_t$, for $t=1,2,3$. On the other hand, Lemma~\ref{lem:types-bridges} implies that every branching point of $\dgm_\Fam$ is a singular face for $\Fam$ of one of the three types. We infer that every branching point of $\dgm_\Fam$ belongs to $\impFaces$. This concludes the proof of Theorem~\ref{thm:impFaces}.

\subsection{Curves, nooses, and sphere cut decompositions}\label{sec:curves}

\paragraph*{Curves.}

Recall that for us a {\em{curve}} on a sphere is a homeomorphic image of an interval, and a {\em{closed curve}} is a homeomorphic image of a circle. A {\em{directed (closed) curve}} $\crv$ is simply a (closed) curve with a specified direction of traversal. The same curve but with the opposite direction of traversal will be denoted by $\crv^{-1}$. For a directed closed curve $\crv$, the {\em{area enclosed}} by $\crv$, denoted $\enc(\crv)$, is the open disk of $\Sigma\setminus \crv$ that lies to the right of $\crv$ in the direction of its traversal (in this definition we fix the orientation of $\Sigma$). The second open disk, called the {\em{area excluded}} by $\crv$, will be denoted by $\exc(\crv)$. Note that $\enc(\crv)=\exc(\crv^{-1})$. 

%For two points $a,b\in \crv$, by $\crv[a,b]$ we shall denote the directed subcurve of $\crv$ from $a$ to $b$, traversed in the same direction as $\crv$. By $\crv(a,b)$ we denote the same directed subcurve, but without endpoints.

\paragraph*{Nooses and the radial graph.}

Let us consider a situation where some connected multigraph $H$ is embedded on the sphere $\Sigma$. We say that a directed closed curve $\crv$ is a {\em{noose}}\footnote{In the classic literature nooses are not directed, but it will be convenient for us to work with directed ones. Also, often nooses are not required to visit every face at most once and nooses with this property are called {\em{tight}}. For us, all the considered nooses are tight.} if it intersects the embedding of $G$ only in the vertices of $H$, and moreover visits every face of $H$ at most once. The {\em{length}} of a noose $\crv$ is equal to $|V(H)\cap \crv|$, the number of vertices it traverses; we require that a noose traverses at least one vertex. We often say that a noose $\crv$ is {\em{with respect to $H$}} if the graph $H$ is not clear from the context.

When considering nooses, we shall work on the {\em{radial}} graph of $H$, denoted $\Rad(H)$. The graph $\Rad(H)$ is a bipartite multigraph, where one partite set consists of $V(H)$, the vertices of $H$, and the second partite set is $F(H)$, the faces of $H$. For every face $f\in F(H)$ and vertex $v\in V(H)$ we put an edge between $f$ and $v$ for every occurrence of $v$ on the boundary of $f$; note that if $v$ appears multiple times on the boundary of $f$, it will be connected to $f$ in $\Rad(H)$ via multiple edges. A sphere embedding of $H$ naturally induces a sphere embedding of $H$ as follows. For every face $f$ we put the corresponding vertex of $\Rad(H)$ inside $f$; we call this vertex the {\em{center}} of $f$, and denote it by $\cn(f)$. Then we connect it to the vertices on the boundary of $f$ radially, according to their order of traversal on $\bnd f$. See Figure~\ref{fig:noose} for an example.

Now observe that every noose $\crv$ can be homeomorphically transformed to an equivalent noose (in terms of the order of vertices and faces visited) that is a simple cycle in the radial graph $\Rad(H)$ with a chosen orientation. More precisely, whenever $\crv$ travels from a vertex $u$ to a vertex $v$ through a face $f$, then we replace this part of $\crv$ by going from $u$ to $\cn(f)$ along an edge of $\Rad(H)$, and then from $\cn(f)$ to $v$ again along an edge of $\Rad(H)$. Note that if $v$ appears on the boundary of $f$ multiple times, we may always pick the edge between $\cn(f)$ and $v$ that corresponds to the direction from which $v$ is accessed by $\crv$. The same holds for entering $f$ from $u$. Hence, from now on whenever considering some embedded connected multigraph $H$, we will consider only nooses that are oriented cycles in $\Rad(H)$. 

%Note that if two such nooses are summable (as directed closed curves), then their sum is also a noose of this type. Moreover, points $a,b$ in the definition of being summable are always vertices of $H$ or centers of some faces of $H$.

Given some noose $\crv$, we may define the {\em{subgraph enclosed}} by $\crv$ as the subgraph $\enc(\crv,H)$ of $H$ consisting of all the vertices and edges of $H$ that are embedded into $\enc(\crv)\cup \crv$. Similarly, the {\em{subgraph excluded}} by $\crv$ is the subgraph $\exc(\crv,H)$ of $H$ that consists of all the vertices and edges of $H$ that are embedded into $\exc(\crv)\cup \crv$. Note that $(E(\enc(\crv,H)),E(\exc(\crv,H)))$ is a partition of $E(H)$, and $V(\enc(\crv,H))\cap V(\exc(\crv,H))$ consists of exactly the vertices traversed by $\crv$.

\begin{figure}[htbp!]
                \centering
                \def\svgwidth{0.5\columnwidth}
                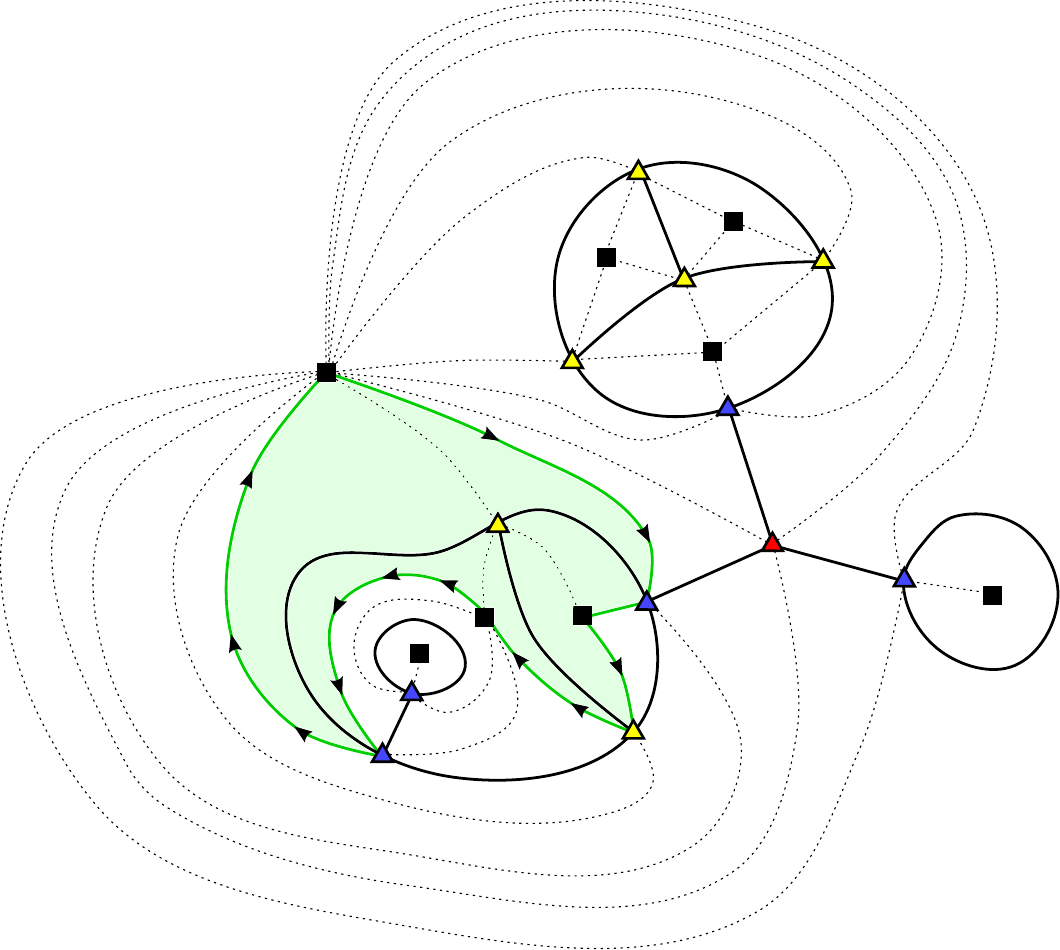
\caption{The Voronoi diagram from Figure~\ref{fig:diagram} together with its radial graph (dotted edges) and an exemplary noose $\crv$. The area enclosed by $\crv$ is depicted in green.}\label{fig:noose}
\end{figure}

\paragraph*{Sphere cut decompositions.}

We now recall the main tool that will be used to justify the correctness of our algorithm, namely sphere cut decompositions. Intuitively speaking, every graph $H$ embedded on a sphere has a hierarchical decomposition (formally, a branch decomposition), where every separator used to decompose the graph is formed by vertices traversed by some noose of length at most $\Oh(\sqrt{|V(H)|})$. In our case we will apply the sphere cut decomposition to the diagram $\dgm_{\SolObj}$, where $\SolObj$ is the optimum solution to the problem. By Lemma~\ref{lem:diagram} we have that $|V(\dgm_\SolObj)|=\Oh(k)$, which means that all the separators used in the (unknown) sphere cut decomposition of $\dgm_{\SolObj}$ have length $\Oh(\sqrt{k})$. By dint of Theorem~\ref{thm:impFaces}, we are able to construct a family of $\onum^{\Oh(\sqrt{k})}$ candidates for these separators. Then we run a recursive algorithm that at each step guesses a separator that splits the (unknown) solution in a balanced way, and then recursively solves simpler subinstances. Therefore, to see why the algorithm is correct, it is necessary to have a good understanding of noose separators in planar graphs, which brings us to the framework of sphere cut decompositions.

We now move to a formal introduction of sphere cut decompositions. Recall that for a multigraph $G$, a pair $\langle\tree,\iso\rangle$ is a {\em{branch decomposition}} of $G$ if $T$ is a tree with all the internal vertices of degree $3$, and $\iso$ is a bijection from $E(G)$ to the set of leaves of $T$. For an edge $e\in E(T)$, we define the {\em{middle set}} of $e$, denoted $\md(e)$, as follows. Suppose that removal of $e$ breaks $T$ into two trees $T_1,T_2$, and let $F_1,F_2$ be the sets of edges of $G$ mapped to the leaves of $T$ contained in $T_1$ and $T_2$, respectively. Then $(F_1,F_2)$ is a partition of $E(G)$. Set $\md(e)$ consists of all the vertices of $G$ that are incident both to an edge of $F_1$ and to an edge of $F_2$. The {\em{width}} of edge $e$ is then defined as $|\md(e)|$, and the {\em{width}} of decomposition $\langle\tree,\iso\rangle$ is the minimum width among the edges of $T$. The {\em{branchwidth}} of $G$, denoted $\bw(G)$, is the minimum possible width of a branch decomposition of $G$.

Suppose now that $G$ is a connected multigraph without loops that is embedded on a sphere $\Sigma$. A {\em{sphere cut decomposition}} ({\em{sc-decomposition}}) of $G$ is a triple $\langle\tree,\iso,\noose\rangle$, where:
\begin{itemize}
\item $\langle\tree,\iso\rangle$ is a branch decomposition of $G$;
\item $\noose$ is a function that assigns to every edge $e\in E(T)$ a noose $\noose(e)$ that traverses exactly the set of vertices $\md(e)$. Moreover, if $(F_1,F_2)$ is the partition of $E(G)$ as in the definition of $\md(e)$, and $G_1,G_2$ are the subgraphs of $G$ spanned by $F_1,F_2$, respectively, then one of $G_1,G_2$ is $\enc(\noose(e),G)$ and the second is $\exc(\noose(e),G)$.
\end{itemize}
Note that in the second condition, whether  $\enc(\noose(e),G)$ is $G_1$ or $G_2$ depends on the choice of direction of the noose $\noose(e)$. Observe also that the given partition $(F_1,F_2)$ of $E(G)$ uniquely defines the noose $\noose(e)$ (treated as a simple cycle in $\Rad(G)$, without orientation), for the following reason: For every face $f$ that is traversed by $\noose(e)$, the intersections of $\bnd f$ with $F_1$ and $F_2$ have to be just two intervals of the walk $\bnd f$, since face $f$ is visited only once by $\noose(e)$. Hence, by examining these intervals we can uniquely deduce between which vertices the noose travels within $f$, and from which directions these vertices are accessed.

We will regard sc-decompositions also as branch decompositions, with all the inherited notions for them.

As observed by Dorn et al.~\cite{DornPBF10}, the following theorem was implicitly proved by Seymour and Thomas in~\cite{SeymourT94}. The improved running time (from original $\Oh(n^4)$ to $\Oh(n^3)$) is due to Gu and Tamaki~\cite{GuT08}.

\begin{theorem}[\cite{DornPBF10,GuT08,SeymourT94}]\label{thm:sc-decomp}
Suppose $G$ is a bridgeless, connected multigraph $G$ without loops, embedded on a sphere $\Sigma$. Then there exists an sc-decomposition $\langle\tree,\iso,\noose\rangle$ such that $\langle\tree,\iso\rangle$ is a branch decomposition of $G$ of optimum width. Moreover, such an sc-decomposition can be found in time $\Oh(n^3)$.
\end{theorem}

Let us now deliberate on the proof of Theorem~\ref{thm:sc-decomp}, as the statement differs in important details from the one provided by Dorn et al.~\cite{DornPBF10} (see Theorem~1 in this work). The following argumentation is based on the sketch of the proof of Theorem~1 in~\cite{DornPBF10}.

The proof follows from (5.1) and the proof of (7.2) in~\cite{SeymourT94}. In (5.1), it is proved that for every $2$-connected multigraph $G$, there exists an optimum carving decomposition of $G$ that uses only bond carvings (is a {\em{bond carving decomposition}}). Here, a carving decomposition is a similar concept to a branch decomposition, but this time leaves of decomposition $T$ correspond one-to-one to vertices of $G$, and the width of an edge $e\in E(T)$ is defined as the number of edges traversing the corresponding partition $(X_1,X_2)$ of $V(G)$. This partition $(X_1,X_2)$ is called a {\em{carving}}, and the carving is a {\em{bond carving}} if both $X_1$ and $X_2$ induce connected graphs in $G$. Equivalently, the edges traversing the carving form an inclusion-wise minimal edge cut between some two vertices of $G$. 

Note now that if $G$ is embedded on a sphere $\Sigma$, then such a minimal cut corresponds to a simple cycle in the dual $G^*$, which traverses duals of the consecutive edges of the cut. This cycle divides the sphere into two open disks such that vertices of $X_1$ are embedded into one of them and vertices of $X_2$ are embedded into the second of them.

Carving decompositions are then translated to branch decompositions using the notion of a {\em{medial graph}} $\Med(G)$, which is just the dual of the radial graph. More precisely, vertices of the $\Med(G)$ are the edges of the primal graph, and for every vertex $v\in V(G)$ we connect the edges incident to $v$ into a cycle, in the order of their appearance around $v$ in the sphere embedding. Thus, the medial graph is also embedded into the same sphere. In the proof of (7.2) in~\cite{SeymourT94}, Seymour and Thomas consider carving decompositions of the medial graph $\Med(G)$. Since the vertex set of $\Med(G)$ is exactly the edge set of $G$, carving decompositions of $\Med(G)$ correspond exactly to branch decompositions of $G$. Moreover, cycles in the dual of $\Med(G)$ (which is just $\Rad(G)$) correspond exactly to nooses in $G$. Hence, bond carving decompositions of $\Med(G)$ correspond to sc-decompositions of $G$. Seymour and Thomas prove that this correspondence translates an optimum bond carving decomposition of $\Med(G)$ into an optimum branch decomposition of $G$, where the width gets divided by two. In particular the carving width of $\Med(G)$ is twice larger than the branchwidth of $G$. Since this translation produces in fact an sc-decomposition, we conclude that there exists an sc-decomposition of $G$ of optimum width.

In (9.1) Seymour and Thomas give an $\Oh(n^4)$ algorithm that constructs an optimum bond carving decomposition of an planar graph. By applying this algorithm to $\Med(G)$ and translating the output to an optimum sc-decomposition of $G$, we obtain the sought result. The running time of the algorithm has been improved from original $\Oh(n^4)$ to $\Oh(n^3)$ by Gu and Tamaki~\cite{GuT08}.

Observe now that in (5.1) of~\cite{SeymourT94}, it was assumed that the graph, whose carving decompositions are considered, needs to be $2$-connected. Indeed, for every vertex $v\in V(G)$ we have that $(\{v\},V(G)\setminus \{v\})$ is one of the carvings in every carving decomposition. If now $v$ was a cutvertex, then this carving would not be a bond carving. 

It is now easy to see that the medial graph $\Med(G)$ is $2$-connected if and only if the primal graph $G$ is connected and bridgeless. Therefore, the whole presented argumentation holds provided that $G$ is connected and bridgeless. And indeed, if $G$ contained some bridge $b$, then $(\{b\},E(G)\setminus \{b\})$ would be one of the partitions $(F_1,F_2)$ considered for some edge of the decomposition $T$. However, then the noose enclosing only the bridge $b$ would need to visit the face on both sides of $b$ twice, which contradicts the definition of a noose. Therefore, in Theorem~\ref{thm:sc-decomp} it is necessary and sufficient to assume that $G$ is connected and bridgeless.

For the other assumptions, Dorn et al.~\cite{DornPBF10} consider only simple graphs. However, all the arguments work in the same way in presence of multiple edges, since the original proof of Seymour and Thomas~\cite{SeymourT94} works for multigraphs (as is noted in the first sentence of~\cite{SeymourT94}). We need, however, to exclude loops, since in the presence of loops, understood as curves connecting a vertex to itself, there arise technical problems with the medial graph $\Med(G)$. Namely, for the natural generalization of this notion to multigraphs with loops, a loop can correspond to a cutvertex in $\Med(G)$. We remark that Seymour and Thomas allow loops in their work, but they understand them differently, as edges of arity $1$, being de facto annotations at vertices. We also remark that in their version of Theorem~\ref{thm:sc-decomp}, Dorn et al.~\cite{DornPBF10} (incorrectly) assume only that the graph does not have vertices of degree $1$, instead of excluding bridges. As we have seen, this assumption is necessary.

The second necessary ingredient is the well-known fact that planar graphs have branchwidth bounded by approximately the square root of the number of vertices.

\begin{theorem}[\cite{FominT04}]\label{thm:bw-sqroot}
For a planar graph $G$ on $n$ vertices, it holds that $\bw(G)\leq \sqrt{4.5 n}$.
\end{theorem}

We note that Theorem~\ref{thm:bw-sqroot} can be trivially extended to multigraphs. Thus, Theorems~\ref{thm:sc-decomp} and~\ref{thm:bw-sqroot} together imply that every bridgeless, connected multigraph $G$ without loops that is embedded on a sphere, admits an sc-decomposition of width at most $\sqrt{4.5 |V(G)|}$. This corollary will be our main tool in the sequel.

\subsection{Balanced nooses in plane graphs}
\label{sec:balanc-noos-plane}

Given a connected multigraph $G$ embedded on a sphere, a noose $\crv$ w.r.t. $G$ is {\em{$\alpha$-edge-balanced}} if the numbers of edges enclosed and excluded by $\crv$ are bounded by $\alpha |E(G)|$, i.e. $|E(\enc(\crv,G))|\leq \alpha |E(G)|$ and $|E(\exc(\crv,G))|\leq \alpha |E(G)|$. The following lemma shows that the plane multigraphs we are interested in admit short edge-balanced nooses.

\begin{theorem}\label{thm:edge-balanced-noose}
Let $G$ be a connected 3-regular multigraph with $n$ vertices, $m\geq 6$ edges, possibly with loops, and embedded on a sphere $\Sigma$. Then there exists a $\frac{2}{3}$-edge-balanced noose w.r.t. $G$ that has length at most $\sqrt{4.5 n}$.
\end{theorem}
\begin{proof}
Let $T^*$ be the tree of bridgeless components of $G$, where with every node $x\in V(T^*)$ we associate a subset $B_x\subseteq V(G)$ inducing a bridgeless component of $G$, and every edge $xy\in E(T^*)$ corresponds to a bridge $b_{xy}$ connecting $G[B_x]$ with $G[B_y]$. For every edge $xy\in E(T^*)$, consider the connected components $G_1$ and $G_2$ of $G-b_{xy}$ and let us direct this edge from $x$ to $y$ if $|E(G_1)|>|E(G_2)|$, from $y$ to $x$ if $|E(G_1)|<|E(G_2)|$, and breaking the tie arbitrarily in case $|E(G_1)|=|E(G_2)|$. Thus, every edge of $T^*$ becomes directed. Since $T^*$ has exactly $|V(T^*)|-1$ edges, we have that the sum of outdegrees in $T^*$ is equal to $|V(T^*)|-1$, and hence there exists a node $x_0\in V(T^*)$ that has no outgoing directed edge.

Let $H=G[B_{x_0}]$ and observe that $H$ is a bridgeless, connected multigraph embedded on $\Sigma$. Moreover, observe that $H$ does not contain loops for the following reason: A loop in a 3-regular multigraph must be attached to a vertex having exactly one other incident edge that is a bridge. Hence, if $H$ contained a loop then $B_{x_0}=V(H)$ would consist of one vertex and $x_0$ would be a leaf in $T^*$. Since the only edge incident to $x_0$ in $T^*$ is directed towards $x_0$ and $G[B_{x_0}]$ has one edge, we would have that $m\leq 2\cdot 1+1=3$. This would be a contradiction with the assumption that $m\geq 6$.

Suppose first that $|B_{x_0}|=1$, i.e., $H$ is graph consisting of a single vertex $u$. Then $u$ is incident to $3$ bridges in $G$, and removal of each of these bridges separates from $u$ a connected component having at most $\frac{m-1}{2}$ edges. Let $C$ be the component that has the most vertices among these three ones; then $|E(C)|\geq \frac{m-3}{3}=\frac{m}{3}-1$. Take a noose $\crv$ of length $1$ that visits $u$ and the unique face incident to $u$, chosen in such a manner that it encloses exactly component $C$ and the bridge connecting $C$ with $u$. Then, as $m\geq 6$, it follows that 
\begin{eqnarray*}
|E(\enc(\crv,G))| & = & |E(C)|+1\leq \frac{m-1}{2}+1\leq \frac{2}{3}m,\\
|E(\exc(\crv,G))| & = & m-(|E(C)|+1)\leq m-\left(\frac{m}{3}-1+1\right)=\frac{2}{3}m.
\end{eqnarray*}
Hence, $\crv$ satisfies the required properties.

From now on suppose that $|B_{x_0}|>1$. We apply Theorems~\ref{thm:sc-decomp} and~\ref{thm:bw-sqroot} to $H$ in order to find an sc-decomposition $\langle\tree,\iso,\noose\rangle$ of $H$ of width at most $\sqrt{4.5 |V(H)|}\leq \sqrt{4.5 n}$. Define a weight function $\wei\colon E(H)\to \N$ as follows: We first put $\wei(e)=1$ for all $e\in E(H)$. Then, for every bridge $uv$ such that $u\in B_{x_0}$ and $v\notin B_{x_0}$, consider the connected component $C$ of $G-uv$ that does not contain $B_{x_0}$ and add its edge count (including the bridge $uv$ itself) to the weight of one of the edges of $H$ incident to $u$; such an edge always exists for $H$ is bridgeless and has more than one vertex. Moreover, every vertex of $H$ is incident either to three edges of $H$ or to two edges of $H$ and one bridge, and hence it is easy to see that the distribution of weights can be done in such a manner that to the weight of every edge $e$ of $H$ we add the weight corresponding to at most one bridge incident to an endpoint of $e$. Observe that thus we have that $\wei(H)=m$, where the weight of a graph is defined as the sum of the weights of its edges. Moreover, since $G$ is $3$-regular and every edge of $T^*$ incident to $x_0$ was directed towards $x_0$, we have that $\wei(e)\leq 2+\frac{m-1}{2}=\frac{m}{2}+\frac{3}{2}$ for each $e\in E(H)$.

Now, consider any edge $e\in E(\tree)$. Let $(F_1,F_2)$ is the partition of $E(H)$ as in the definition of $\md(e)$, and let $H_1,H_2$ be the subgraphs of $H$ spanned by $F_1,F_2$; hence we have that one of $H_1,H_2$ is $\enc(\noose(e),H)$ and the second is $\exc(\noose(e),H)$, depending on the orientation of $\noose(e)$. For every $e\in E(\tree)$, let us select this orientation so that $\wei(\enc(\noose(e),H))\leq \wei(\exc(\noose(e),H))$, breaking ties arbitrarily in case of equality. Also, let us direct the edge $e$ in $\tree$ so that it points to the subtree of $\tree$ whose leaves span $\exc(\noose(e),G)$. Again, since after directing edges of $\tree$ the total sum of outdegrees in $|V(\tree)|-1$, there exists a node $z\in V(\tree)$ such that there is no outgoing edge from $z$ in $\tree$.

Consider first the case when $z$ is a leaf in $\tree$; let $uv=\iso(z)$ and let $e_z$ be the unique edge of $\tree$ incident to $z$. Hence, $\noose(e_z)$ is a noose that excludes only the edge $e_z$, and encloses the whole rest of $E(H)$. Since $e_z$ was directed towards $z$, we have that $\wei(uv)=\wei(\exc(\noose(e_z),H))\geq \frac{m}{2}$. Since $m\geq 6$, this means that one of the endpoints of $uv$, say $u$, is incident to a bridge $uw$ such that the weight corresponding to this bridge contributed to the weight of $uv$. More precisely, if $C$ is the connected component of $G-uw$ that does not contain $B_{x_0}$, then $\wei(uv)$ was increased from initial value of $1$ by contribution $1+|E(C)|$. Let $\crv$ be the length-$1$ noose that visits $u$ and the face around bridge $uw$, and encloses exactly $E(C)$ and the bridge $uw$. Then we have that $|E(\enc(\crv,G))|=|E(C)|+1=\wei(uv)-1\leq \frac{m}{2}+\frac{3}{2}-1\leq \frac{2}{3}m$, as $m\geq 6$. On the other hand, $|E(\exc(\crv,G))|=m-(|E(C)|+1)=m-(\wei(uv)-1)\leq m-\frac{m}{2}+1\leq \frac{2}{3}m$, again as $m\geq 6$. Hence, $\crv$ satisfies the required properties.

Finally, we are left with the case when $z$ is an internal vertex of $\tree$. Let $e_1,e_2,e_3$ be the edges of $\tree$ incident to $z$; recall that they are all directed towards $z$. For $t=1,2,3$, let $H_t=\enc(\noose(e_t),H)$. Then we have that $\wei(H_1)+\wei(H_2)+\wei(H_3)=\wei(H)=m$, and $\wei(H_t)\leq \frac{m}{2}$ for all $t=1,2,3$. W.l.o.g. suppose that $\wei(H_1)$ is the largest among $\wei(H_t)$ for $t=1,2,3$; then $\wei(H_1)\geq \frac{m}{3}$. Let $\crv'=\noose(e_1)$, which is a noose w.r.t. $H$; recall that the length of $\crv'$ is at most $\sqrt{4.5 |V(H)|}\leq \sqrt{4.5 n}$. Let us create a noose $\crv$ w.r.t. $G$ by modifying noose $\crv'$ as follows: whenever $\crv'$ traverses a vertex $u\in V(H)$ that in $G$ is incident to a bridge $uw$ connecting $u$ with a component $C$, then draw $\crv$ around $u$ in such a manner that it leaves $uw$ together with the whole component $C$ on the same side of $\crv$ as the edge of $H$ whose weight $C$ was contributing to (see Figure~\ref{fig:moving-noose}). Note that since $G$ is $3$-regular, this is always possible. In this manner, we have that $|E(\enc(\crv,G))|=\wei(H_1)$ and $|E(\exc(\crv,G))|=\wei(H_2)+\wei(H_3)$, and $\crv$ has the same length as $\crv'$.

\begin{figure}[htbp!]
                \centering
                \def\svgwidth{0.6\columnwidth}
                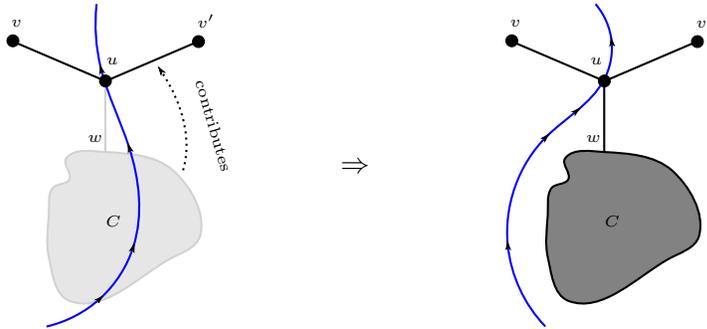
\caption{The modification applied to construct $\crv$ from $\crv'$ around vertices of $B_{x_0}$ that are incident to a bridge in $G$. The left panel shows noose $\crv'$ w.r.t. $H$, and the right panel shows the adjusted noose $\crv$ after reintroducing component $C$.}\label{fig:moving-noose}
\end{figure}

Of course, since $\wei(H_1)\leq \frac{m}{2}$, we have that $|E(\enc(\crv,G))|=\wei(H_1)\leq \frac{m}{2}\leq \frac{2}{3}m$. On the other hand, since $\wei(H_1)\geq \frac{m}{3}$, we have that $|E(\exc(\crv,G))|=m-\wei(H_1)\leq \frac{2}{3}m$. Thus, $\crv$ satisfies all the required properties.
\end{proof}

For the sake of later argumentation, we need a slightly different definition of balanceness that concerns faces instead of edges. We say that a noose $\crv$ is {\em{$\alpha$-face-balanced}} if the numbers of faces of $G$ that are strictly enclosed and strictly excluded by $\crv$, respectively, are not larger than $\alpha |F(G)|$. Here, $F(G)$ is the set of faces of $G$. Note that faces traversed by $\crv$ are neither strictly enclosed nor strictly excluded by $\crv$, and hence they contribute to neither of these numbers.

\begin{lemma}\label{lem:edge-to-faces}
Let $G$ be a connected 3-regular multigraph, possibly with loops, and embedded on a sphere $\Sigma$. If a noose $\crv$ w.r.t. $G$ is $\frac{2}{3}$-edge-balanced, then it is also $\frac{2}{3}$-face-balanced.
\end{lemma}
\begin{proof}
Since $G$ is connected and $3$-regular, by Euler's formula we have that $|V(G)|-|E(G)|+|F(G)|=2$ and by degree count we have that $|E(G)|=\frac{3}{2}|V(G)|$. From this it follows that 
\begin{equation}\label{eq:FG}
|F(G)|=2+\frac{1}{3}|E(G)|.
\end{equation}

As $G$ is connected and noose $\crv$ visits every face of $G$ at most once, it follows that both $G_1:=\enc(G,\crv)$ and $G_2:=\exc(G,\crv)$ are connected. Hence, $G_1$ is a connected plane multigraph, possibly with loops, with maximum degree at most $3$. From Euler's formula it follows that $|V(G_1)|-|E(G_1)|+|F(G_1)|=2$, and from the degree bound we infer that $|E(G_1)|\leq \frac{3}{2}|V(G_1)|$. Therefore,
\begin{eqnarray*}
|F(G_1)| & = & 2-|V(G_1)|+|E(G_1)|\leq 2-\frac{2}{3}|E(G_1)|+|E(G_1)|=2+\frac{1}{3}|E(G_1)|\\
& \leq & 2+\frac{1}{3}\cdot \frac{2}{3}|E(G)|= \frac{2}{3}+\frac{2}{3}\left(2+\frac{1}{3}|E(G)|\right)=\frac{2}{3}+\frac{2}{3}|F(G)|;
\end{eqnarray*}
here, the second inequality follows from the assumption that $\crv$ is $\frac{2}{3}$-edge-balanced, whereas the last equality follows from (\ref{eq:FG}). Observe now that the number of faces of $G$ strictly enclosed by $\crv$ is equal to $|F(G_1)|-1$; indeed the face set of $G_1$ comprises exactly the faces of $G$ strictly enclosed by $\crv$ plus one extra new face that contains the curve $\crv$. Hence, the number of faces of $G$ strictly enclosed by $\crv$ is at most $\frac{2}{3}+\frac{2}{3}|F(G)|-1<\frac{2}{3}|F(G)|$. A symmetric argument shows that also the number of faces strictly excluded by $\crv$ is at most $\frac{2}{3}|F(G)|$.
\end{proof}

From Theorem~\ref{thm:edge-balanced-noose} and Lemma~\ref{lem:edge-to-faces} we infer the following corollary.

\begin{corollary}\label{thm:face-balanced-noose}
Let $G$ be a connected 3-regular multigraph with $n$ vertices, $m\geq 6$ edges, possibly with loops, and embedded on a sphere $\Sigma$. Then there exists a $\frac{2}{3}$-face-balanced noose w.r.t. $G$ that has length at most $\sqrt{4.5 n}$.
\end{corollary}

\newcommand{\perm}{\Gamma}

\subsection{Voronoi separators}\label{sec:vorseps}

Armed with a good understanding of Voronoi diagrams and knowledge about the existence of short balanced nooses in plane graphs, we can combine these two ingredients to design a Divide\&Conquer algorithm for \covProb. More precisely, if $\SolObj$ is the optimum solution to the considered instance of \covProb, then the algorithm will iterate through all the possible candidates for balanced nooses of the Voronoi diagram $\dgm_\SolObj$. Each candidate noose will separate the instance into a number of subinstances, which will be solved recursively. For the correct selection of the balanced noose of $\dgm_\SolObj$, each of the subinstances will contain only at most $\frac{2}{3}k$ objects from the solution, and hence we will apply the algorithm recursively only to instances with at most this value of the parameter; This will ensure that the running time is as promised in Theorem~\ref{sec:algo}.

However, first we need to understand formally what are the ``possible candidates'' for balanced nooses, and in what sense they separate the instance at hand into subinstances; these two questions are the topics of the current and the following section. More precisely, we will investigate the properties of {\em{Voronoi separators}}, which are structures to be used as separators in the forthcoming Divide\&Conquer algorithm.

Whenever we have some normal subfamily of objects $\Fam\subseteq \Obj$, then we have a corresponding Voronoi partition $\Vorpar_{\Fam}$ and Voronoi diagram $\dgm_{\Fam}$. The intuition now is that the nooses w.r.t. $\dgm_{\Fam}$ can be projected back to the original graph $G$. That is, suppose in some noose $\crv$ we travel from some center of a face $\cn(f)$ to an adjacent vertex $u$. Then in $G$ this corresponds to traveling from the center $\cen(p)$ of the corresponding object $p$ to the corresponding branching point of $\dgm_{\Fam}$ using a path inside the tree $\extree_{\Fam}(p)$. The crucial observation now is that, for a noose $\crv$, the knowledge of the order in which $\crv$ visits faces of $\dgm_{\Fam}$ (corresponding to objects of $\Fam$) and branching points of $\dgm_{\Fam}$, is sufficient to deduce the whole projection to $G$. This is because when projecting we can simply use shortest paths.

\paragraph*{Voronoi separators.}

We move on to a formal argumentation. A {\em{Voronoi separator}} of {\em{length}} $r$ is a sequence $\sep$ of the following form:
\begin{eqnarray}
\sep=\langle & p_1,u_1,f_1,v_1, & \nonumber \\
& p_2,u_2,f_2,v_2, & \nonumber \\
& \ldots & \nonumber \\
& p_r,u_r,f_r,v_r & \rangle\label{vorsep}
\end{eqnarray}
For this sequence, we put the following requirements (from now on, indices behave cyclically):
\begin{enumerate}[label=(\alph*)]
\item\label{vs1} $p_1,p_2,\ldots,p_r$ are pairwise disjoint objects from $\Obj$, forming a normal family $\Obj(\sep)=\{p_1,p_2,\ldots,p_r\}$.
\item\label{vs2} $f_1,f_2,\ldots,f_r$ are pairwise different faces of $G$.
\item\label{vs3} For $t=1,2,\ldots,r$, $u_t$ and $v_t$ are two different vertices on face $f_t$. Moreover, $v_{t-1},u_t\in \Vorpar_{\Obj(\sep)}(p_t)$, for all $t=1,2,\ldots,r$. In other words, $v_{t-1}$ and $u_t$ are closest to the location of $p_t$ in terms of the distance measure $\phi(v,p)=\dist(v,\loc(p))-\rad(p)$, among all the locations of objects traversed by $\sep$.
\end{enumerate}
We shall treat two Voronoi separators that differ only in a cyclic shift of the sequence as the same separator. However, similarly as with nooses, reverting the sequence results in a different separator, which we shall denote by $\sep^{-1}$. Observe that, given a sequence like in \eqref{vorsep}, we can verify in polynomial time whether it satisfies all the required properties of a Voronoi separator.

For $t=1,2,\ldots,r$, let $P_t$ be the path from $\cen(p_t)$ to $u_t$ obtained by concatenating a path inside $\tree(p_t)$ from $\cen(p_t)$ to the vertex of $\loc(p_t)$ closest to $u_t$, with the shortest path from $u_t$ to $\loc(p_t)$. Similarly define $Q_{t+1}$ for $v_t$ and $p_{t+1}$. By property \ref{vs3}, paths $P_t$ and $Q_t$ are subpaths of the tree $\extree_{\Obj(\sep)}(p_t)$, for every $t=1,2,\ldots,t$. Consequently, paths $\{P_t,Q_t\}_{t=1,2,\ldots,r}$ are all pairwise vertex-disjoint, apart from paths $P_t$ and $Q_t$ that can intersect on a common prefix from the side of $\cen(p_t)$, and otherwise they are vertex-disjoint.

For a Voronoi separator $\sep$, the {\em{perimeter}} of $\sep$ is defined as $\perm(\sep)=\bigcup_{t=1}^r V(P_t)\cup V(Q_t)$, i.e., it is the union of the vertex sets of all the paths $P_t$ and $Q_t$. Moreover, we define a directed closed walk $\walk(\sep)$ in $G$ by concatenating consecutive paths $Q_1,P_1,Q_2,P_2,\ldots,Q_t,P_t$, where $P_t$ and $Q_{t+1}$ are joined using edge $u_tv_t$, which lies on the boundary of face $f_t$. Note that $V(\walk(\sep))=\perm(\sep)$. See Figure~\ref{fig:vorsep} for an example. 

\begin{figure}[t]
                \centering
                \def\svgwidth{0.5\columnwidth}
                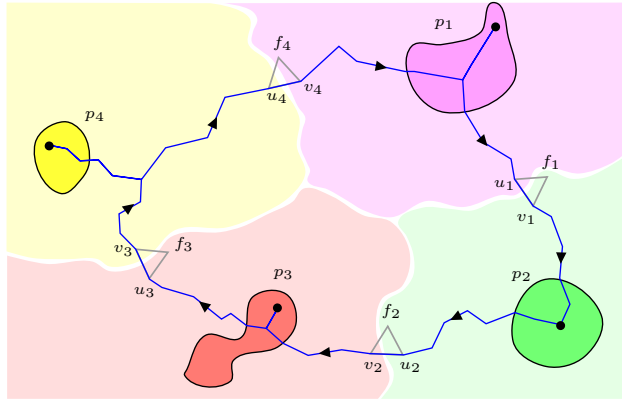
\caption{An exemplary Voronoi separator $\sep$ of length $4$, together with the corresponding walk $\walk(\sep)$ (blue). The Voronoi regions of $\Vorpar_{\Obj(\sep)}$ have been depicted in respective light colors.}\label{fig:vorsep}
\end{figure}

Note that, by the definition of a Voronoi separator, $\walk(\sep)$ is a simple cycle in $G$, possibly with simple paths attached to different vertices that correspond to common prefixes of paths $P_t$ and $Q_t$. In particular, it still holds that removal of $\walk(\sep)$ from $\Sigma$ partitions $\Sigma$ into two open disks, one on the right side of $\walk(\sep)$ (denoted $\enc(\sep)$), and the second on the left side (denoted $\exc(\sep)$). For any object $p\in \Obj\setminus \Obj(\sep)$, we shall say that $p$ is {\em{strictly enclosed}} by $\sep$ if all the vertices of $\loc(p)$ are embedded into $\enc(\sep)$ (in particular $\loc(p)\cap \perm(\sep)=\emptyset$). Similarly, $p\in \Obj\setminus \Obj(\sep)$ is {\em{strictly excluded}} by $\sep$ if all the vertices of $\loc(p)$ are embedded into $\exc(\sep)$. Finally, a client $q\in \Cli$ is {\em{enclosed}} by $\sep$, resp. {\em{excluded}} by $\sep$, if $\pla(q)$ is embedded into $\overline{\enc(\sep)}$, resp. into $\overline{\exc(\sep)}$; Note that thus every client $q$ such that $\pla(q)\in \perm(\sep)$ is both enclosed and excluded by $\sep$. Observe also that any path in $G$ that connects the location of a strictly enclosed object with the location of another, strictly excluded object, must necessarily cross walk $\walk(\sep)$, so it has to traverse a vertex of $\perm(\sep)$. The same observation holds also for paths connecting the location of a strictly enclosed object with the placement of an excluded client, and vice versa.

\paragraph*{Separators inherited from the diagram.}

Assume now that we have a normal subfamily of objects $\Fam\subseteq \Obj$, and let $\Vorpar=\Vorpar_{\Fam}$ and $\dgm=\dgm_{\Fam}$ be the corresponding Voronoi partition and diagram. Suppose further that $\crv$ is a noose w.r.t. $\dgm$. We define the corresponding Voronoi separator $\sep(\crv)$ as follows; see Figure~\ref{fig:prdg-vorsep} for a visualization.

\begin{figure}[t]
                \centering
                \def\svgwidth{0.7\columnwidth}
                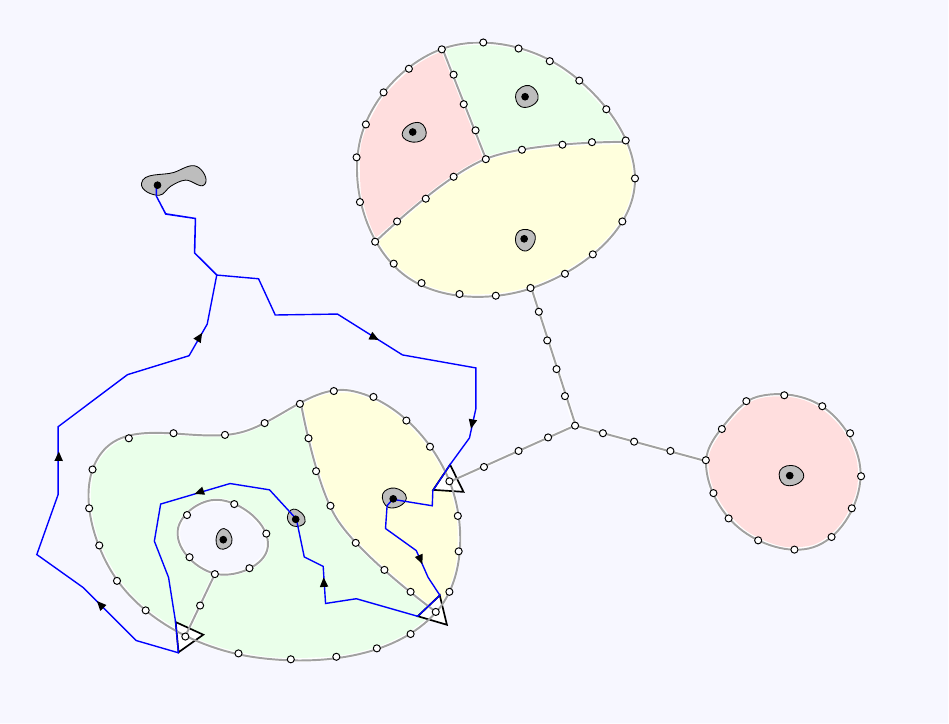
\caption{Separator $\sep(\crv)$ and corresponding walk $\walk(\sep(\crv))$ for the noose $\crv$ from Figure~\ref{fig:noose}.}\label{fig:prdg-vorsep}
\end{figure}

Suppose that $\crv$ visits faces and branching points $f_1^*,f_1,f_2^*,f_2,\ldots,f_r^*,f_r$ in this order, where $f_t^*$ are different faces of $\dgm$, and $f_t$ are different branching points of $\dgm$, that is, faces of $G$ (recall that $\dgm$ is constructed from the dual of $G$). By Lemmas~\ref{lem:pdiagram} and~\ref{lem:diagram}, faces of $\dgm$ correspond one-to-one to objects of $\Fam$. Let $p_1,p_2,\ldots,p_r\in \Fam$ be the objects that correspond to $f_1^*,f_2^*,\ldots,f_r^*$, respectively. Suppose that when entering $f_t$ from $f_t^*$, for some $t\in \{1,2,\ldots,r\}$, noose $\crv$ entered it between edges $\tilde{e}_1^*$ and $\tilde{e}_2^*$. Let $e_1^*$ and $e_2^*$ be the edges of the prediagram $\pdgm_{\Fam}$ incident to $f_t$ that got contracted onto $\tilde{e}_1^*$ and $\tilde{e}_2^*$ respectively; in case $\tilde{e}_1^*$ or $\tilde{e}_2^*$ is a loop, we choose the edge of $\pdgm_{\Fam}$ corresponding to the endpoint of the loop by which $f_t$ was entered from $f_t^*$. Then we define $u_t$ to be the common endpoint of the primal edges $e_1$ and $e_2$ corresponding to $e_1^*$ and $e_2^*$, respectively (recall that $f$ is a triangle, since $G$ is triangulated). We define $v_t$ analogically, based on how $\crv$ leaves branching point $f_t$ to face $f_{t+1}^*$. Note that since a noose, when crossing some branching point, never leaves this branching point to the same face and using the same direction as it entered, we have that $u_t\neq v_t$ for all $t=1,2,\ldots,r$. Then separator $\sep(\crv)$ is defined using $p_t$, $f_t$, $u_t$, and $v_t$ as in formula \eqref{vorsep}. Note that $\sep(\crv)$ has the same length as $\crv$, and moreover, $\sep(\crv)^{-1}=\sep(\crv^{-1})$. The following lemma encapsulates the main properties of separators inherited from the diagram that we shall use later on.

\begin{lemma}\label{lem:inherited-properties}
If $\crv$ is a noose w.r.t. $\dgm$, then:
\begin{enumerate}[label=(\roman*)]
\item\label{p:issep} $\sep:=\sep(\crv)$ is a Voronoi separator;
\item\label{p:vordgm} $\perm(\sep)\subseteq \bigcup_{p\in \Obj(\sep)} \Vorpar_{\Fam}(p)$;
\item\label{p:encl} For every object $p\in \Fam$ such that the face of $\dgm$ corresponding to $p$ is strictly enclosed by $\crv$ (resp. strictly excluded by $\crv$), we have that $p$ is strictly enclosed by $\sep$ (resp. strictly excluded by $\sep$).
\end{enumerate}
\end{lemma}
\begin{proof}
For \ref{p:issep}, the only non-trivial condition is the last one, that $u_t,v_{t-1}\in \Vorpar_{\Obj(\sep)}(p_t)$ for every $t=1,2,\ldots,r$. However, from the construction of $u_t$ and $v_{t-1}$ and properties of the Voronoi diagram (see Lemmas~\ref{lem:pdiagram} and~\ref{lem:diagram}), it follows that $u_t,v_{t-1}\in \Vorpar_{\Fam}(p_t)$. Since $\Obj(\sep)\subseteq \Fam$, the fact that also $u_t,v_{t-1}\in \Vorpar_{\Obj(\sep)}(p_t)$ follows from Lemma~\ref{lem:inclusion}.

For \ref{p:vordgm}, in the paragraph above we have argued that $u_t,v_{t-1}\in \Vorpar_{\Fam}(p_t)$ for all $t=1,2,\ldots,r$. By Lemma~\ref{lem:connectivity}, also the unique shortest paths from $u_t$ and $v_{t-1}$ to $\loc(p_t)$ are entirely contained in $\Vorpar_{\Fam}(p_t)$. As $\loc(p_t)\subseteq \Vorpar_{\Fam}(p_t)$, we have that $V(Q_t),V(P_t)\subseteq \Vorpar_{\Fam}(p_t)$, and hence $\perm(\sep)\subseteq \bigcup_{p\in \Obj(\sep)} \Vorpar_{\Fam}(p)$.

Property \ref{p:encl} follows directly from \ref{p:vordgm} and the construction of the Voronoi diagram $\dgm$.
\end{proof}

\subsection{Interaction graph and separability of the problem}

\newcommand{\IntGraph}{\mathcal{L}}
\newcommand{\banned}{\mathbf{ban}}
\newcommand{\covered}{\mathbf{cov}}
\newcommand{\grds}{\mathcal{Q}}
\newcommand{\grdsep}{\mathbb{X}}

Suppose $(G,\Obj,\Cli,k)$ is an instance of \covProb. By the {\em{interaction graph}} $\IntGraph=\IntGraph(G,\Obj,\Cli,k)$ we mean a graph with vertex set $\Obj\cup \Cli$, where (a) for $p\in \Obj$ and $q\in \Cli$ we put $pq\in E(\IntGraph)$ if and only if object $p$ covers client $q$, and (b) for $p,p'\in \Obj$ we put $pp'\in E(\IntGraph)$ if and only if $p$ and $p'$ contradict the normality requirement, i.e., family $\{p,p'\}$ is not normal.

We now define the final abstraction of a separator, which we call a {\em{guarded separator}}. A guarded separator is simply a pair $\grdsep=(\grds,\perm)$, where $\grds\subseteq \Obj$ is a normal subfamily of objects and $\perm\subseteq V(G)$ is an arbitrary subset of vertices. We will use operators $\grds(\grdsep)$ and $\perm(\grdsep)$ to extract the first and the second coordinate of pair $\grdsep$, respectively. The {\em{length}} of a guarded separator is simply $|\grds(\grdsep)|$. Note that every Voronoi separator $\sep$ naturally induces a guarded separator $\grdsep(\sep)=(\Obj(\sep),\perm(\sep))$; this will be the main source of guarded separators in our algorithm.

Let $\grdsep=(\grds,\perm)$ be a guarded separator. We say that a client $q\in \Cli$ is {\em{covered}} by $\grdsep$ if $q$ is covered by an object belonging to $\grds$. We also say that an object $p\in \Obj\setminus \grds$ is {\em{banned}} by $\grdsep$ if either $\grds\cup \{p\}$ is not normal, or there exists $v\in \perm$ such that $\dist(v,\loc(p))-\rad(p)<\dist(v,\loc(p'))-\rad(p')$, where $p'$ is the object of $\grds$ closest to $v$ in terms of the distance measure $\phi(v,p')=\dist(v,\loc(p'))-\rad(p')$. An object $p\in \Obj\setminus \grds$ that is not banned by $\grdsep$ is said to be {\em{allowed}} by $p$. By $\covered(\grdsep)$ we denote the set of clients covered by $\grdsep$, and by $\banned(\grdsep)$ the set of objects banned by $\grdsep$. Obviously, by Lemma~\ref{lem:sanity} every object $p\in \Obj\setminus \grds$ such that $\loc(p)\cap \perm\neq \emptyset$ is banned by $\grdsep$. 

%Let $\grdsep:=\grdsep(\sep)$ be the guarded separator induced by some Voronoi separator $\sep$. Since every object $p\in \Obj\setminus \grds$ whose location intersects $\perm(\sep)$ is banned by $\grdsep$, all the locations of objects allowed by $\grdsep$ are either strictly enclosed, or strictly excluded by $\sep$. We shall call the former {\em{allowed objects enclosed by $\sep$}} and the latter {\em{allowed objects excluded by $\sep$}}.

We say that two allowed objects $p,p'\in \Obj\setminus \grds$ are {\em{separated}} by $\grdsep$ if the shortest path in $G$ between $\loc(p)$ and $\loc(p')$ traverses a vertex of $\perm(\grdsep)$. Similarly, an allowed object $p\in \Obj\setminus \grds$ and a client $q\in \Cli$ are {\em{separated}} by $\grdsep$ if the shortest path between $\loc(p)$ and $\pla(q)$ traverses a vertex of $\perm(\grdsep)$. 

The intuition is that if a Voronoi separator $\sep$ is inherited from some solution, then the objects banned by $\grdsep(\sep)$ for sure are not used by the solution. This is because the inclusion of any banned object would either contradict the normality, or the fact that walk $\walk(\sep)$ traverses only the Voronoi regions corresponding to $\Obj(\sep)$ in the diagram induced by the whole solution. Thus, the clients covered by $\sep$ and the object banned by $\sep$ form a ``border'' that separates the part of the solution enclosed by $\sep$ from the part excluded by $\sep$. Here, by ``separates'' we mean the definition of separation for the guarded separator $\grdsep(\sep)$. In the following two lemmas we formalize this separation property.

\begin{lemma}\label{lem:sep-oc}
Suppose $\grdsep=(\grds,\perm)$ is a guarded separator. Suppose further that an allowed object $p\in \Obj\setminus \grds$ and a client $q\in \Cli$ are separated by $\grdsep$. Then the following implication holds: if $p$ covers $q$, then $q$ is covered by $\grdsep$.
\end{lemma}
\begin{proof}
Suppose $p$ covers $q$, and let $P$ be the shortest path in $G$ between $\loc(p)$ and $\pla(q)$. By the definition of separation, we have that $P$ traverses a vertex of $\perm$, say $v$. Let $p'$ be the object that is closest to $v$ among objects of $\grds$ in terms of $\phi(v,p')$. Since $p$ is allowed by $\grdsep$, we have that 
\begin{equation}\label{oc1}
\dist(v,\loc(p))-\rad(p)>\dist(v,\loc(p'))-\rad(p').
\end{equation}
Since $p$ covers $q$, we also have that 
\begin{equation}\label{oc2}
\dist(\pla(q),\loc(p))\leq \sen(q)+\rad(p).
\end{equation}
Finally, since $P$ is the shortest path between $\loc(p)$ and $\pla(q)$, we have that
\begin{equation}\label{oc3}
\dist(\pla(q),\loc(p))=\dist(v,\loc(p))+\dist(\pla(q),v).
\end{equation}
Using (\ref{oc1}), (\ref{oc2}), and (\ref{oc3}), we infer that
\begin{eqnarray*}
\dist(\pla(q),\loc(p')) & \leq & \dist(v,\pla(q))+\dist(v,\loc(p'))\\
& < & \dist(v,\pla(q))+\dist(v,\loc(p))+\rad(p')-\rad(p)\\
& = & \dist(\pla(q),\loc(p))+\rad(p')-\rad(p)\\
& \leq & \sen(q)+\rad(p').
\end{eqnarray*}
This means that client $q$ is covered by $p'\in \grds$.
\end{proof}

\begin{lemma}\label{lem:sep-oo}
Suppose $\grdsep=(\grds,\perm)$ is a guarded separator. Suppose further that allowed objects $p_1,p_2\in \Obj\setminus \grds$ are separated by $\grdsep$. Then the family $\{p_1,p_2\}$ is normal.
\end{lemma}
\begin{proof}
Let $P$ be the shortest path in $G$ between $\loc(p_1)$ and $\loc(p_2)$. By the definition of separation, we have that $P$ traverses some vertex of $\perm$, say $v$. Let $p'$ be the object of $\grds$ that is closest to $v$, in terms of $\phi(v,p')$. Since both $p_1$ and $p_2$ are allowed by $\grdsep$, we have that
\begin{align}
\dist(\loc(p_1),v)+\dist(\loc(p'),v)&\geq \dist(\loc(p_1),\loc(p')) \geq |\rad(p_1)-\rad(p')|\label{oo1}\\
\dist(\loc(p_2),v)+\dist(\loc(p'),v)&\geq \dist(\loc(p_2),\loc(p')) \geq |\rad(p_2)-\rad(p')|\label{oo2}\\
\dist(\loc(p_1),v)-\rad(p_1)&> \dist(\loc(p'),v)-\rad(p')\label{oo3}\\
\dist(\loc(p_2),v)-\rad(p_2)&> \dist(\loc(p'),v)-\rad(p').\label{oo4}
\end{align}
Since $P$ is the shortest path between $\loc(p_1)$ and $\loc(p_2)$, we have
\begin{equation}\label{oo5}
\dist(\loc(p_1),\loc(p_2))=\dist(v,\loc(p_1))+\dist(v,\loc(p_2)).
\end{equation}
By adding (\ref{oo2}) and (\ref{oo3}), we obtain:
\begin{eqnarray*}
& & \dist(\loc(p_1),v)+\dist(\loc(p_2),v)+\dist(\loc(p'),v)-\rad(p_1)\\
& >& |\rad(p_2)-\rad(p')|+\dist(\loc(p'),v)-\rad(p')\\
& \geq & \dist(\loc(p'),v)-\rad(p_2).
\end{eqnarray*}
Thus, by (\ref{oo5}) we obtain that
\begin{equation}\label{oo6}
\dist(\loc(p_1),\loc(p_2))\geq \rad(p_1)-\rad(p_2).
\end{equation}
By adding equations (\ref{oo1}) and (\ref{oo4}), we obtain:
\begin{eqnarray*}
& & \dist(\loc(p_1),v)+\dist(\loc(p_2),v)+\dist(\loc(p'),v)-\rad(p_2)\\
& > & |\rad(p_1)-\rad(p')|+\dist(\loc(p'),v)-\rad(p')\\
& \geq & \dist(\loc(p'),v)-\rad(p_1)
\end{eqnarray*}
Again by (\ref{oo5}) we obtain that
\begin{equation}\label{oo7}
\dist(\loc(p_1),\loc(p_2))\geq \rad(p_2)-\rad(p_1).
\end{equation}
Equations (\ref{oo6}) and (\ref{oo7}) together imply that $\dist(\loc(p_1),\loc(p_2))\geq |\rad(p_1)-\rad(p_2)|$, which means that $\{p_1,p_2\}$ is normal.
\end{proof}

For a guarded separator $\grdsep$, we will consider the modified interaction graph $\IntGraph(\grdsep)=\IntGraph-(\grds(\grdsep)\cup \banned(\grdsep)\cup \covered(\grdsep))$. Let $\Fam\subseteq \Obj$ be a normal subfamily of objects, and let $\Vorpar=\Vorpar_{\Fam}$ and $\dgm=\dgm_{\Fam}$ be the corresponding Voronoi partition and diagram. We will say that a guarded separator $\grdsep$ is {\em{compatible}} with $\Fam$ if $\grds(\grdsep)\subseteq \Fam$ and $\perm(\grdsep)\subseteq \bigcup_{p\in \grds(\grdsep)} \Vorpar_{\Fam}(p)$. We will moreover say that a guarded separator $\grdsep$ compatible with $\Fam$ is {\em{$\alpha$-interaction-balanced w.r.t. $\Fam$}} if every connected component of $\IntGraph(\grdsep)$ contains at most $\alpha|\Fam|$ objects from $\Fam$. Observe that if $\grdsep$ is compatible with $\Fam$, then $\Fam\cap \banned(\grdsep)=\emptyset$ by the definitions of banned objects and compatibility.

\begin{lemma}\label{lem:noose-guarded-sep}
If $\crv$ is an $\alpha$-face-balanced noose in $\dgm$, then $\grdsep(\sep(\crv))$ is compatible and $\alpha$-interaction-balanced w.r.t. $\Fam$.
\end{lemma}
\begin{proof}
Let $\sep=\sep(\crv)$ and $\grdsep=\grdsep(\sep)$. Lemma~\ref{lem:inherited-properties} (\ref{p:issep} and \ref{p:vordgm}) implies that $\grdsep$ is compatible with $\Fam$, so by the definition of banned objects we have that $\Fam\cap \banned(\grdsep)=\emptyset$. Hence $\Fam\setminus \Obj(\sep)\subseteq V(\IntGraph(\grdsep))$. 

Observe that whenever we have two allowed objects $p_1,p_2\in \Fam\setminus \Obj(\sep)$ such that $p_1$ is strictly enclosed by $\sep$ and $p_2$ is strictly excluded by $\sep$, then $\loc(p_1)$ and $\loc(p_2)$ lie entirely in two different connected components of $G-\perm(\sep)$, since any path connecting $\loc(p_1)$ and $\loc(p_2)$ has to cross the walk $\walk(\sep)$. Therefore, any such objects $p_1,p_2$ are separated by $\grdsep$, and Lemma~\ref{lem:sep-oo} implies that they are not adjacent in $\IntGraph(\grdsep)$. Similarly, if an allowed object $p\in \Fam\setminus \Obj(\sep)$ is strictly enclosed by $\sep$ and a client $q\in \Cli$ is excluded by $\sep$ (or vice versa), then they also lie in different connected components of $G-\perm(\sep)$ (or $q$ lies on $\perm(\sep)$), and hence are separated by $\grdsep$. Then Lemma~\ref{lem:sep-oc} implies that $p$ and $q$ are not adjacent in $\IntGraph(\grdsep)$.

Concluding, every connected component of $\IntGraph(\grdsep)$ either consists only of objects strictly enclosed by $\sep$ and clients enclosed by $\sep$, or of objects strictly excluded by $\sep$ and clients excluded by $\sep$. Lemma~\ref{lem:inherited-properties} \ref{p:encl} ensures that objects of $\Fam$ corresponding to faces of $\dgm$ strictly enclosed (resp. excluded) by $\crv$ are exactly the objects of $\Fam$ that are strictly enclosed (resp. excluded) by $\sep$. Since $\crv$ is $\alpha$-face-balanced, we infer that there are at most $\alpha|\Fam|$ objects of $\Fam$ that are strictly enclosed by $\sep$, and the same holds also for objects strictly excluded by $\Fam$. Hence, every connected component of $\IntGraph(\grdsep)$ can contain at most $\alpha|\Fam|$ objects of $\Fam$, and $\grdsep$ is $\alpha$-interaction-balanced. 
\end{proof}

We are finally ready to prove the main result of this section, that is, an enumeration algorithm for candidates for balanced guarded separators. This is the result whose simplified variant was Lemma~\ref{lem:guardedenum0}.

\begin{theorem}\label{thm:enumeration}
There exists an algorithm that, given an instance $\Ii=(G,\Obj,\Cli,k)$ of \covProb with $k\geq 4$, enumerates a family $\Vorsepfam$ of guarded separators with the following properties:
\begin{enumerate}[label=(\roman*)]
\item\label{p:size} $|\Vorsepfam|\leq (2\onum)^{15\sqrt{k}}$; and
\item\label{p:balanced} for every normal subfamily $\Fam\subseteq \Obj$ of cardinality exactly $k$, there exists a guarded separator $\grdsep\in \Vorsepfam$ that is $\frac{2}{3}$-interaction-balanced for $\Fam$.
\end{enumerate}
The algorithm works in total time $(2\onum)^{15\sqrt{k}}\cdot (\onum \cnum n)^{\Oh(1)}$ and outputs the guarded separators of $\Vorsepfam$ one by one, using additional (working tape) space $(\onum \cnum n)^{\Oh(1)}$.
\end{theorem}
\begin{proof}
We first apply the algorithm of Theorem~\ref{thm:impFaces} to compute a family $\impFaces$ of important faces of size at most $\onum^4$. Then, we enumerate all candidates sequences for Voronoi separators of length at most $3\sqrt{k}$ as in formula (\ref{vorsep}), where faces $f_t$ are chosen from the family $\impFaces$. Note that for a fixed length $r$ ($1\leq r\leq 3\sqrt{k}$), we have at most $\onum^r\cdot \onum^{4r}\cdot 6^r$ such candidates, since every object $p_t$ is chosen among $\onum$ options, every face $f_t$ is chosen among at most $\onum^4$ options, and for choosing vertices $u_t,v_t$ on face $f_t$ we have $6$ options. Observe that
$$\sum_{r=1}^{\lfloor 3\sqrt{k}\rfloor} (6\onum^5)^r\leq \sum_{r=1}^{\lfloor 3\sqrt{k}\rfloor} \frac{1}{4}(2\onum)^{5r}\leq (2\onum)^{15\sqrt{k}}.$$
As family $\Vorsepfam$, we output guarded separators $\grdsep(\sep)$ for all the candidates $\sep$ that are indeed a Voronoi separator; recall that this property can be checked in polynomial time, as well as the construction of $\grdsep(\sep)$ takes polynomial time. Thus we have $|\Vorsepfam|\leq (2\onum)^{15\sqrt{k}}$, which is exactly property \ref{p:size}. Moreover, $\Vorsepfam$ can be enumerated by examining the candidates one by one within the required time and additional space. We are left with proving property \ref{p:balanced}.

Take any normal subfamily $\Fam\subseteq \Obj$ with $|\Fam|=k$, and let $\Vorpar=\Vorpar_{\Fam}$ and $\dgm=\dgm_{\Fam}$ be the corresponding Voronoi partition and diagram. By Lemma~\ref{lem:diagram}, we have that $\dgm$ is a connected, $3$-regular multigraph (possibly with loops), $|V(\dgm)|=2|\Fam|-4=2k-4$, and $|E(\dgm)|=3|\Fam|-6=3k-6$. Since $k\geq 4$, we have that $|E(\dgm)|\geq 6$ and Corollary~\ref{thm:face-balanced-noose} asserts that there exists a $\frac{2}{3}$-face-balanced noose $\crv$ w.r.t. $\dgm$ that has length at most $\sqrt{4.5 |V(\dgm)|}\leq 3\sqrt{k}$. By Lemma~\ref{lem:inherited-properties}\ref{p:issep}, $\sep(\crv)$ is a Voronoi separator of the same length as $\crv$. Since $\Fam\subseteq \Obj$ is a normal subfamily, by Theorem~\ref{thm:impFaces} we have that $V(\dgm)\subseteq \impFaces$, and hence $\sep(\crv)$ uses only faces from $\impFaces$. This means that $\grdsep(\sep(\crv))\in \Vorsepfam$. By Lemma~\ref{lem:noose-guarded-sep}, we have that guarded separator $\grdsep(\sep(\crv))$ is $\frac{2}{3}$-interaction-balanced w.r.t. $\Fam$, and hence property \ref{p:balanced} is proven.
\end{proof}

\newcommand{\cc}{\textrm{cc}}
\newcommand{\comp}{\textrm{comp}}
\newcommand{\kv}{\mathbf{k}}
\newcommand{\AlgName}{\mathtt{SolveDNC}}

\subsection{The algorithm}

We first briefly discuss the intuition. Suppose $\SolObj$ is a solution to the considered instance $(G,\Obj,\Cli,k)$, i.e., it is a normal family of exactly $k$ objects from $\Obj$. We can assume that $k\geq 4$, since otherwise the instance can be solved in polynomial time by brute force. Theorem~\ref{thm:enumeration} gives us a method to enumerate a small family of candidate guarded separators with a promise, that one of them separates the optimum solution evenly. More precisely, there is a guarded separator whose objects all belong to the optimum solution $\SolObj$, and moreover after inferring all the information from this fact, i.e., excluding banned objects and covered clients, we arrive at the situation where every connected component of the remaining interaction graph contains only two thirds of the objects of $\SolObj$. This gives rise to a natural Divide\&Conquer algorithm: For every guarded separator $\grdsep$ output by Theorem~\ref{thm:enumeration}, we recurse into all the components of the modified interaction graph $\IntGraph(\grdsep)$ for all parameters between $0$ and $\lfloor\frac{2}{3}k\rfloor$, and then combine the results using a standard knapsack dynamic programming. Note that in the new subinstances solved recursively we only modify the sets of objects and clients, whereas the underlying graph $G$ remains unchanged. We now proceed to implementing this plan formally.

Let us first introduce some notation. For an instance $\Ii=(G,\Obj,\Cli,k)$ of \covProb, by $\Val[\Ii]$ we denote the maximum revenue of a solution in $\Ii$ that should be reported by the algorithm. When $\SolObj\subseteq \Obj$ is a normal subfamily of objects, by $\Val[\Ii,\SolObj]$ we denote the revenue of $\SolObj$ in $\Ii$. By $\Vorsepfam$ we denote the family of guarded separators computed for instance $\Ii$ using the algorithm of Theorem~\ref{thm:enumeration}. Suppose $\grdsep$ is some guarded separator in $\Vorsepfam$. By $\cc(\grdsep)$ we denote the set of connected components of $\IntGraph(\grdsep)=\IntGraph-(\grds(\grdsep)\cup \banned(\grdsep)\cup \covered(\grdsep))$, where $\IntGraph$ is the interaction graph of instance $\Ii$. For every $C\in \cc(\grdsep)$, by $\Obj(C)$ and $\Cli(C)$ we denote the sets of objects and clients in $C$, respectively. Finally, we shall say that a vector of nonnegative integers $\kv=(k_C)_{C\in \cc(\grdsep)}$ is {\em{compatible with $\grdsep$}} if $k_C\leq \frac{2}{3}k$ for each $C\in \cc(\grdsep)$ and $\sum_{C\in \cc(\grdsep)}k_C=k-|\grds(\grdsep)|$. The set of vectors compatible with $\grdsep$ will be denoted by $\comp(\grdsep)$.

The following lemma is the main argument justifying the correctness of the algorithm.

\begin{lemma}\label{lem:main-correctness}
Provided $k\geq 4$, the following recursive formula holds for every instance $(G,\Obj,\Cli,k)$ of \covProb:
\begin{equation}\label{eq:corr}
\Val[G,\Obj,\Cli,k]=\max_{\grdsep\in \Vorsepfam}\ \max_{\kv\in \comp(\grdsep)}\  \Pri(\grds(\grdsep))+\sum_{C\in \cc(\grdsep)}\Val[G,\Obj(C),\Cli(C),k_C] 
\end{equation}
\end{lemma}
\begin{proof}
Let $\Ii=(G,\Obj,\Cli,k)$. For a connected component $C$ on the right-hand side we shall denote $\Ii_C=(G,\Obj(C),\Cli(C),k_C)$. Let $L$ and $R$ denote the left- and right-hand side of formula (\ref{eq:corr}), respectively. 

We first show that $L\leq R$. If $L=\minf$ then the inequality is obvious, so suppose that $\SolObj\subseteq \Obj$ is a normal subfamily of cardinality exactly $k$ that maximizes the revenue on the left-hand side, i.e., $\Val[\Ii]=\Val[\Ii,\SolObj]$. By Theorem~\ref{thm:enumeration}, there exists a guarded separator $\grdsep\in \Vorsepfam$ that is $\frac{2}{3}$-interaction-balanced w.r.t. $\SolObj$. In particular $\grdsep$ is compatible with $\SolObj$, so $\banned(\grdsep)\cap \SolObj=\emptyset$.

For a connected component $C\in \cc(\grdsep)$, let $\SolObj_C=\SolObj\cap \Obj(C)$ and $k_C=|\SolObj_C|$; observe that $\SolObj_C$ is a normal subfamily of $\Obj(C)$. Since $\banned(\grdsep)\cap \SolObj=\emptyset$, we infer that $\SolObj$ is a disjoint union of $\grds(\grdsep)$ and sets $\SolObj_C$ for $C\in \cc(\grdsep)$. Hence $\sum_{C\in \cc(\grdsep)}k_C=k-|\grds(\grdsep)|$. Moreover, since $\grdsep$ is $\frac{2}{3}$-interaction-balanced w.r.t. $\SolObj$, we have that $k_C\leq \frac{2}{3}k$ for all $C\in \cc(\grdsep)$. This implies that vector $\kv:=(k_C)_{C\in \cc(\grdsep)}$ is compatible with $\grdsep$. 

Observe now that $\Val[\Ii,\SolObj]=\Pri(\grds(\grdsep))+\sum_{C\in \cc(\grdsep)}\Val[\Ii_C,\SolObj_C]$, where we denote $\Ii_C=(G,\Obj(C),\Cli(C),k_C)$. Indeed, the cost of every object of $\SolObj$ is counted once on both sides of the equality, and the same holds also for the prizes of clients covered by $\SolObj$: Prize of a client covered by an object of $\grds(\grdsep)$ is counted only in the term $\Pri(\grds(\grdsep))$ due to removing these clients when constructing $\IntGraph(\grdsep)$, and every client not covered by $\grds(\grdsep)$ but covered by $\SolObj$ is counted only in the term $\Val[\Ii_C,\SolObj_C]$ for the component $C$ it belongs to. Recall also that $\Val[\Ii]=\Val[\Ii,\SolObj]$ and $\Val[\Ii_C]\geq \Val[\Ii_C,\SolObj_C]$ by the definition of $\Val[\Ii_C]$. Therefore, we have:
\begin{equation*}
L = \Val[\Ii]=\Val[\Ii,\SolObj]=\Pri(\grds(\grdsep))+\sum_{C\in \cc(\grdsep)}\Val[\Ii_C,\SolObj_C]\leq \Pri(\grds(\grdsep))+\sum_{C\in \cc(\grdsep)}\Val[\Ii_C] \leq R;
\end{equation*}
the last inequality follows from the fact that pair $(\grdsep,\kv)$ is considered in the maxima on the right-hand side of (\ref{eq:corr}).

Now we prove that $L\geq R$. For this, it suffices to show that 
\begin{equation}\label{eq:arbitrary}
\Val[\Ii]\geq \Pri(\grds(\grdsep))+\sum_{C\in \cc(\grdsep)}\Val[\Ii_C]
\end{equation}
for every guarded separator $\grdsep$ and every vector $\kv=(k_C)_{C\in \cc(\grdsep)}$ such that $\sum_{C\in \cc(\grdsep)}k_C=k-|\grds(\grdsep)|$; here, we denote $\Ii_C=(G,\Obj(C),\Cli(C),k_C)$. Hence, let us fix such pair $(\grdsep,\kv)$. For each $C\in \cc(\grdsep)$, let $\SolObj_C\subseteq \Obj(C)$ be a normal subfamily of cardinality $k_C$ that maximizes the revenue in the instance $\Ii_C$; i.e., $\Val[\Ii_C,\SolObj_C]=\Val[\Ii_C]$. If such a family does not exist, then $\Val[\Ii_C]=\minf$ and the claimed inequality is trivial.

Let us now construct $\SolObj:=\grds(\grdsep)\cup \bigcup_{C\in \cc(\grdsep)} \SolObj_C$. Since all the sets in this union are pairwise disjoint, we have that $|\SolObj|=k$. 

We first verify that $\SolObj$ is normal. For the sake of contradiction, suppose that there are two objects $p,p'\in \SolObj$ that contradict the definition of normality. If $p,p'\in \grds(\grdsep)$, then $\grds(\grdsep)$ would not be normal, a contradiction with the definition of a guarded separator. If $p\in \grds(\grdsep)$ and $p'\notin \grds(\grdsep)$, then it would follow that $p'\in \banned(\grdsep)$, which means that $p'$ could not have been included into any $\SolObj_C\subseteq \Obj(C)$, since these sets are disjoint with $\banned(\grdsep)$. If $p,p'\notin \grds(\grdsep)$ but $p,p'\in \SolObj_C$ for the same component $C$, then we would have a contradiction with the normality of family $\SolObj_C$. Finally, if $p\in \SolObj_C$ and $p'\in \SolObj_{C'}$ for different components $C,C'\in \cc(\grdsep)$, then we would have that $pp'\in E(\IntGraph(\grdsep))$, a contradiction with $C$ and $C'$ being different connected components of $\IntGraph(\grdsep)$.

Now, we verify that 
\begin{equation}\label{eq:solobj-val}
\Val[\Ii,\SolObj]=\Pri(\grds(\grdsep))+\sum_{C\in \cc(\grdsep)}\Val[\Ii_C,\SolObj_C].
\end{equation}
Firstly, for the objects' costs, since $\SolObj=\grds(\grdsep)\cup \bigcup_{C\in \cc(\grdsep)} \SolObj_C$ we have that the cost of every object of $\SolObj$ is counted exactly once on each side of (\ref{eq:solobj-val}). Second, we check that the prize of every client $q\in \Cli$ covered by $\SolObj$ is counted exactly once in the right-hand side of (\ref{eq:solobj-val}); note that this formula does not count the prize of any client not covered by $\SolObj$. If $q$ is covered by $\grds(\grdsep)$, then $\pri(q)$ is counted once in $\Pri(\grds(\grdsep))$ and in none of the terms $\Val[\Ii_C,\SolObj_C]$, since $q$ was removed when constructing the graph $\IntGraph(\grdsep)$. If now $q$ is not covered by $\grds(\grdsep)$, then $q$ belongs to some component $C\in \cc(\grdsep)$. Since $q$ is covered by $\SolObj$, there exists some $p\in \SolObj$ that covers $q$; observe that by the definition of $\IntGraph(\grdsep)$, for each such $p$ we have that $pq\in E(\IntGraph(\grdsep))$, which implies that each such $p$ must also belong to $C$. We infer that $\pri(q)$ is counted once in term $\Val[\Ii_C,\SolObj_C]$, and is counted zero times in every term $\Val[\Ii_{C'},\SolObj_{C'}]$ for $C'\neq C$. This concludes the proof of formula (\ref{eq:solobj-val}).

Now, since $\Val[\Ii]\geq \Val[\Ii,\SolObj]$ and $\Val[\Ii_C]=\Val[\Ii_C,\SolObj_C]$ for every $C\in \cc(\grdsep)$, by (\ref{eq:solobj-val}) we obtain that
$$\Val[\Ii]\geq \Val[\Ii,\SolObj]=\Pri(\grds(\grdsep))+\sum_{C\in \cc(\grdsep)}\Val[\Ii_C,\SolObj_C]=\Pri(\grds(\grdsep))+\sum_{C\in \cc(\grdsep)}\Val[\Ii_C],$$
which proves (\ref{eq:arbitrary}). Since $\grdsep$ and $\kv$ were chosen arbitrarily, we have thus proved that $L\geq R$.
\end{proof}

Lemma~\ref{lem:main-correctness} justifies the correctness of the following recursive algorithm for computing the value of $\Val[G,\Obj,\Cli,k]$, summarized as Algorithm $\AlgName$. As the border case, if $k\leq 3$, then we compute the optimum revenue in a brute-force manner, by iterating through all the $k$-tuples of the objects. Otherwise, we aim at computing $\Val[\Ii]$ using formula (\ref{eq:corr}). To this end, we run the algorithm of Theorem~\ref{thm:enumeration} to enumerate the family $\Vorsepfam$ of candidates for guarded separators separating evenly the solution. For each enumerated guarded separator $\grdsep\in \Vorsepfam$ we need to compute the value $\max_{\kv\in \comp(\grdsep)}\ \Pri(\grds(\grdsep))+\sum_{C\in \cc(\grdsep)}\Val[G,\Obj(C),\Cli(C),k_C]$, since Lemma~\ref{lem:main-correctness} asserts that then computing $\Val[G,\Obj,\Cli,k]$ will boil down to taking the maximum of these values. For a fixed guarded separator $\grdsep$, we investigate all the connected components of $\cc(\grdsep)$, and for each $C\in \cc(\grdsep)$ and every integer $\ell$ with $0\leq \ell \leq \frac{2}{3}k$ we compute the value $\Val[G,\Obj(C),\Cli(C),\ell]$ using a recursive call to Algorithm $\AlgName$. Computing value $\max_{\kv\in \comp(\grdsep)}\ \Pri(\grds(\grdsep))+\sum_{C\in \cc(\grdsep)}\Val[G,\Obj(C),\Cli(C),k_C]$ based on all the relevant values $\Val[G,\Obj(C),\Cli(C),\ell]$ amounts to running a standard knapsack dynamic programming algorithm that keeps track of the optimum revenue for selecting $\ell$ objects from the first $i$ connected components of $\cc(\grdsep)$, for all $\ell\in \{0,1,\ldots,k\}$ and $i=1,2,\ldots,|\cc(\grdsep)|$.

\newcommand{\Myto}{\ \mathbf{to}\ }

\begin{algorithm}[h!]%[H]
  
  \KwIn{An instance $\Ii=(G,\Obj,\Cli,k)$} 
  \KwOut{Value $\Val[\Ii]$}  \Indp \BlankLine
  \If{$k\leq 3$}{Compute the optimum revenue $ret$ by iterating through all the $k$-tuples of objects\\ \KwRet{$ret$}}
  $ret\leftarrow \minf$\\
  \For{$\grdsep\in \Vorsepfam$, enumerated using the algorithm of Theorem~\ref{thm:enumeration}}{
	Compute $\Pri(\grds(\grdsep))$\\
	Compute sets $\banned(\grdsep)$ and $\covered(\grdsep)$\\
	Compute $\IntGraph(\grdsep)=\IntGraph-(\grds(\grdsep)\cup \banned(\grdsep)\cup \covered(\grdsep))$\\
	$C_1,C_2,\ldots,C_p\leftarrow$ Connected components of $\IntGraph(\grdsep)$\\
	\For{$i=1\Myto p$}{
		\For{$\ell=0\Myto\lfloor\frac{2}{3}k\rfloor$}{
			$A[i][\ell]\leftarrow \AlgName(G,\Obj(C_i),\Cli(C_i),\ell)$
		}	
	}
	$D[0][0]\leftarrow 0$\\
	\For{$\ell=1\Myto k$}{
		$D[0][\ell]\leftarrow \minf$
	}
	\For{$i=1\Myto p$}{
		\For{$\ell=0\Myto k$}{
			$D[i][\ell]\leftarrow \minf$\\
			\For{$\ell'=0\Myto \min(\ell,\lfloor\frac{2}{3}k\rfloor)$}{
				$D[i][\ell]\leftarrow \max(D[i][\ell],A[i][\ell']+D[i-1][\ell-\ell'])$
			}
		}	
	}
	$ret\leftarrow \max(ret,\Pri(\grds(\grdsep))+D[p][k-|\grds(\grdsep)|])$
  }
  \KwRet{$ret$}   
\caption{Algorithm $\AlgName$}
  \label{alg:main}
\end{algorithm}

Lemma~\ref{lem:main-correctness} ensures that the algorithm correctly computes the value $\Val[\Ii]$, so we are left with estimating the running time.

\begin{lemma}\label{lem:running-time}
Algorithm $\AlgName$, applied to an instance $(G,\Obj,\Cli,k)$ with $|V(G)|=n$, $|\Obj|=\onum$, and $|\Cli|=\cnum$, runs in time $\onum^{\Oh(\sqrt{k})}\cdot (\cnum n)^{\Oh(1)}$ and space $(\onum\cnum n)^{\Oh(1)}$.
\end{lemma}
\begin{proof}
We shall say that a subcall $(G,\Obj',\Cli',k')$ to Algorithm $\AlgName$ is {\em{trivial}} if $\Obj'=\emptyset$. Note that trivial subcalls are resolved by $\AlgName$ in polynomial time since family $\Vorsepfam$ is empty.

Consider the application of $\AlgName$ to an instance $(G,\Obj,\Cli,k)$ with $k\geq 4$, and let us count the number of subcalls to $\AlgName$ it invokes. We have that $|\Vorsepfam|\leq (2\onum)^{15\sqrt{k}}$, and for every $\grdsep\in \Vorsepfam$ we call $\AlgName$ on $|\cc(\grdsep)|\cdot (\lfloor \frac{2}{3}k\rfloor+1)$ subinstances. At most $\onum$ components of $\cc(\grdsep)$ contain an object of $\Obj$, which means that at most $\onum\cdot (\lfloor \frac{2}{3}k\rfloor+1)\leq \onum k$ of these calls are non-trivial. In total, we have at most $\onum k\cdot (2\onum)^{15\sqrt{k}}\leq (2\onum)^{17\sqrt{k}}$ nontrivial subcalls.

Let $T^*$ be the tree of nontrivial recursive subcalls of the application of $\AlgName$ to instance $(G,\Obj,\Cli,k)$. Observe that $T^*$ has depth $\Oh(\log k)$, since for each call of $\AlgName$, in all the recursive subcalls the target size of the sought family is at most a $\frac{2}{3}$-fraction of the original one. For $0\leq \ell\leq k$, let $\phi(\ell)$ be the maximum number of leaves in the subtrees of $T^*$ rooted at calls where the target size of the sought family is at most $\ell$ (or $1$, if there are no such calls). Since in the recursive calls both the target family size and the total number of objects can only decrease, we have that $\phi(\ell)$ satisfies that:
\begin{equation}\label{eq:leaves-rec}
\phi(\ell)\leq (2\onum)^{17\sqrt{\ell}}\cdot \phi(\lceil 2\ell/3\rceil)\qquad\textrm{for $4\leq \ell\leq k$.}
\end{equation}
Moreover, since for $\ell\leq 3$ Algorithm $\AlgName$ solves the instance by brute-force, we have that
\begin{equation}\label{eq:leaves-base}
\phi(\ell)=1\qquad\textrm{for $0\leq \ell\leq 3$.}
\end{equation}

\begin{claim}\label{cl:leaves-induction}
For every function $\phi:\mathbb{N}\to \mathbb{N}$ that satisfies (\ref{eq:leaves-rec}) and (\ref{eq:leaves-base}), the following holds:
\begin{equation}
\phi(\ell)\leq (2\onum)^{17(3+\sqrt{6})\sqrt{\ell}}\qquad\textrm{for $0\leq \ell\leq k$.}
\end{equation}
\end{claim}
\begin{proof}
We prove the claim by induction. The base for $\ell\leq 3$ is trivial by (\ref{eq:leaves-base}), so let us show the inductive step for $\ell\geq 4$:
\begin{eqnarray*}
\phi(\ell)& \leq & (2\onum)^{17\sqrt{\ell}}\cdot \phi(\lceil 2\ell/3\rceil)\\
& \leq & (2\onum)^{17\sqrt{\ell}}\cdot (2\onum)^{17(3+\sqrt{6})\sqrt{\frac{2}{3}}\sqrt{\ell}}\\
& = & (2\onum)^{17(1+\sqrt{6}+2)\sqrt{\ell}}=(2\onum)^{17(3+\sqrt{6})\sqrt{\ell}};
\end{eqnarray*}
Here, in the first inequality we used (\ref{eq:leaves-rec}) and in the second we used the induction hypothesis.
\cqed\end{proof}

Claim~\ref{cl:leaves-induction} implies that the subcall tree $T^*$ has at most $(2\onum)^{17(3+\sqrt{6})\sqrt{k}}<(2\onum)^{93\sqrt{k}}$ leaves. Since the depth of $T^*$ is $\Oh(\log k)$, we infer that $T^*$ has $\Oh(\log k\cdot (2\onum)^{93\sqrt{k}})$ nodes. Now observe that for each call of $\AlgName$, the amount of work used in the main body of this algorithm (excluding the recursive subcalls) is equal to the maximum possible number of guarded separators considered at this step, i.e., $(2\onum)^{15\sqrt{k}}$, times a factor polynomial in $n$, $\onum$, and $\cnum$; here, we assume that the time used for resolving trivial subcalls is charged to the parent call. Hence, the total running time used by the algorithm is at most $(2\onum)^{108\sqrt{k}}\cdot (\onum\cnum n)^{\Oh(1)}\leq \onum^{\Oh(\sqrt{k})}\cdot (\cnum n)^{\Oh(1)}$, as requested. To argue that the algorithm runs in polynomial space, observe that at each moment the algorithm keeps track of the data stored for $\Oh(\log k)$ recursive calls of $\AlgName$, and, due to the algorithm of Theorem~\ref{thm:enumeration} using only polynomial working space, each call uses only polynomial space for its internal data.
\end{proof}

Thus, Lemmas~\ref{lem:main-correctness} and~\ref{lem:running-time} conclude the proof of Theorem~\ref{thm:main1}.

\section{Hardness results}

\newcommand{\parf}{\sigma}
\newcommand{\infvect}{\mathbf{u}}
\newcommand{\vectne}{\nearrow}
\newcommand{\vectnw}{\nwarrow}
\newcommand{\vectse}{\searrow}
\newcommand{\vectsw}{\swarrow}

\newcommand{\xR}{Q}% \overline{R}}
\newcommand{\parity}{\oplus}
\newcommand{\dirvect}[1]{\oplus[#1]}
\newcommand{\dirvectx}[1]{\oplus^{\circlearrowright}[#1]}

In this section, we prove the lower bounds in Theorems
\ref{thm:intro_planarcover}--\ref{thm:intro_recthard2} suggesting that
for many natural covering problems the $n^{O(k)}$ time brute force
algorithms are almost optimal. Let us start with a very simple
reduction showing such a lower bound in the setting of planar graphs. P\u{a}tra\c{s}cu and Williams gave a very tight lower bound for \probDS, assuming the Strong Exponential Time Hypothesis (SETH).
\begin{theorem}[\cite{DBLP:conf/soda/PatrascuW10}]\label{thm:domset-SETH}
Assuming SETH, there is no $f(k)\cdot n^{k-\epsilon}$ time algorithm for \probDS, where $n$ is the number of vertices of the input graph $G$.
\end{theorem}
Using this result, the proof of Theorem~\ref{thm:intro_planarcover} is very simple and transparent.

\restateintroplanarcover*

\begin{proof}% [Proof (of Theorem~\ref{thm:intro_planarcover})]
  Let $H$ be an arbitrary graph with $n$ vertices $v_1$, $\dots$,
  $v_n$. Let $G$ be a star with leaves $v_1$, $\dots$, $v_n$ and
  center $x$. For every vertex $v_i\in V(H)$, let us introduce a set
  $S_i$ into $\Obj$ that contains $v_i$, every neighbor of $v_i$ in
  $H$, and the center vertex $x$ of $G$; clearly, the set $S_i$ is
  connected in $G$. It is easy to see that $k$ vertices of $H$ form a
  dominating set of $G$ if and only if the corresponding $k$ sets of
  $\Obj$ cover every vertex of $G$. Therefore, an $f(k)\cdot
  (|\Obj|+|V(G)|)^{k-\epsilon}$ time algorithm for the covering
  problem on $G$ would give an $f(k)n^{k-\epsilon}$ time algorithm for
  solving \probDS on $H$; this contradicts SETH by
  Theorem~\ref{thm:domset-SETH}.
\end{proof}

In Sections~\ref{sec:convex-polygons}--\ref{sec:almost-squares}, we
give similar lower bounds for geometric objects, proving
Theorems~\ref{thm:intro_convexhard}--\ref{thm:intro_recthard2}.

\subsection{Convex polygons}
\label{sec:convex-polygons}

The proof of Theorem~\ref{thm:intro_convexhard} is again a simple
reduction from \probDS. We observe that if a family of points lies on
a convex curve, then any subset of them is in convex position, hence
convex polygons can cover arbitrary subsets of these points. This simple idea was also used, for example, by Har-Peled~\cite{DBLP:journals/corr/abs-0908-2369}, but here we use it to obtain a different conclusion (tight lower bound on the exponent instead of APX-hardness).

\restateintroconvexhard*

\begin{proof}% [Proof (of Theorem~\ref{thm:intro_convexhard})]
  Given an instance $(G,k)$ of \probDS with $n=|V(G)|$, we construct a
  set $\Obj$ of $n$ convex polygons and a set $\pointfam$ of $n+2$
  points such that $G$ has a dominating set of size $k$ if and only if
  $\pointfam$ can be fully covered by $k$ convex polygons from $\Obj$.
  Therefore, an $f(k)\cdot n^{k-\epsilon}$ time algorithm for the
  problem would imply an algorithm for \probDS with similar running
  time, which would contradict SETH by Theorem~\ref{thm:domset-SETH}.
\begin{figure}
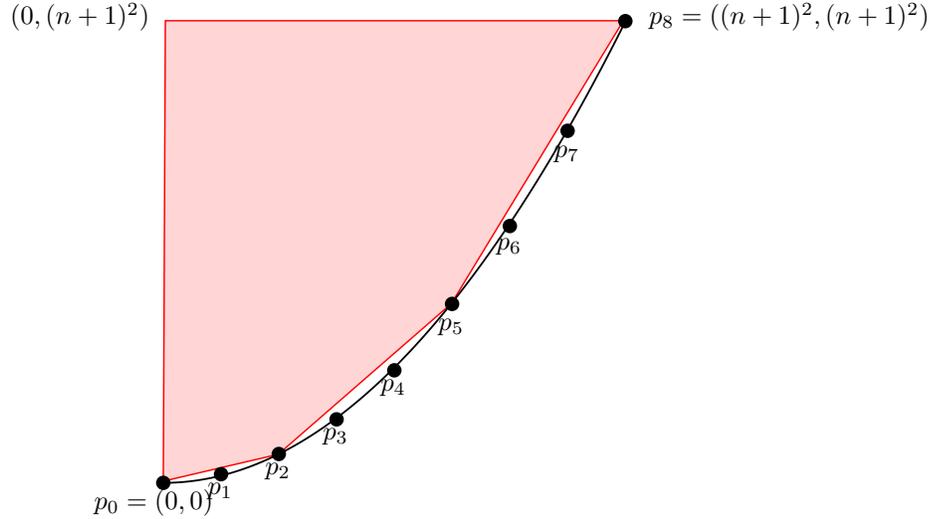

\begin{center}
{\small \svg{0.6\linewidth}{convex}}
\end{center}
\caption{Proof of Theorem~\ref{thm:intro_convexhard}. The points on a parabola and the convex polygon corresponding to a vertex whose closed neighborhood is $\{v_2,v_5\}$.}\label{fig:parabola}
\end{figure}

  We define the point set $\pointfam=\{ p_i=(i(n+1),i^2) \mid 0 \le i
  \le n+1\}$, which is a set of $n+1$ points on a parabolic curve (see
  Figure~\ref{fig:parabola}). Let $v_1$, $\dots$, $v_n$ be the
  vertices of $G$. For every $1\le i \le n$, we introduce into $\Obj$
  a convex polygon corresponding to $v_i$. Suppose that the closed
  neighborhood of $v_i$ is $v_{j_1}$, $\dots$, $v_{j_d}$ for some
  $1\le j_1 < j_2 < \dots <j_d \le n$. Then $v_i$ is represented by a
  convex polygon $P_i$ with vertices $(0,(n+1)^2)$, $p_0$, $p_{j_1}$,
  $p_{j_2}$, $\dots$, $p_{j_r}$, $p_{n+1}$. As the parabola is a
  convex curve, this polygon $P_i$ is convex and it covers exactly the
  points $p_0$, $p_{j_1}$, $p_{j_2}$, $\dots$, $p_{j_r}$, $p_{n+1}$.
  That is, if $v_j$ is not in the closed neighborhood of $v_i$, then $p_j$ is
  outside $P_j$. It is easy to see that a subset of $\Obj$ covers
  $\pointfam$ if and only if the union of the closed neighborhoods of
  the corresponding vertices of $G$ cover $V(G)$, or in other words,
  the corresponding vertices of $G$ form a dominating set.
\end{proof}

\subsection{Thin rectangles}
\label{sec:thin-rectangles}
In the proof of Theorem~\ref{thm:intro_recthard1}, we are
reducing from the \partbic problem: given a graph $G$ with a partition
$(A_1,\dots,A_k,B_1,\dots,B_k)$ of the vertices, the task is to find
$2k$ vertices $a_1\in A_1$, $\dots$, $a_k\in A_k$, $b_1\in B_1$,
$\dots$, $b_k\in B_k$ such that $a_i$ and $b_j$ are adjacent for every
$1\le i , j \le k$. There is a very simple reduction from
\probClique to \partbic (see, e.g., \cite{DBLP:conf/soda/DellM12}) where the output parameter equals the input one,
hence the lower bound of Chen et
al.~\cite{DBLP:journals/jcss/ChenHKX06} can be transferred to \partbic.

\begin{theorem}\label{th:mcb}
Assuming ETH, \partbic cannot be solved in time $f(k)n^{o(k)}$ for any computable function $f$.
\end{theorem}

The proof of Theorem~\ref{thm:intro_recthard1} is a parameterized reduction from \partbic.

\restateintrorectharda*
\begin{proof}
  We prove the theorem by a reduction from \partbic. Let $G$ be a graph
  with a partition $(A_1,\dots,A_k,B_1,\dots,B_k)$ of the vertices.
Without loss of generality, we assume that every class of the partition has the same size $n$. For $1\le i \le k$, we denote the vertices of $A_i$ and $B_i$ by $a_{i,j}$ and $b_{i,j}$ ($1\le j \le n$), respectively.

\myparagraph{Construction.}  We construct a set $\rectfam$ of axis-parallel
rectangles and a set $\pointfam$ of points such that there are $k':=8k$
rectangles in $\rectfam$ covering $\pointfam$ if and only if $G$ contains the required partitioned biclique. For simplicity of notation, we construct
open rectangles, that is, selecting a rectangle does not cover the
points on its boundary. However, it is easy to modify the reduction
(by decreasing the size of each rectangle slightly) so that it works
for closed rectangles as well.

First, for every $1\le i \le k$, we define the rectangles (see Figure~\ref{fig:thinrectangles})
\[
\begin{array}{l}
V^T_{i}=(4i-0.5,4i+0.5)\times (4(k+1),8(k+1))\\
V^B_{i}=(4i-0.5,4i+0.5)\times (-4(k+1),0)\\
V^L_{i}=(4i-1,4i)\times (0,4(k+1))\\
 V^R_{i}=(4i,4i+1)\times (0,4(k+1))\\
\end{array}
\]
and for every $1\le j \le k$, we define the rectangles
\[
\begin{array}{l}
H^T_{j}=(0,4(k+1))\times (4j,4j+1)\\
H^B_{j}=(0,4(k+1))\times (4j-1,4j)\\
H^L_{j}=(-4(k+1),0)\times (4j-0.5,4j+0.5)\\
H^R_{j}=(4(k+1),8(k+1))\times (4j-0.5,4j+0.5)\\
\end{array}
\]

\begin{figure}
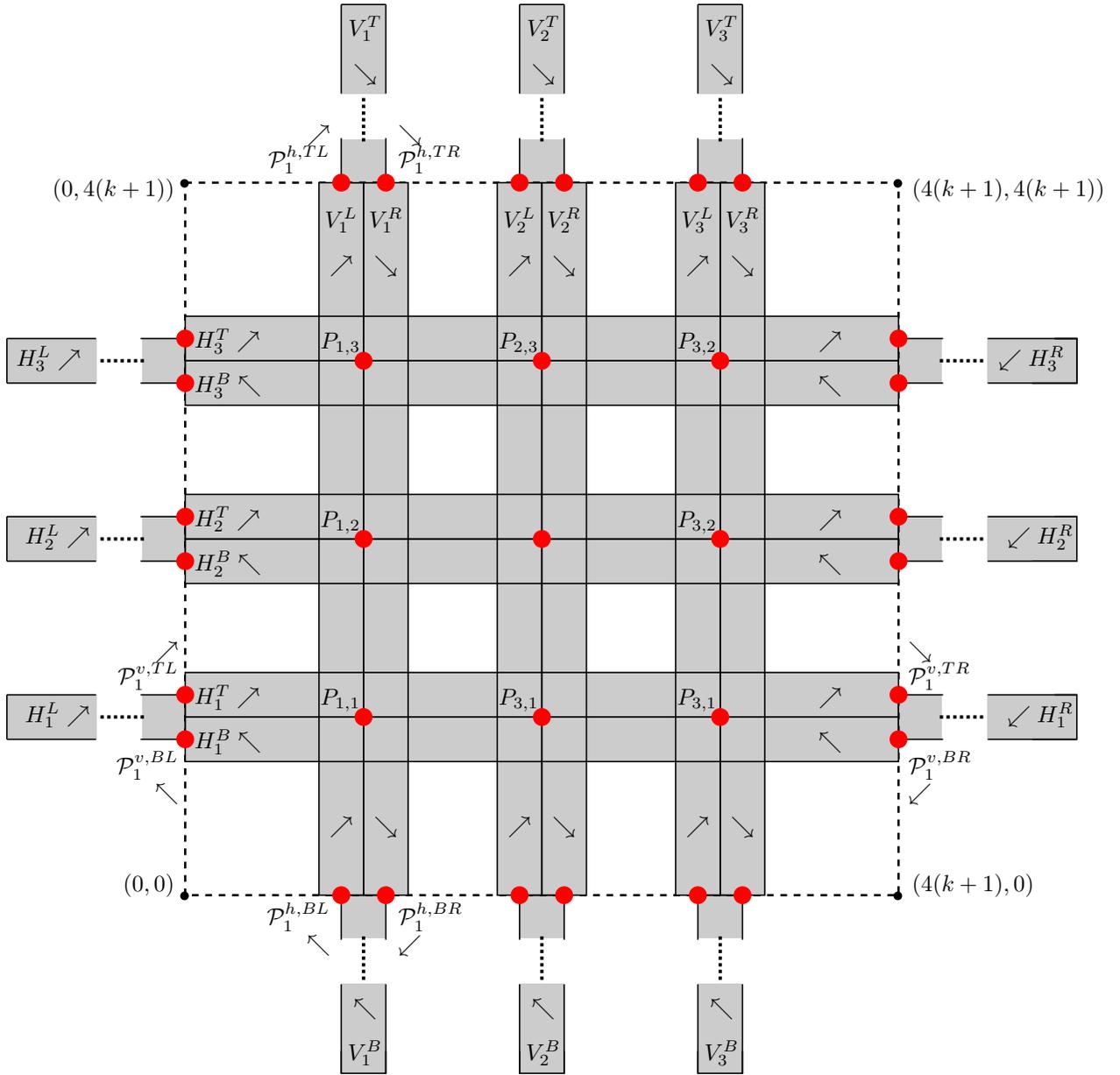

\begin{center}
{\small \svg{\linewidth}{thinrectangles}}
\end{center}
\caption{Proof of Theorem~\ref{thm:intro_recthard1}. The rectangles on
  the left and the right are shortened to save space. The directions
  of the arrows on the rectangles show the offsets of the rectangles
  $V^T_{i,\alpha}$ etc. as $\alpha$ and $\beta$ grow. Similarly, the
  arrows at $\pointfam^{h,TL}_i$ etc. show the offset of the vertices
  $h^{TL}_i$.  }\label{fig:thinrectangles}
\end{figure}

These rectangles themselves do not appear in the set $\rectfam$, but they will used be for the definition of rectangles in $\rectfam$, as follows.
For a rectangle $R$, we use the notation $R+(x,y)$ to denote the rectangle obtained from $R$ by shifting it horizontally by $x$ and vertically by $y$.
We use the notation $\nearrow$, $\searrow$ etc.~for the diagonal vectors $(1,1)$, $(1,-1)$ etc.~with unit projections; then, e.g., $\lambda\cdot \searrow$ denotes the vector $(\lambda,-\lambda)$.

Let $\epsilon=1/(100n)$.
For every $1\le  i\le k$ and $1\le \alpha \le n$, we introduce the following rectangles into $\rectfam$ (see the directions shown by the arrows in Figure~\ref{fig:thinrectangles}): 
\[
\begin{array}{l}
V^T_{i,\alpha}=V^T_{i}+\alpha\epsilon\cdot \vectse\\
V^B_{i,\alpha}=V^B_{i}+\alpha\epsilon\cdot \vectnw\\
V^L_{i,\alpha}=V^L_{i}+\alpha\epsilon\cdot \vectne\\
V^R_{i,\alpha}=V^R_{i}+\alpha\epsilon\cdot \vectse\\
\end{array}
\]
Furthermore, for every $1\le j \le k$ and $1\le \beta \le n$, we introduce the following rectangles:
\[
\begin{array}{l}
H^T_{j,\beta}=H^T_{j}+\beta\epsilon\cdot \vectne\\
H^B_{j,\beta}=H^B_{j}+\beta\epsilon\cdot \vectnw\\
H^L_{j,\beta}=H^L_{j}+\beta\epsilon\cdot \vectne\\
H^R_{j,\beta}=H^R_{j}+\beta\epsilon\cdot \vectsw\\
\end{array}
\]
This completes the description of $\rectfam$; note that $|\rectfam|=8kn$. To define $\pointfam$, for every $1\le i \le k$ and $1\le \alpha \le n$, we introduce the points (see again the directions shown by arrows in Figure~\ref{fig:thinrectangles})
\[
\begin{array}{ll}
v^{TL}_{i,\alpha}=(4i-0.5,4(k+1))+(\alpha+0.5)\epsilon\cdot \vectne & v^{TR}_{i,\alpha}=(4i+0.5,4(k+1))+(\alpha+0.5)\epsilon\cdot \vectse \\
v^{BL}_{i,\alpha}=(4i-0.5,0)+(\alpha+0.5)\epsilon\cdot \vectnw& v^{BR}_{i,\alpha}=(4i+0.5,4(k+1)) +(\alpha+0.5)\epsilon\cdot \vectsw\\
% v^{TL}_{i,\alpha}=(4i-0.5+\alpha\epsilon+\epsilon/2,4(k+1)+\alpha\epsilon+\epsilon/2) & v^{TR}_{i,\alpha}=(4i+0.5+\alpha\epsilon+\epsilon/2,4(k+1)-\alpha\epsilon-\epsilon/2) \\
% v^{BL}_{i,\alpha}=(4i-0.5-\alpha\epsilon-\epsilon/2,\alpha\epsilon+\epsilon/2) & v^{BR}_{i,\alpha}=(4i+0.5-\alpha\epsilon-\epsilon/2,4(k+1)-\alpha\epsilon-\epsilon/2) \\
\end{array}
\]
and for every $1\le j \le k$ and $1\le \beta \le n$, we introduce the points
\[
\begin{array}{ll}
h^{TL}_{j,\beta}=(0,4j+0.5)+(\beta+0.5)\epsilon\cdot \vectne & h^{TR}_{j,\beta}=(4(k+1),4j+0.5)+(\beta+0.5\epsilon)\cdot \vectse \\
h^{BL}_{j,\beta}=(0,4j-0.5)+(\beta+0.5)\epsilon\cdot \vectnw & h^{BR}_{j,\beta}=(4(k+1),4j-0.5)+(\beta+0.5)\epsilon\cdot \vectsw \\
\end{array}
\]
We will use the notation $\pointfam^{v,TL}_i:=\{v^{TL}_{i,\alpha}:1\le \alpha
\le n\}$ etc. for these sets of points. Notice that these points are
not on the boundary of any of the rectangles (because of the terms
$0.5\epsilon$), thus it does not matter for the covering of these
points that the rectangles are open.

Finally, for every $1\le i,j \le k$, we add a set $P_{i,j}$ of points
to $\pointfam$ the following way: for every $1\le \alpha,\beta\le n$,
the point $(4i+\alpha\epsilon, 4j +\beta\epsilon)$ is in $P_{i,j}$ if
$a_{i,\alpha}$ and $b_{j,\beta}$ are {\em not} adjacent in $G$. This
completes the description of construction, we have
\[
\pointfam=\bigcup_{i=1}^k (\pointfam^{v,TL}_{i}\cup \pointfam^{v,TR}_{i}\cup \pointfam^{v,BL}_{i}\cup \pointfam^{v,BR}_{i})
\cup \bigcup_{j=1}^k (\pointfam^{h,TL}_{j}\cup \pointfam^{h,TR}_{j}\cup \pointfam^{h,BL}_{h}\cup \pointfam^{h,BR}_{j})
\cup \bigcup_{1\le i,j\le k} P_{i,j}.
\]
Note that $|\pointfam|\leq 8kn+k^2n^2$.

 We claim that we can cover
the points in $\pointfam$ with $8k$ rectangles of $\rectfam$ if and
only if the \partbic instance has a solution. The intuition behind the construction is the following. A solution of size $8k$ has to select exactly one rectangle of each of the $8k$ types $V^T_{1,\alpha}$, $V^B_{1,\alpha}$ etc. The fact that the point set $\pointfam^{v,TL}_{i}\cup \pointfam^{v,TR}_{i}\cup \pointfam^{v,BL}_{i}\cup \pointfam^{v,BR}_{i}$ has to be covered by the four selected rectangles of the form $V^T_{i,\alpha_T}$, $V^B_{i,\alpha_B}$,
$V^L_{i,\alpha_L}$, $V^R_{i,\alpha_R}$ ensures that the selection of these rectangles are synchronized, that is, we have to select $V^T_{i,\alpha_i}$, $V^B_{i,\alpha_i}$,
$V^L_{i,\alpha_i}$, $V^R_{i,\alpha_i}$ for some $1\le \alpha_i \le n$. Similarly, for every $1\le j \le k$, we have to select  $H^T_{j,\beta_j}$, $H^B_{j,\beta_j}$,
$H^L_{j,\beta_j}$, $H^R_{j,\beta_j}$ for some $1\le \beta_j \le n$. The role of the point set $P_{i,j}$ is to ensure that vertices $a_{i,\alpha_i}$ and $b_{j,\beta_j}$ are adjacent in $G$. Indeed, if they are not adjacent, then the point $(4i+\alpha\epsilon, 4j +\beta\epsilon)$ is in $P_{i,j}$, but it is not covered by any rectangle: it is on the boundary of each of $V^L_{i,\alpha_i}$, $V^R_{i,\alpha_i}$,
$H^T_{j,\beta_j}$, $H^B_{j,\beta_j}$.

For the formal proof of the equivalence, let us start with a useful preliminary observation.
\begin{claim}\label{cl:rectvert}
For every $1\le j \le k$ and $1\le \alpha_T,\alpha_B,\alpha_L,\alpha_R\le n$, the rectangles $V^T_{i,\alpha_T}$, $V^B_{i,\alpha_B}$,
$V^L_{i,\alpha_L}$, $V^R_{i,\alpha_R}$ cover all the points in $\pointfam^{v,TL}_{i}\cup \pointfam^{v,TR}_{i}\cup \pointfam^{v,BL}_{i}\cup \pointfam^{v,BR}_{i}$ if and only if $\alpha_T=\alpha_B=\alpha_L=\alpha_R$.
\end{claim}
\begin{proof}
  Assume first that $\alpha_T=\alpha_B=\alpha_L=\alpha_R=\alpha^*$. Consider a point
  $v^{TL}_{i,\alpha}$. If $\alpha<\alpha^*$, then rectangle $V^L_{i,\alpha^*}$ covers this
  point: the vertical coordinate of this point is $4(k+1)+\alpha\epsilon+\epsilon/2$ and the top boundary of $V^L_{i,\alpha^*}$ is at $4(k+1)+\alpha^*\epsilon$. If $\alpha\ge \alpha^*$, then rectangle $V^{T}_{i,\alpha^*}$ covers the point:
the horizontal coordinate of this point is $4i-0.5+(\alpha+0.5)\epsilon$ and the left boundary of $V^T_{i,\alpha^*}$ is at $4i-0.5+\alpha^*\epsilon$.
 In a
  similar way, we can verify that the four rectangles cover
  $v^{TR}_{i,\alpha}$, $v^{BL}_{i,\alpha}$, $v^{BR}_{i,\alpha}$ for every $1\le \alpha \le
  n$.

  Suppose now that the four rectangles cover all the $4n$ points. If
  $\alpha_L<\alpha_T$, then point $v^{TL}_{i,\alpha_L}$ is not
  covered: rectangle $V^L_{i,\alpha_L}$ covers only the points
  $v^{TL}_{1,1}$, $\dots$, $v^{TL}_{1,\alpha_L-1}$ of $\pointfam^{v,TL}_{i}$,
  while rectangle $V^L_{i,\alpha_T}$ covers only the points
  $v^{TL}_{i,\alpha_T}$, $\dots$, $v^{TL}_{i,n}$ of
  $\pointfam^{v,TL}_{i}$. Therefore, we have $\alpha_L\ge \alpha_T$. In a
  similar way, we can prove $\alpha_T\ge \alpha_R\ge \alpha_B\ge
  \alpha_L\ge \alpha_T$, implying that all these values are equal.
  \cqed
\end{proof}
We can prove an analogous statement for the vertical rectangles.
\begin{claim}\label{cl:recthor}
For every $1\le j \le k$ and $1\le \beta_T,\beta_B,\beta_L,\beta_R\le n$, the rectangles $H^T_{j,\beta_T}$, $H^B_{j,\beta_B}$,
$H^L_{j,\beta_L}$, $H^R_{j,\beta_R}$ cover all the points in $\pointfam^{h,TL}_{j}\cup \pointfam^{h,TR}_{j}\cup \pointfam^{h,BL}_{j}\cup \pointfam^{h,BR}_{j}$ if and only if $\beta_T=\beta_B=\beta_L=\beta_R$.
\end{claim}

\myparagraph{Biclique $\Rightarrow$ covering rectangles.}
Suppose that vertices $a_{1,\alpha_{1}}$, $\dots$, $a_{k,\alpha_{k}}$, $b_{1,\beta_{1}}$, $\dots$, $b_{k,\beta_{k}}$ form a biclique in $G$. We claim that the $8k$ rectangles
\begin{itemize}
\item $V^T_{i,\alpha_{i}}$, $V^B_{i,\alpha_{i}}$, $V^L_{i,\alpha_{i}}$, $V^R_{i,\alpha_{i}}$ for $1\le i \le k$, and 
\item $H^T_{j,\beta_{j}}$, $H^B_{j,\beta_{j}}$, $H^L_{j,\beta_{j}}$, $H^R_{j,\beta_{j}}$ for $1\le j \le k$
\end{itemize}
cover every point in $\pointfam$. By
Claim~\ref{cl:rectvert}, for every $1\le i \le k$, the rectangles $V^T_{i,\alpha_{i}}$,
$V^B_{i,\alpha_{i}}$, $V^L_{i,\alpha_{i}}$, $V^R_{i,\alpha_{i}}$ cover every point in
$\pointfam^{v,TL}_{i}\cup \pointfam^{v,TR}_{i}\cup \pointfam^{v,BL}_{i}\cup \pointfam^{v,BR}_{i}$
 and by Claim~\ref{cl:recthor}, for every $1\le j \le k$, the rectangles
$H^T_{j,\beta_{j}}$, $H^B_{j,\beta_{j}}$, $H^L_{j,\beta_{j}}$,
$H^R_{i,\beta_{j}}$ cover every point in $\pointfam^{h,TL}_{j}\cup \pointfam^{h,TR}_{j}\cup \pointfam^{h,BL}_{h}\cup \pointfam^{h,BR}_{j}$. Consider now a point
$p=(4i+\alpha\epsilon,4j+\beta\epsilon)\in P_{i,j}$.  If $\alpha<\alpha_i$, then
$V^L_{i,\alpha_i}$ covers $p$; if $\alpha>\alpha_i$, then
$V^R_{i,\alpha_i}$ covers $p$.  If $\beta<\beta_j$, then
$H^B_{j,\beta_j}$ covers $p$; if $\beta>\beta_j$, then
$V^T_{j,\beta_j}$ covers $p$. Therefore, we have a problem only if
$\alpha=\alpha_i$ and $\beta=\beta_j$. However, we know that vertices
$a_{i,\alpha_i}$ and $b_{j,\beta_j}$ are adjacent in $G$, which
means that the point $p=(4i+\alpha\epsilon,4j+\beta\epsilon)$ was {\em
  not} added to $P_{i,j}$. Therefore, the selected $8k$ rectangles indeed
cover every point in $\pointfam$.

\myparagraph{Covering rectangles $\Rightarrow$ biclique.}
For the proof of the reverse direction, suppose that $\pointfam$ was
covered by $8k$ rectangles from $\rectfam$. For every $1\le i \le k$,
at least one rectangle of the from $V^T_{i,\alpha}$ has to be selected: no rectangle
of some other form can cover the point $v^{TL}_{i,n}$ with vertical
coordinate $4(k+1)+n\epsilon+\epsilon/2$. Similarly, at least one
rectangle of each of the forms $V^B_{i,\alpha}$, $V^L_{i,\alpha}$,
$V^R_{i,\alpha}$ has to be selected. An analogous statement holds for
every $1\le j \le k$ and rectangles of the form $H^T_{j,\beta}$,
$H^B_{j,\beta}$, $H^L_{j,\beta}$, $H^R_{j,\beta}$.  As exactly $8k$
rectangles were selected, exactly one rectangle has to be selected for
each type. Consider some $1\le i \le k$ and suppose that
$V^T_{i,\alpha_T}$, $V^B_{i,\alpha_B}$, $V^L_{i,\alpha_L}$, $V^R_{i,\alpha_R}$
were selected. As these rectangles have to cover 
$\pointfam^{v,TL}_{i}\cup \pointfam^{v,TR}_{i}\cup \pointfam^{v,BL}_{i}\cup \pointfam^{v,BR}_{i}$, Claim~\ref{cl:rectvert} implies that 
$\alpha_T=\alpha_B=\alpha_L=\alpha_R$; let us define $\alpha_i$ to be this number. Similarly, using Claim~\ref{cl:recthor}, we can show that for every $1\le j \le k$, there is a $1 \le \beta_j\le n$ such that the rectangles $H^T_{j,\beta_j}$, $H^B_{j,\beta_j}$, $H^L_{j,\beta_j}$, $H^R_{j,\beta_j}$ were selected.

We claim that $a_{1,\alpha_1}$, $\dots$, $a_{k,\alpha_k}$,
$b_{1,\beta_1}$, $\dots$, $b_{k,\beta_k}$ form a solution to
the \partbic problem. Suppose for a contradiction that
$a_{i,\alpha_i}$ and $b_{j,\alpha_j}$ are not adjacent in $G$. Then
the point $(4i+\alpha_i\epsilon,4j+\beta_j\epsilon)$ is in
$P_{i,j}$. We arrive to a contradiction by showing that this point is
not covered by any of the selected rectangles. Indeed, this point is
on the boundary of each of $V^L_{i,\alpha_i}$, $V^R_{i,\alpha_i}$,
$H^T_{j,\beta_j}$, $H^B_{j,\beta_j}$ and recall that all the
rectangles are open. Therefore, we have shown that the \partbic instance has a solution.
\end{proof}

\subsection{Almost squares}
\label{sec:almost-squares}The goal of this section is to prove Theorem~\ref{thm:intro_recthard2}, which is another lower bound for covering
points with axis-parallel rectangles. In
Theorem~\ref{thm:intro_recthard1}, we had only two types of ``thin''
rectangles. In this section, we consider sets of rectangles of many
different sizes, but they are all ``almost unit squares'': every
rectangle has width and height in the range $[1-\epsilon_0,1+\epsilon_0]$ for
some fixed $\epsilon_0>0$.  

Restricting the problem to almost unit squares makes the hardness
proof significantly more challenging.  Let us observe that the proof
of Theorem~\ref{thm:intro_recthard1} crucially relied on the fact that
the selection of the horizontal rectangles $H^T_{j,\beta}$,
$H^B_{j,\beta}$ interacted with the selection of the vertical
rectangles $V^L_{i,\alpha}$, $V^R_{i,\alpha}$ for every $1\le i,j\le
k$. Intuitively speaking, the selection of these $O(k)$ rectangles
were constrained by $O(k^2)$ interactions. The reason why such a large
number of interactions could be reached is because a wide horizontal
rectangle could be intersected by $O(k)$ independent thin vertical
rectangles. However, if every rectangle is almost a unit square, then
this is no longer possible. After trying some possible configurations
that may be useful for a hardness proof, one gets the impression that
an almost unit square can have independent interactions only with
$O(1)$ other almost unit squares. This means that we can have a total
of $O(k)$ interactions, which does not seem to be sufficient to
express a problem such as \probClique or \partbic, as these problems
have a richer interaction structure with $O(k^2)$
interactions. Therefore, it seems difficult to reduce from
these problems with only a linear blowup of the parameter, which is
necessary for the lower bounds that we want to prove.

To get around these limitations, we are reducing from a problem where
the interaction structure is sparser and a lower bound ruling out
$f(k)n^{o(k)}$ (or so) algorithms is known even with only $O(k)$
constraints.  The \partsub problem is similar to \partbic, but now the
graph we are looking for can be arbitrary, it is not necessarily a
biclique. The input consists of a graph $H$ with vertex set
$\{u_1,\dots,u_k\}$ and a graph $G$ whose vertex set is partitioned
into $k$ classes $V_1$, $\dots$, $V_k$. The task is to find a mapping
$\mu:V(H)\to V(G)$ such that $\mu(u_i)\in V_i$ for every $1\le i \le
k$ and $\mu$ is a subgraph embedding, that is, if $u_i$ and $u_j$ are
adjacent in $H$, then $\mu(u_i)$ and $\mu(u_j)$ are adjacent in
$G$. This problem contains \partbic as a special case, thus
Theorem~\ref{th:mcb} implies that it cannot be solved in time
$f(k)n^{o(k)}$ for any computable function $f$. However, if we
express \partbic with \partsub, then $H$ has $O(k^2)$
edges. Interestingly, the lower bound remains valid (up to a log
factor) even if we restrict $H$ to be sparse.
\begin{theorem}[\cite{marx-toc-treewidth}]\label{thm:subgraph}
  Assuming ETH, \partsub cannot be solved in time $f(k)n^{o(k/\log
    k)}$ (where $k=|V(H)|$) for any computable function $f$, even when
  $H$ is restricted to be a 3-regular bipartite graph.
\end{theorem}

\newcommand{\I}{\mathcal I}

We prove a lower bound first for an intermediary problem, covering
points by 2-track intervals, and then obtain
Theorem~\ref{thm:intro_recthard2} using a simple reduction of Chan and
Grant~\cite{DBLP:journals/comgeo/ChanG14}. Many natural optimization
problems on intervals (such as finding a maximum set of disjoint
intervals or covering points on the line by the minimum number of
intervals) are polynomial-time solvable. One can obtain
generalizations of these problems by considering {\em $c$-intervals}
instead of intervals, that is, the objects given in the input are
unions of $c$ intervals and we have to select $k$ such objects subject
to some disjointness or covering condition
\cite{DBLP:journals/talg/ButmanHLR10,DBLP:journals/tcs/FellowsHRV09,DBLP:conf/latin/FominGGSSLVV12,DBLP:journals/algorithmica/Francis0O15,DBLP:journals/algorithmica/Jiang13,DBLP:journals/tcs/JiangZ12}. These
problems are typically much harder than the analogous problems on
intervals. A special case of a $c$-interval is a {\em $c$-track
  interval:} we imagine the problem to be defined on $c$ independent
lines (tracks) and each $c$-track interval consists of an interval on
each of the tracks. Many of the hardness results for $c$-intervals
remain valid also for $c$-track intervals.

Let us formally define the intermediary problem and prove the lower
bound we need.  A {\em 2-track interval} is an ordered pair
$(I^1,I^2)$ of intervals on the real line; for concreteness the reader may assume that the intervals are closed, but this does not have any influence on the validity of the claims to follow.  We say that 2-track interval $(I^1,I^2)$ {\em covers a
  point $x$ on the first (resp., second) track} if $x\in I^1$ (resp.,
$x\in I^2$).  A {\em 2-track point set $\pointfam$} is an ordered pair
$(\pointfam^1,\pointfam^2)$ of two sets of real numbers.  We say that
that a set $\I$ of 2-track intervals covers a 2-track point set
$(\pointfam^1,\pointfam^2)$ if
\begin{itemize}
\item for every $x\in \pointfam^1$, there is a 2-track interval $(I^1,I^2)\in\I$ that covers $x$ on the first track (i.e., $x\in I^1$) and 
\item for every $x\in \pointfam^2$, there is a 2-track interval $(I^1,I^2)\in\I$ that covers $x$ on the second track (i.e., $x\in I^2$).
\end{itemize}
We prove a lower bound for the problem of covering a 2-track point set
by selecting $k$ 2-track intervals. Note that this problem is related
to, but not the same as, the \probDS problem for 2-track
interval graphs, for which W[1]-hardness results are known
\cite{DBLP:journals/tcs/FellowsHRV09,DBLP:journals/tcs/JiangZ12} (however no ETH-based lower bound was stated).

\begin{theorem}\label{th:2track}
  Consider the problem of covering a 2-track point set
  $\pointfam=(\pointfam^1,\pointfam^2)$ by selecting $k$ 2-track intervals from
  a set $\I$.  Assuming ETH, there is no algorithm for this problem with running time
  $f(k)\cdot (|\pointfam|+|\I|)^{o(k/\log k)}$ for any computable function $f$.
\end{theorem}
\begin{proof}
  We prove the theorem by a reduction from \partsub. Let $H$ and $G$
  be two graphs, let $V(H)=\{u_1,\dots, u_{k}\}$, and let $(V_1,\dots,
  V_{k})$ be a partition of $V(G)$. By copying vertices if necessary, we may assume that every $V_i$
  has the same size $n$; let us denote by $\{v_{i,1},\dots,v_{i,n}\}$
  the vertices in $V_i$. As $H$ is bipartite and 3-regular, the edges
  of $H$ are 3-colorable. Let us fix a 3-coloring of the edges of $H$
  with colors in $\{0,1,2\}$. If there is an edge $u_{i_a}u_{i_b}$ in
  $H$ having color $c\in\{0,1,2\}$, then we will also refer to the
  edges of $G$ between $V_{i_a}$ and $V_{i_b}$ as having color $c$. We
  may assume that every edge of $G$ receives a color this way, that
  is, there is no edge induced by any $V_{i}$ and there can be an edge between
  $V_{i_a}$ and $V_{i_b}$ only if $u_{i_a}u_{i_b}$ is an edge of
  $H$ (every other edge of $G$ is useless for the purposes of this problem). As $H$ is 3-regular, the two partite classes of $H$ have to be
  of the same size; we may assume without loss of generality that
  $\{u_1,\dots,u_{k/2}\}$ and $\{u_{k/2+1},\dots,u_k\}$ are these two
  classes.

\textbf{Construction.} We define the set $\I$ of 2-track intervals the following way (see Figure~\ref{fig:2track}). First, for every $1\le i \le k/2$ and $1\le j \le n$, let us introduce the following two 2-track intervals into $\I$:
\begin{itemize}
\item $I_{i,j,1}:=([4ni,4ni+j],[4ni,4ni+n-j])$
\item $I_{i,j,2}:=([4ni+3n+j,4ni+4n],[4ni+n-j,4ni+4n])$
\end{itemize}
For every $k/2+1 \le i\le k$ and $1\le j\le n$, we proceed similarly, but exchanging the role of the two tracks. That is, we introduce the following two 2-track intervals into $\I$:
\begin{itemize}
\item $I_{i,j,1}:=([4ni,4ni+n-j],[4ni,4ni+j])$
\item $I_{i,j,2}:=([4ni+n-j,4ni+4n],[4ni+3n+j,4ni+4n])$.
\end{itemize}
Next, we introduce a 2-track interval into $\I$ for every edge of $G$. Let $e=v_{i_a,j_a}v_{i_b,j_b}$ be an edge of $G$ with $1\le i_a \le k/2$ and $k/2+1\le i_b\le k$. Suppose that this edge $e$ has color $c\in\{0,1,2\}$ in the edge coloring we fixed.  Then we introduce
\begin{itemize}
\item $I^*_{e}:=([4ni_a+cn+j_a,4ni_a+(c+1)n+j_a],
[4ni_b+cn+j_b,4ni_b+(c+1)n+j_b])$
\end{itemize}
into the set $\I$. This completes the construction of $\I$; note that $|\I|=3.5nk$ (as $H$ is 3-regular).
\begin{figure}
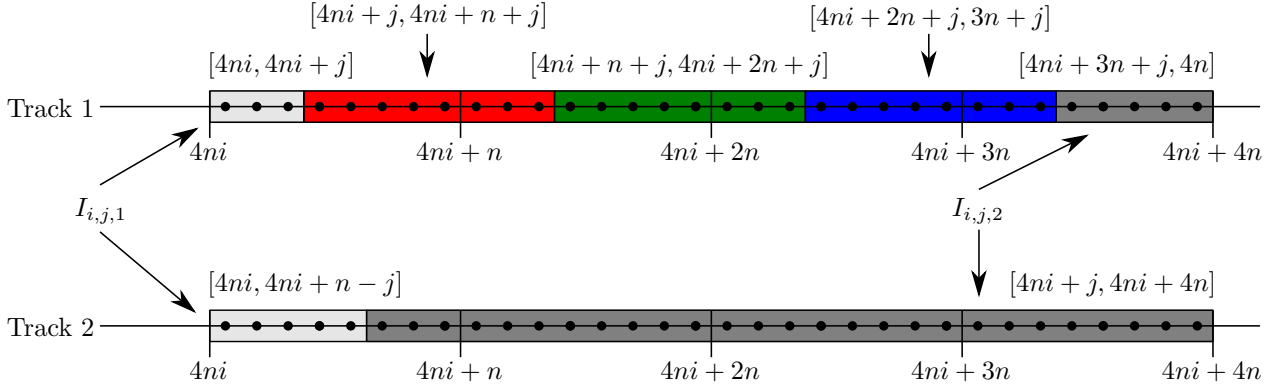

\begin{center}
{\small \svg{\linewidth}{2track2}}
\end{center}
\caption{Proof of Theorem~\ref{th:2track}. Let $n=8$ and $j=3$. The figure shows the 2-track intervals $I_{i,j,1}$ (light gray) and $I_{i,j,2}$ (dark gray). Notice that together they cover $[4ni,4ni+4n]$ on the second track. On the first track, the three intervals of length $n$ (shown by red, green, and blue) are needed to complete the covering of the range $[4ni,4ni+4n]$. These intervals can be provided by the 2-track intervals $I^*_{e^0}$, $I^*_{e^1}$, $I^*_{e^2}$, where $e^c$ is an edge of $G$ incident to $v_{i,j}$ and having color $c$.}\label{fig:2track}
\end{figure}

Let $\pointfam=\{4n+0.5,4n+1.5,\dots, 4n(k+1)-0.5\}$. Let $k'=2|V(H)|+|E(H)|=3.5k$. We claim that it is possible to select $k'$ members of $\I$ to cover the 2-track point set $(\pointfam,\pointfam)$ if and only if the \partsub instance has a solution.

\textbf{Subgraph embedding $\Rightarrow$ 2-track intervals.} Suppose first that vertices $v_{1,s_1}$, $\dots$, $v_{k,s_k}$ form a
solution of the \partsub instance. Then we can cover $(\pointfam,\pointfam)$ by selecting the following set $\I'$ of 2-track intervals:
\begin{itemize}
\item $I_{i,s_i,1}$ and $I_{i,s_i,2}$ for $1\le i \le k$, and
\item for every edge $u_{i_a}u_{i_b}$ of $H$ (with $1\le i_a\le k/2$
  and $k/2+1\le i_b \le k$), edge $e=v_{i_a,s_{i_a}}v_{i_b,s_{i_b}}$
  has to be an edge of $G$; we select the corresponding 2-track
  interval $I^*_e$.
\end{itemize}
We claim that every point in $[4n,4n(k+1)]$ is covered by $\I'$ on both
tracks; this clearly implies that $(\pointfam,\pointfam)$ is covered. Consider the range $[4ni,4ni+4n]$ for some $1\le i \le k/2$ (the case $k/2+1\le i \le k$ is similar). On the second track, $I_{i,s_i,1}$ covers $[4ni,4ni+n-s_i]$ and $I_{i,s_i,2}$ covers $[4ni+n-s_i,4ni+4n]$, thus the range is indeed covered (see Figure~\ref{fig:2track}). Let $k/2\le i^0,i^1,i^2\le k$ be the three neighbors of $u_i$ in $H$, with the edge $u_iu_{i^c}$ having color $c$. This means that the edge $e^c=v_{i,s_i}v_{i^c,s_{i^c}}$ exists in $G$ and $I^*_{e^c}$ was selected into $\I'$ for every $c\in \{0,1,2\}$.
Recall from the definition that $I^*_{e^c}$ covers $[4ni+cn+s_i,4ni+(c+1)n+s_i]$ on the first track.
 Now we have that, on the first track,
\begin{itemize}
\item $[4ni,4ni+s_i]$ is covered by $I_{i,j,1}$,
\item $[4ni+s_i,4ni+n+s_i]$ is covered by $I^*_{e^0}$,
\item $[4ni+n+s_i,4ni+2n+s_i]$ is covered by $I^*_{e^1}$,
\item $[4ni+2n+s_i,4ni+3n+s_i]$ is covered by $I^*_{e^2}$, and
\item $[4ni+3n+s_i,4ni+4n]$ is covered by $I_{i,j,2}$.
\end{itemize}
Thus $[4ni,4ni+4n]$ is indeed covered and we have shown that $(\pointfam,\pointfam)$ is covered by $\I'$.

\textbf{2-track intervals $\Rightarrow$ subgraph embedding.} Suppose now that there is subset $\I'\subseteq \I$ of size $k'$ that
covers the 2-track point set $(\pointfam,\pointfam)$. Let us define the following subsets of $\I$:
\begin{itemize}
\item $\I_{i,1}:=\{I_{i,j,1}\mid 1\le j \le n\}$,
\item $\I_{i,2}:=\{I_{i,j,2}\mid 1\le j \le n\}$,
\item $\I^*_{i,c}:=\{I^*_e\mid \text{edge $e$ has color $c$ and has an endpoint in $V_i$}\}$.
\end{itemize}
Notice that if $u_{i_a}u_{i_b}$ is an edge of $H$ having color $c$, 
then $\I^*_{i_a,c}=\I^*_{i_b,c}$ (as every edge with color $c$ and having an endpoint in $V_{i_a}$ has its other endpoint in $V_{i_b}$). Looking at the definitions, we
can observe that
\begin{itemize}
\item for $1\le i \le k$, point $4ni+0.5$ on the first track is covered only by members of $\I_{i,1}$,
\item for $1\le i \le k/2$, point $4ni+4n-0.5$ on the second track is covered only by members of $\I_{i,2}$,
\item for $k/2+1\le i \le k$, point $4ni+4n-0.5$ on the first track is covered only by members of $\I_{i,2}$,
\item for $1\le i \le k/2$, $0\le c \le 2$, point $4ni+(c+1)n+0.5$ on the first track is covered only by members of $I^*_{i,c}$, and
\item for $k/2+1\le i \le k$, $0\le c \le 2$, point $4ni+(c+1)n+0.5$ on the second track is covered only by members of $I^*_{i,c}$, and
\end{itemize}
That is, at least one member has to be selected from each of the $5k$ sets $\I_{i,1}$, $\I_{i,2}$, $\I^*_{i,c}$. Note that every member of $\I$ is contained either in exactly one of
these sets (in $\I_{i,1}$ or $\I_{i,2}$ for some $1\le i \le k$) or in
exactly two of these sets (in $\I^*_{i_a,c}$ and $\I^*_{i_b,c}$ for some
$1\le i_a \le k/2$, $k/2+1\le i_b \le k$, $c\in\{0,1,2\}$). Therefore,
the only way to cover the $5k$ points enumerated above by selecting
$k'=3.5k$ members of $\I$ is to select exactly one member of each
$\I_{i,1}$, $\I_{i,2}$, and $\I^*_{i,c}$.

Suppose that for some $1\le i \le k/2$, the solution $\I'$ contains
the unique members $I_{i,j}\in \I_{i,1}$, $I_{i,j'}\in \I_{i,2}$, and
$I^*_{e^c}\in \I^*_{i,c}$ for $c\in \{0,1,2\}$.  By definition of
$\I^*_{i,c}$, edge $e^c$ has an endpoint $v_{i,j^c}\in V_i$. We
claim that $j=j'=j^0=j^1=j^2$. Recall that $I_{i,j,1}$ covers
$[4ni,4ni+n-j]$ on the second track, while $I_{i,j',2}$ covers
$[4ni+n-j',4ni+4n]$ on the second track. Thus if $j>j'$, then point
$4ni+n-j'-0.5$ would not be covered on the second track; hence we have $j'\ge
j$. Observe next that $I_{i,j,1}$ covers $[4ni,4ni+j]$ and $I^*_{e^0}$ covers
$[4ni+j^0,4ni+n+j^0]$ on the first track. Thus we infer that
$j\ge j^0$: otherwise, point $4ni+j^0-0.5$ on the first track would not be
covered. In a similar way, we can show that $j^1\ge j^2$ and $j^2\ge
j'$ also hold. Therefore, we get the chain of inequalities
\[
j\ge j^0 \ge j^1 \ge j^2 \ge j' \ge j,
\]
implying that we have equalities throughout, as claimed. 

In the previous paragraph, we have shown that 
for every $1\le i \le k/2$, there is a $1\le s_i \le n$ such that $\I'$ contains
  the unique members $I_{i,s_i}\in \I_{i,1}$, $I_{i,s_i}\in \I_{i,2}$,
  and $I^*_{e^c}\in \I^*_{i,c}$ for $c\in \{0,1,2\}$ with $v_{i,s_i}$ being an endpoint of $e^c$.
In a similar way, we can prove the same statement for every $k/2+1\le i \le k$.
We claim that $v_{1,s_1}$, $\dots$, $v_{k,s_k}$ form a
  solution. Suppose that $u_{i_a}u_{i_b}$ for some $1\le i_a \le k/2$,
  $k/2+1\le i_b \le k$ is an edge of $H$, having color $c\in
  \{0,1,2\}$. We have to show that $v_{i_a,s_{i_a}}$ and
  $v_{i_b,s_{i_b}}$ are adjacent in $G$. Recall that $\I^*_{i_a,c}=\I^*_{i_b,c}$ and consider the unique member $I^*_e\in \I^*_{i_a,c}$ selected into $\I'$. As shown above, this means that edge $e$ of $G$ has an endpoint $v_{i_a,s_{i_a}}\in V_{i_a}$ and also an endpoint $v_{i_b,s_{i_b}}\in V_{i_b}$, what we had to show.
\end{proof}

We are now able to prove Theorem~\ref{thm:intro_recthard2} by a simple
and transparent reduction from covering points by 2-track
intervals. The basic idea is that if we have two parallel lines, then
for any pair of segments on these two lines, there is a rectangle that
intersects these lines exactly at these two segments. Therefore, by
arranging the points of a 2-track point set on two parallel lines, we
can easily turn the problem of covering by 2-track intervals into a
problem of covering by rectangles. Chan and
Grant~\cite{DBLP:journals/comgeo/ChanG14} used the same reduction to
establish the APX-hardness of the problem (by reducing a certain
special case of covering by 2-track intervals). Thus once we have
Theorem~\ref{th:2track}, our lower bound for covering by 2-track
intervals, the reduction of Chan and
Grant~\cite{DBLP:journals/comgeo/ChanG14} implies
Theorem~\ref{thm:intro_recthard2}.

\restateintrorecthardb*
\begin{proof}
  The proof is by reduction from the problem of covering points by
  2-track intervals. Let $\I$ be a set of 2-track intervals, let
  $\pointfam=(\pointfam^1,\pointfam^2)$ be a 2-track point set, and
  let $k$ be an integer. By rescaling if necessary, we may assume without loss of generality that
  every point in $\pointfam^1$ and $\pointfam^2$ is in the interval
  $[0,1]$ and hence every $(I^1,I^2)\in \I$ satisfies
  $I^1,I^2\subseteq [0,1]$. As shown in Figure~\ref{fig:2track-almost}, we arrange the points of $\pointfam^1$
  (resp., $\pointfam^2$) on the line $y=x+1$ (resp., $y=x-1$). That
  is, we construct a 2-dimensional point set $\pointfam^*$ that
  contains
\begin{itemize}
\item point $(\epsilon_0 s,1+\epsilon_0 s)$ for every $s\in \pointfam^1$ and 
\item point $(1+\epsilon_0 s,\epsilon_0 s)$ for every $s\in \pointfam^2$.
\end{itemize}
\begin{figure}
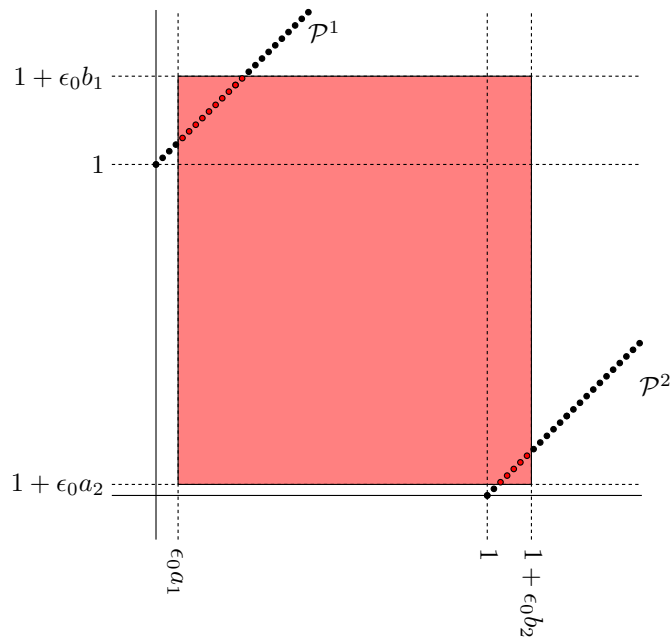

\begin{center}
{\small \svg{0.6\linewidth}{2track-almost}}
\end{center}
\caption{Proof of Theorem~\ref{thm:intro_recthard2}.
The rectangle  $[\epsilon_0 a_1,1+\epsilon_0 b_2]\times [\epsilon_0 a_2,1+\epsilon_0 b_1]$ covers exactly those points that correspond to points in $[a_1,b_1]$ on the first track and or to points in $[a_2,b_2]$ on the second track.
}\label{fig:2track-almost}
\end{figure}

Then we represent every 2-track interval
$([a_1,b_1],[a_2,b_2])\in \I$ by a rectangle whose horizontal
projection is $[\epsilon_0 a_1,1+\epsilon_0 b_2]$ and vertical
projection is $[\epsilon_0 a_2,1+\epsilon_0 b_1]$; note that in this proof we assume that the rectangles are closed, but to obtain the same result for open rectangles we can simply make them slightly larger. Observe that the
sizes of both projections are between $1-\epsilon_0$ and
$1+\epsilon_0$. Moreover, these rectangles faithfully represent the way the 2-track intervals cover the points in $(\pointfam^1,\pointfam^2)$.
For example, 2-track interval $([a_1,b_1],[a_2,b_2])$ covers point $s$ on the first track if and only if $a_1\le s \le b_1$ and the corresponding rectangle covers point $(\epsilon_0 s,1+\epsilon_0 s)$ exactly under the same condition (if $s<a_1$, then the point is to the left of the rectangle; if $s>b_1$, then the point is above the rectangle). Therefore, we can cover the point set $\pointfam^*$ with $k$ of the constructed rectangles if and only if we can cover the 2-track point set $\pointfam$ with $k$ of the 2-track intervals from $\I$.
\end{proof}

%%% Local Variables: 
%%% mode: latex
%%% TeX-master: "voronoi"
%%% End: 

%\section{Conclusions}\label{sec:conc}
%\input{conc}

\bibliographystyle{abbrv}
\bibliography{voronoi}
\fi

\end{document}

%%% Local Variables:
%%% mode: latex
%%% TeX-master: "voronoi"
%%% End: